\documentclass[12pt]{article}
\usepackage{amsmath, amssymb, amsthm}
\usepackage{xcolor}
\usepackage{mathtools}
\usepackage{hyperref}
\usepackage{graphicx}
\usepackage{enumerate}
\usepackage{caption}
\usepackage{bibunits}
\usepackage{minitoc}
\usepackage{subcaption}
\usepackage{placeins}

\usepackage{tocloft}\makeatletter
\renewcommand{\numberline}[1]{\@cftbsnum #1\@cftasnum~\@cftasnumb }

\usepackage{natbib}
\usepackage{cleveref}
\crefname{appendix}{Supplement}{Supplements}

\usepackage{macros}

\usepackage{etoolbox}
\newtoggle{showdeleted}
\togglefalse{showdeleted}

\iftoggle{showdeleted}{
  \newcommand{\deleted}[1]{{\color{gray} #1}}
  
}{
  \newcommand{\deleted}[1]{}
  
}

\newcommand{\blind}{1}

\usepackage{IEEEtrantools}
\usepackage{algorithm,algcompatible,amsmath}

\usepackage[utf8]{inputenc}

\begin{document}

\doparttoc \faketableofcontents 

\def\spacingset#1{\renewcommand{\baselinestretch}{#1}\small\normalsize} \spacingset{1}

\if1\blind
{
  \title{\bf Robustifying Likelihoods by Optimistically Re-weighting Data}
\author{Miheer Dewaskar\thanks{
    M.D. and C.T. contributed equally. }\hspace{.2cm}\\
    Department of Mathematics and Statistics,\\ University of New Mexico, Albuquerque,  NM 87106\\
    and \\
    Christopher Tosh\footnotemark[1]  \\
    Department of Epidemiology and Biostatistics, \\
    Memorial Sloan Kettering Cancer Center, New York, NY 10065\\
    and\\
    Jeremias Knoblauch\\
    Department of Statistics, UCL, London WC1E7HB, UK
\\
    and\\
    David B. Dunson\\
    Department of Statistical Science, Duke University, Durham,  NC 27705}
  \maketitle
} \fi

\if0\blind
{
  \bigskip
  \bigskip
  \bigskip
  \begin{center}
    {\LARGE\bf Robustifying Likelihoods by Optimistically Re-weighting Data}
\end{center}
  \medskip
} \fi

\begin{bibunit}[plainnat]
\begin{abstract}
Likelihood-based inferences have been remarkably successful in wide-spanning application areas. However, even after due diligence in selecting a good model for the data at hand, there is inevitably some amount of model misspecification: outliers, data contamination or inappropriate parametric assumptions such as Gaussianity mean that most models are at best rough approximations of reality.
A significant practical concern is that for certain  inferences, even small amounts of model misspecification may have a substantial impact; a problem we refer to as {\em brittleness}.
This article attempts to address the brittleness problem in likelihood-based inferences by choosing the most model friendly data generating process in a distance-based neighborhood of the empirical measure. This leads to a new Optimistically Weighted Likelihood (OWL), which robustifies the original likelihood by formally accounting for a small amount of model misspecification. Focusing on total variation (TV) neighborhoods, we
study theoretical properties, develop estimation algorithms and illustrate the methodology in applications to mixture models and regression.
\end{abstract}

\noindent {\it Keywords:}
Coarsened Bayes; Data contamination;  Mixture models; Model misspecification; Outliers;  Robust inference; Total variation distance. 
\vfill

\newpage
\spacingset{1.9}

\section{Introduction}

When the likelihood is correctly specified, there is arguably no substitute for likelihood-based statistical inferences \cite[see e.g.][]{zellner1988optimal,royall2017statistical}. 
However, if the likelihood is misspecified, inferences may be flawed \citep{huber1964robust, tsou1995robust, huber2009, Hampel2005,miller2018robust}.
This has motivated methods for model comparison and goodness-of-fit assessment \cite[see e.g.][]{huber2012goodness, claeskens2015model}.
A key is to verify that the assumed likelihood is consistent with the data at hand.
Yet, even with substantial care in model assessment, some amount of model misspecification is inevitable.

Unfortunately, even slight model misspecification can have dire consequences in certain settings; a problem we refer to as {\em brittleness}.
Brittleness can occur in various applications, including in high dimensional problems \cite[e.g.][]{bradic2011penalized, bradic2016robustness,zhou2020detangling}, in the presence of outliers and contaminating distributions \cite[e.g.][]{huber1964robust,huber2009, Hampel2005}, or for mixture models \cite[e.g.][]{markatou2000mixture, chandra2023escaping, liu2023robustly, cai2021finite}.
This article is focused on robustifying likelihoods to avert such brittleness.

\Cref{fig:simple-failure} illustrates brittleness and our proposed solution in the setting of model-based clustering with kernel mixture models.
Here, most of the data are perfectly modeled by a mixture of two well-separated Gaussians.
However, a small fraction of the data have been corrupted, and are instead drawn uniformly from an interval between the two modes.
As the left panel demonstrates, the maximum likelihood estimate (MLE) 
is sensitive to the corrupted data.
\new{Our proposal is to instead maximize an optimistically weighted likelihood (OWL), which slightly perturbs the data to be more consistent with the assumed model. The left panel 
shows that maximizing the OWL reduces  brittleness, while the right panel shows the corresponding \emph{optimistically re-weighted data.}}

\begin{figure}
	\centering
	\includegraphics[width=0.95\textwidth]{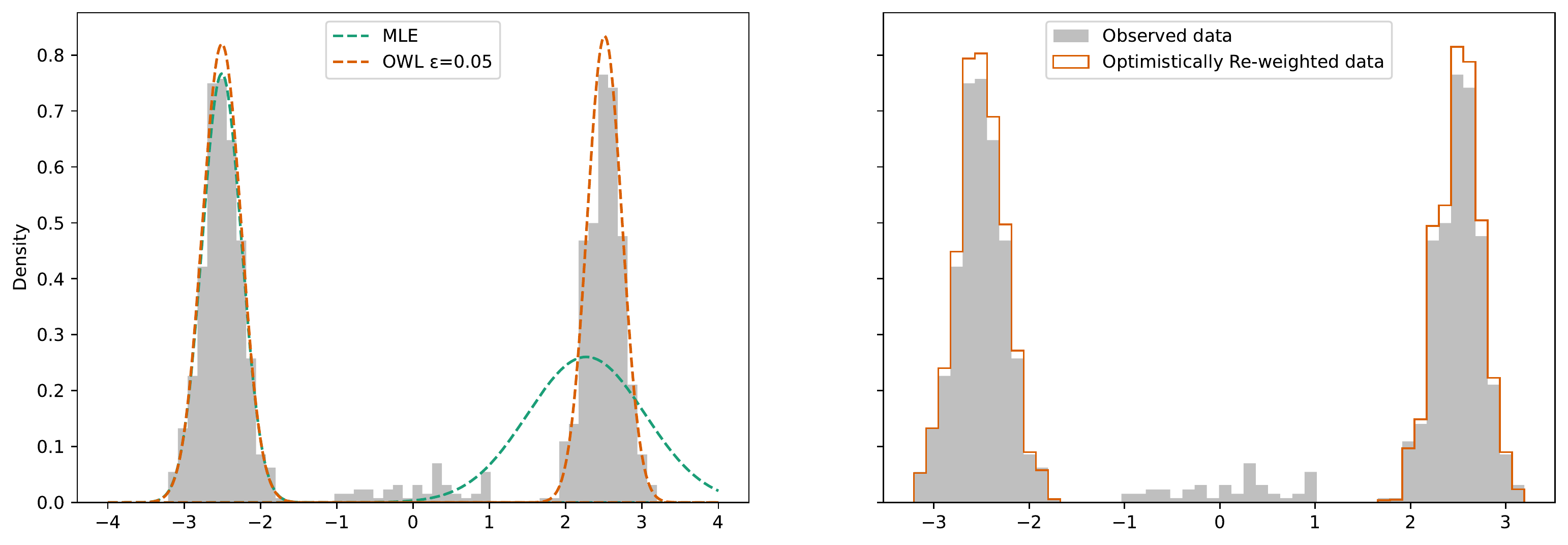}
	\caption{An example illustrating failures of  maximum likelihood estimation (MLE) for misspecified models.
		The data consist \new{of $n=1000$} observations randomly generated from an equally-weighted mixture of two Gaussians with means $-2.5$ and $2.5$, and with standard deviation $1/4$.
$5\%$ of the data were corrupted and drawn i.i.d uniform$(-1,1)$.
\textbf{Left}: The \new{green dashed  line} denotes the MLE found using expectation maximization. The \new{orange dashed line} denotes the solution found by maximizing our optimistically weighted likelihood based on $\varepsilon_0=0.05$.
\textbf{Right}: \new{Depicts the \emph{optimistically re-weighted data} that were used in place of the observed data to compute the MLE.}}
	\label{fig:simple-failure}
\end{figure}

\new{Although in principle one can address the issue of misspecification by designing more flexible models, e.g.~via mixture or nonparametric models \citep{taplin1993robust, scarpa2009bayesian, beath2018mixture}, such flexibility often comes at the expense of interpretability, identifiability, and efficiency, both computational and statistical. Moreover, there is always the risk that one's model is not flexible enough and is therefore subject to the brittleness concerns outlined above. These issues motivate the core approach of the field of robust estimation: rather than designing a complicated model that fits all aspects of the data generating process, one instead fits a model that is only approximately correct, albeit with a different methodology. At the heart of many robust estimation methods is the concept of the \emph{weighted likelihood}.}

For a likelihood function $f$, the weighted likelihood replaces $L(x_{1:n}|\theta) =
\prod_{i=1}^n f(x_i|\theta)$ by the weighted counterpart
\begin{IEEEeqnarray}{rCl}
	L_\bw(x_{1:n}|\theta) & = &
	\prod_{i=1}^n f(x_i|\theta)^{w_i}
	\label{eq:wlikelihood}
\end{IEEEeqnarray}
for a collection of weights $\bw = \{w_i\}_{i=1}^n$ that typically depend on $\theta$ and the data.
\new{Estimation of $\theta$ and $\bw$ often proceeds via \emph{iteratively reweighted maximum likelihood estimation (IRW-MLE)} by repeating the following steps until convergence: a) maximize likelihood $\theta \mapsto L_{\bw}(x_{1:n}|\theta)$  based on the current $\bw$ and, b) update $\bw$ using the current $\theta$ \citep{avella2015robust,maronna2019robust}. A variety of classical robust methods  \citep{huber2009,Hampel2005} fall in this framework, providing weights measuring the influence of each observation \citep[e.g., Figure 8 in][]{mancini2005optimal}. Typical practice
assumes $w_i = w(x_i, \theta)$ with $w(\cdot)$ chosen so that $w_i \approx 0$ for observations having very low likelihood under the presumed model; e.g., $w(x,\theta)$ is a function of the density $f(x|\theta)$ \citep{windham1995} or its cumulative distribution function \citep{field1994, dupuis2002}. See \Cref{sec:more-classical-robustness} for more discussion about these methods.}

\new{Typical robust estimation methods focus on sensitivity to outliers, while our interest is in broader types of  misspecification. 
In \Cref{fig:model-selection-example} from \Cref{sec:model-selection-illustration}, we illustrate how misspecification in the shape of the  kernel can lead to errors in selecting the number of components in mixture models. We aim to develop a general weighted likelihood methodology that allows for model misspecification beyond outliers. We focus on robustness to a small amount of misspecification with respect to an appropriate metric $\D$ in the space of probability distributions.}

\new{
A well-studied approach to achieving robustness is \emph{minimum distance estimation} \citep{MinimumDistanceMethod,william1981minimum}, which attempts to find a model that minimizes the distance to the generating distribution with respect to the metric $\D$. While this method can be shown to be robust to misspecification in $\D$ \citep{donoho1988automatic}, it is often challenging to compute \citep[e.g.][]{yatracos1985rates} beyond suitable $\phi$-divergences  \citep{lindsay1994efficiency}. In continuous spaces, when $\D$ is a $\phi$-divergence with smooth $\phi$, minimum distance estimation can be computed using weighted likelihood equations with $w(x,\theta) = \frac{A(\delta(x|\theta)) -  A(-1)}{\delta(x|\theta) + 1}$ where $A$ is the residual adjustment function corresponding to $\phi$ and $\delta(x|\theta) = \hat{p}(x)/f(x|\theta) - 1$ is the Pearson residual based on a density estimate $\hat{p}$ of the observed data \citep{basu1994minimum,basu1998,markatou1998,greco2020weighted}. Unfortunately, most $\phi$-divergences outside the Hellinger and total-variation (TV) cases are not metrics, making interpretation difficult. Moreover, the important case of TV distance is excluded by this work since the corresponding $\phi$ is not smooth.  Recent work \citep{Cherief-Abdellatif2019,briol2019statistical,bernton2019parameter} has provided tools for computing minimum distance estimators when $\D$ is the maximum mean discrepancy (MMD) or Wasserstein distance, however this approach relies on generative models and does not make use of the likelihood.}

\new{
Here, 
we develop a weighted likelihood methodology that can perform \emph{approximate} minimum distance estimation, providing a recipe for robustifying  Bayesian and frequentist approaches by re-weighting the likelihood (see \Cref{sec:model-selection-illustration}). For non-negative weights $\bw = \{w_i\}_{i=1}^n$ that sum to $n$, our key observation is that the weighted likelihood \eqref{eq:wlikelihood} is an ordinary (i.e.~un-weighted) likelihood for perturbed data represented by the weighted empirical distribution  $\hat{Q}_{\bw} = n^{-1}\sum_{i=1}^n w_i \delta_{x_i}$, where $\delta_{x}$ denotes the Dirac measure at $x$. Thus IRW-MLE simultaneously finds: i) an  optimal data perturbation $\hat{Q}_{\bw^*}$, and ii)  the ordinary maximum likelihood estimate $\hat{\theta}$ based on the perturbed data $\hat{Q}_{\bw^*}$. For approximate minimum distance estimation, we require the perturbations $\hat{Q}_{\bw}$ to lie in a $(\D,\epsilon)$ neighborhood around the empirical distribution of the observed data $\hat{P} = n^{-1}\sum_{i=1}^n \delta_{x_i}$ for a choice of $\epsilon > 0$ given by the user.}

\new{
Our optimal weights $\bw^*$ are determined using the principle of \emph{optimism}: in the absence of  information other than constraint $\D(\hat{Q}_{\bw}, \hat{P}) \leq \epsilon$,  we choose the perturbation $\hat{Q}_{\bw}$ closest to the model family in the Kullback Leibler (KL) sense. If $\{P_\theta\}_{\theta \in \Theta}$ denotes our model family and  $\hat{\KL}$ is an estimator for  the KL divergence, we (\Cref{alg:owl}) adapt the IRW-MLE procedure to solve the following \emph{Optimistic Kullback Leibler} (OKL) minimization:}
\begin{equation}
\label{eq:okl-intuitive}
\min_{\theta \in \Theta} \min_{\bw: \D(\hat{Q}_\bw, \hat{P}) \leq \epsilon} \hat{\KL}(\hat{Q}_{\bw}|P_\theta).
\end{equation}

\new{When $\epsilon = 0$, OKL minimization is equivalent to the usual MLE,
and the choice $\epsilon > 0$ reflects the degree of misspecification. If $P_0$ is the true distribution generating $x_1, \ldots, x_n$, and $\epsilon$ is greater than the degree of misspecification $\varepsilon_0 \doteq \min_{\theta \in \Theta} \D(P_0, P_\theta)$, then 
any global minimizer $\hat{\theta}$ of \eqref{eq:okl-intuitive} can be asymptotically shown to lie in the set $\{\theta : \D(P_0, P_\theta) \leq \epsilon\}$ (\Cref{prop:owl-consistency}). Thus when $\epsilon$ is tuned to approximate $\varepsilon_0$ (\Cref{sec:tune-epsilon}), $\hat{\theta}$ will be a minimum distance estimator. }

\new{In this work, we show that the OKL is intimately related to the concept of the \emph{coarsened likelihood} \citep{miller2018robust} -- a genuine  likelihood that acknowledges small misspecification in terms of $\D$. While even evaluating the coarsened likelihood requires the computation of a high-dimensional integral, we use techniques from large deviation theory \citep{demboLargeDeviationsTechniques2010} to show that the coarsened likelihood is asymptotically equivalent to the OKL. Thus, the OWL methodology can be seen as a bridge, connecting minimum distance estimation and weighted likelihood methods on one side and the coarsened likelihood approach on the other.}

\new{
When $\D(Q,P) = \int \log \frac{dQ}{dP} dQ$ is the KL divergence, 
the Lagrangian formulation of our OKL minimization problem corresponds to minimizing a Cressie-Read power-divergence between the observations and the model family \citep{ghosh2018new}. Since different computational strategies may be required across metrics, we focus on  total variation (TV) distance for $\D$, which is easy to interpret, is a proper metric, and is robust to outliers and contamination \citep{donoho1988automatic}. Further, TV based estimators have been shown to attain minimax estimation rates under Huber’s contamination model \citep{chen2016general}.}

\new{
Coarsened likelihoods \citep{miller2018robust} in general, and the principle of optimism used here in particular, are related to statistical learning from fuzzy or imprecise data \citep{hullermeier2015superset,hullermeier2014learning}; the key idea is to treat observed data as imprecise to facilitate robust learning \citep{lienen2021instance}. 
Among various approaches to robust inference in regression \cite[e.g.][]{she2011outlier,rousseeuw2005robust,avella2018robust},  \cite{bondell2013efficient} implements weighted least squares with weights  learned within an empirical likelihood framework subject to a target weighted squared residual error.  While OKL minimization  \eqref{eq:okl-intuitive} finds weights 
in a fashion similar to this and other empirical likelihood approaches \citep[e.g.][]{choi2000rendering}, our optimization problem arises from large deviation formulas and leads to weights that are constrained to lie within a small total-variation ($\ell_1$) distance of the uniform weight vector.}

The remainder of this paper is structured as follows:
\Cref{sec:methodology} presents the OWL methodology, the associated alternating optimization scheme, \new{and discusses some formal robustness guarantees}.
\Cref{sec:coarsened-inference} demonstrates the asymptotic connection between OWL and coarsened inference.
\Cref{sec:simul} presents a suite of simulation experiments for the OWL methodology in both regression and clustering tasks.
 \Cref{sec:micro-credit} uses OWL to infer the average intent-to-treat effect in a micro-credit study \citep{angelucci2015microcredit}, whose inference using ordinary least squares (OLS) was shown to be brittle to the removal of an handful of observations \citep{broderick2020automatic}. Finally, we provide a robust clustering application for single-cell RNAseq data in \Cref{sec:application}.  
\section{Optimistically Weighted Likelihoods}
\label{sec:methodology}

As \Cref{fig:simple-failure} shows, maximum likelihood estimation can be brittle when the data generating distribution $P_0$ \new{lies outside, albeit suitably close, to the  model family $\{P_\theta\}_{\theta \in \Theta}$.}  
To assuage the problem of brittleness under misspecification, we propose an Optimistically Weighted Likelihood (OWL) approach that iterates between (1) an optimistic (or model-aligned) re-weighting of the observed data points and (2) updating the parameter estimate by maximizing a weighted likelihood based on the current data weights.

In \Cref{sec:okl}, we study this parameter inference methodology at the population level where \new{we formally allow $P_0$ to be misspecified as long as it is close to the model family in the total-variation (TV) distance.} Here we introduce the population level Optimistic Kullback Leibler (OKL) function with parameter $\epsilon \in [0,1]$ (\Cref{def:okl}), and show that its minimizer will be a parameter $\theta$ for which $P_\theta$ is $\epsilon$-close to $P_0$ in TV distance. \new{We provide formal robustness guarantees for these minimizers by analyzing a  bias-distortion curve}.

Motivated by this population analysis, in \Cref{sec:owl}, we derive the OWL-based parameter estimation methodology when only samples $x_1, \ldots, x_n$ from $P_0$ are available.

\subsection{Population level Optimistic Kullback Leibler minimization}
\label{sec:okl}

\begin{figure}
	\includegraphics{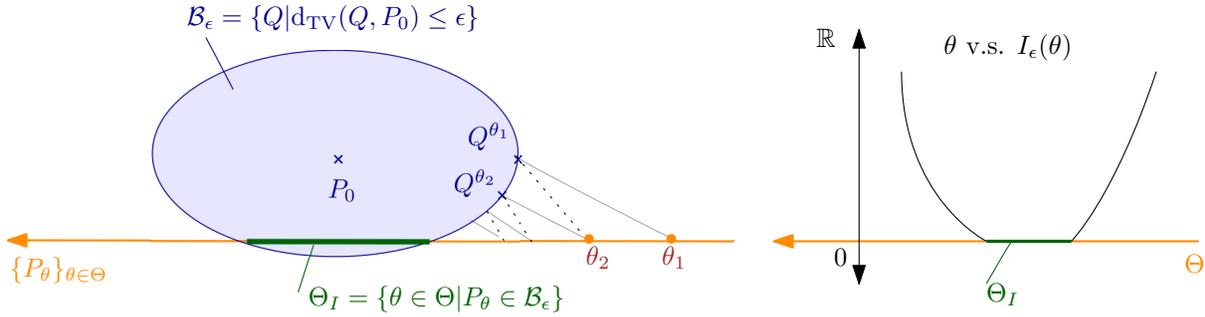}
	\caption{Population level description of OWL (left) and the OKL function (right). The point $P_\theta$ is labeled using $\theta \in \Theta$. The inference  problem is to find a point in $\Theta_I \subseteq \Theta$ where the model family intersects the  $\epsilon$-neighborhood $\ball_\epsilon$ of the data distribution $P_0$. The set $\Theta_I$ is the set of minimizers of the OKL function $\theta \mapsto I_\epsilon(\theta)$ (right). Starting from initial point $\theta_1 \in \Theta$, the OWL procedure finds a saddle point of the OKL function by iterating $\theta_{t+1} = \argmin_{\theta \in \Theta} \KL(Q^{\theta_t}|P_\theta)$ for $t=1,2, \ldots,$ until convergence, where $Q^\theta = \argmin_{Q \in \ball_\epsilon}  \KL(Q|P_\theta)$ denotes the information projection (\cite{amari2016information}) of the model $P_\theta$ on the total variation (TV) neighborhood $\ball_\epsilon$ of $P_0$. Iterations alternate between I-projection and weighted likelihood estimation steps, illustrated via solid and dashed lines.} 
\label{fig:okl}
\end{figure}

Let $\cP(\cX)$ denote the space of probability distributions on the data space $\cX$. In principle, our methodology can accommodate misspecification in terms of a variety of probability metrics (e.g.~MMD or Wasserstein) on $\cP(\cX)$, but here we mainly focus on the total variation (TV) distance for concreteness and interpretability. Let $\tv(P,Q) = \sup_{A \subseteq \cX} |P(A)-Q(A)|$ denote the TV metric between two probability distributions $P, Q \in \cP(\cX)$, \new{where the supremum ranges over all measurable subsets $A$ of $\cX$}. Given a model family $\{P_\theta\}_{\theta \in \Theta} \subseteq \cP(\cX)$, we assume that the data generating distribution $P_0 \in \cP(\cX)$ for the data population in question satisfies $\tv(P_0, P_{\theta^*}) \leq \epsilon$ for a known value $\epsilon \geq 0$ and some unknown $\theta^* \in \Theta$. In other words, we make the following assumption (satisfied under Huber's $\epsilon$-contamination model~\citep{huber1964robust} given by $P_0 = (1-\epsilon) P_{\theta^*} + \epsilon C$ for any $C \in \cP(\cX)$).

\begin{assume}
	\label{ass:misspecification}
Given $\epsilon \geq 0$ and the true data distribution $P_0$, the set of parameters
$\Theta_I = \{ \theta \in \Theta | \tv(P_0, P_\theta) \leq \epsilon\}$ is non-empty.
\end{assume}

\new{
In terms of \cite{liu2009building}, \Cref{ass:misspecification} requires that our model be \emph{adequate for the data} at level $\epsilon$ in terms of the TV neighborhood. While \cite{liu2009building} develop a likelihood-ratio test for this assumption, here we focus on estimating parameters by maximizing a natural likelihood (\Cref{sec:coarsened-inference}), leading to a different optimization problem.}

In general, under \Cref{ass:misspecification}, it may only be possible to identify the set $\Theta_I$, rather than any particular $\theta^* \in \Theta$. 
Although such indeterminacy may be inherent, it is practically insignificant whenever $\epsilon$ is sufficiently small so that the distinction between two elements from $\Theta_I$ is irrelevant~\citep{huber1964robust}. In line with this insight, the goal throughout the rest of the paper will be to identify \emph{some} parameter in $\Theta_I$. \new{Indeed, selecting $\epsilon$ to approximate $\min_{\theta} \tv(P_0,P_\theta)$ (see \Cref{sec:tune-epsilon}) will recover a minimizer of the map $\theta \mapsto \tv(P_0,P_\theta)$.}

At the population level, maximum likelihood parameter estimation amounts to minimizing the Kullback Leibler (KL) function $\theta \mapsto \KL(P_0|P_{\theta})$ on the parameter space $\Theta$. Even under small amounts of misspecification, KL minimizers are very brittle since any minimizer of the KL function must place sufficient probability mass wherever $P_0$ does, including on outliers.
In contrast, TV distance is far less sensitive to \new{outliers and slight shifts in the shape of the distribution.} 
Hence one may minimize $\theta \mapsto \tv(P_0, P_\theta)$ as a robust alternative, particularly under  \Cref{ass:misspecification}. However, direct minimization of TV distance over the parameter space $\Theta$ is difficult to implement in practice due to the lack of suitable optimization primitives (e.g.~maximum likelihood estimators) and the  non-convex and non-smooth nature of the optimization problem \citep[see e.g.][]{yatracos1985rates}. 

Instead, we modify the KL loss by minimizing its first argument over the $\epsilon$-neighborhood of $P_0$ in the TV distance.
The resulting function, which we term the Optimistic Kullback Leibler (OKL), is defined as follows.

\begin{definition}(Optimistic Kullback Leibler) Given $P_0$ and $\epsilon > 0$, the OKL function $I_\epsilon: \Theta \to [0,\infty]$  is defined as:
\begin{equation}
	\label{eq:okl}
    I_\epsilon(\theta) = \inf_{Q \in \ball_\epsilon(P_0)} \KL(Q|P_\theta),
\end{equation}
where $\ball_\epsilon(P_0) = \{Q \in \cP(\cX) | \tv(P_0,Q) \leq \epsilon\}$ is the TV ball of radius $\epsilon$ around $P_0$.
If $I_\epsilon(\theta) < \infty$, the underlying optimization over $\ball_\epsilon(P_0)$ has a unique minimizer $Q^\theta$ called the I-projection \citep{csiszar1975divergence}.
\label{def:okl}
\end{definition}

The OKL function $\theta \mapsto I_\epsilon(\theta)$ measures the fit of a model $P_\theta$ to the data $P_0$, allowing for a degree $\varepsilon$ of  data re-interpretation in the TV distance before assessing model fit. More precisely, since $I_\epsilon(\theta) = \KL(Q^\theta|P_\theta)$, the OKL is equal to the KL divergence between the \emph{optimistic} (i.e.~most model aligned) distribution $Q^\theta \in \ball_\epsilon(P_0)$ in the neighborhood of $P_0$ and the model $P_\theta$.
Here,  $\epsilon \geq 0$ regulates the permitted degree of re-interpreting the data by controlling the neighborhood size.

\new{
The OKL function can be used to find a parameter from the set $\Theta_I$.
Indeed, as illustrated in \Cref{fig:okl} (right), the minimum value of zero for the OKL function is attained exactly on the set $\Theta_I$ since $Q^\theta = P_\theta$ if and only if $\theta \in \Theta_I$.}
However, the OKL can be non-convex as a function of $\theta$, so that calculating the global minimizer of OKL may not be straightforward. 
\new{Fortunately, the OKL lends itself to a feasible alternating optimization scheme (\Cref{fig:okl}; left) that is guaranteed to decrease the OKL objective or reach a saddle point under regularity conditions.}

\deleted{
Minimizing the OKL function is equivalent to solving the joint minimization problem $\min_{\theta \in \Theta} \min_{Q \in \ball_\epsilon(P_0)} \KL(Q|P_\theta)$. For this, we can use the following alternating optimization strategy: start from a $\theta_1 \in \Theta$ such that $I_\epsilon(\theta_1)  < \infty$, and perform the $Q$-step: $Q^{\theta_{t}} = \argmin_{Q \in \ball_\epsilon(P_0)} \KL(Q|P_{\theta_t})$ and the $\theta$-step: $\theta_{t+1} = \argmin_{\theta \in \Theta} \KL(Q^{\theta_t}|P_\theta)$ for $t \in \{1,2,
\ldots \}$. If the model family $\{P_\theta\}_{\theta \in \Theta}$ and $P_0$ have densities $\{p_\theta\}_{\theta \in \Theta}$ and $p_0$ with respect to a common measure $\lambda$, then the I-projection $Q^{\theta_t}$ will also have a density $q_t \in \Den = \{q : \cX \to [0,\infty) \mid \int q(x) d\lambda(x) = 1\}$ with respect to $\lambda$, and the iterations will take the following form:
\begin{enumerate}
	\item \underline{Q-step}: $q_t = \argmin_{q} \int q(x) \log \frac{q(x)}{p_{\theta_t}(x)} d\lambda(x)$, where the minimization is over all $q \in \Den$ that satisfy the total-variation constraint $\frac{1}{2} \int |q-p_0| d\lambda  \leq \epsilon$.

	\item \underline{$\theta$-step}: $
	\theta_{t+1} = \argmax_{\theta \in \Theta} \int q_t(x) \log p_\theta(x) d\lambda(x).
	$
\end{enumerate}

The above iterations provide a scheme to minimize the OKL function in which the $Q$-step can  approximately be performed using tools from convex optimization, while the $\theta$-step can be approximated by maximization of a suitably weighted log-likelihood, which can be computed for many standard models. 
Our resulting population level methodology is illustrated in \Cref{fig:okl}.}

\deleted{We remark that one can use optimization of the OKL function $I_\epsilon(\theta)$ as a subroutine to minimize the function $\theta \mapsto \tv(P_0, P_\theta)$. Namely, we can perform binary search over $\epsilon \in [0,1]$, increasing $\epsilon$ whenever we have $I_\epsilon(\theta) > 0$ and decreasing $\epsilon$ whenever we have $I_\epsilon(\theta) = 0$. Used this way, OKL optimization can be seen as a computationally-palatable approach to minimizing the TV distance over a model class.}

\new{We now summarize formal robustness guarantees for the global minimizers of the OKL function against perturbations of $P_0$ in small TV-neighborhoods around the model family; details can be found in \Cref{sec:robustness-okl-minimizer}. We call the mapping from $P_0$ to the global minimizers of the OKL function as the \emph{OWL functional} $\OF: \cP(\cX) \to 2^{\Theta}$, defined as $\OF(P_0) \doteq \argmin_{\theta \in \Theta} I_\epsilon(\theta)$ where $2^{\Theta}$ denotes the power-set of $\Theta$. Motivated by the robustness theory for minimum distance functionals \citep{donoho1988automatic}, when $P_0 = P_{\theta^*}$ lies on the model family, we bound the growth of a bias-distortion (BD) curve $\delta \mapsto \maxbias(\delta|P_0)$ defined as $\maxbias(\delta|P_0) \doteq \sup_{\substack{P \in \cP(\cX) \\ \dtv(P, P_0) \leq \delta}} \dH(\OF(P), \OF(P_0))$, where $\dH$ denotes the Hausdorff distance on $2^{\Theta}$. Indeed, classically important robustness indicators like \emph{sensitivity} and \emph{qualitative-robustness} are properties about the growth of a suitable BD curve when $\delta$ is close to zero, and the \emph{breakdown-point} is the point at which the BD curve has a vertical asymptote \citep[Figure 1]{donoho1988automatic}. Given this setup, under suitable regularity conditions, our two main results in \Cref{sec:robustness-okl-minimizer} show that the \emph{OWL functional} is qualitatively-robust with finite sensitivity (corollary of \Cref{lem:bd-small-bounds}) and has a breakdown-point of at least $\epsilon$ whenever $\epsilon$ is less than half of the best breakdown-point for the model family among all Fisher consistent functionals (\Cref{lem:owl-breakdown}).} 

\subsection{Optimistically Weighted Likelihood (OWL) estimation}
\label{sec:owl}

Here we extend the population level methodology from \Cref{sec:okl} to handle the practical case when  samples $x_1, \ldots, x_n \iid P_0$ are available, \new{providing a computable approximation to the alternating procedure in \Cref{fig:okl}.
Particularly, the I-projection step is now approximated by a convex optimization problem over weight vectors constrained to lie within the intersection of the $n$-dimensional probability simplex and the $\ell_1$ ball of radius $2\epsilon$ around the uniform probability vector; these optimal weights can be interpreted as an optimistic re-weighting of the original data points $x_{1}, \ldots, x_{n}$ to match the current model estimate. Given these weights, a new parameter estimate is then found by maximizing the weighted likelihood.} We call this the Optimistically Weighted Likelihood (OWL) method.

\subsubsection{Approximating OKL by a finite dimensional optimization problem}

Given observed data $x_1, \ldots, x_n \iid P_0$,  we start by approximating the OKL function $I_\epsilon(\theta)$ in terms of a finite dimensional optimization problem. Henceforth, let us assume that the model family $\{P_\theta\}_{\theta \in \Theta}$ and measure $P_0$ have densities $\{p_\theta\}_{\theta \in \Theta}$ and $p_0$ with respect to a common measure $\lambda$. We will focus on two cases of interest: when $\cX$ is a \new{discrete} space and $\lambda$ is the counting measure, and when $\cX=\R^d$ and $\lambda$ is the Lebesgue measure.

When $\cX$ is \new{discrete}, we look to solve the optimization problem in \cref{eq:okl} over data re-weighting $Q = \sum_{i=1}^n w_i \delta_{x_i}$ as the weight vector $w=(w_1, \ldots, w_n)$ varies over the \ndsimplex\ $\Delta_n$ and satisfies the TV constraint $\frac{1}{2}\|w-o\|_1 \leq \epsilon$ where $o=(1/n, \ldots, 1/n) \in \Delta_n$. Formally, our finite space OKL approximation is given by

\begin{equation}
	\label{eq:finiteokle}
	\finitehatI(\theta) = \inf_{\substack{w \in {\Delta}_n \\ \frac{1}{2} \|w-o\|_1 \leq \epsilon}} \sum_{i=1}^n w_i \log \frac{nw_i \pfinite(x_i)}{p_\theta(x_i)},
\end{equation}
where $\pfinite(y) = \frac{|\{i \in [n] | x_i = y\}|}{n}$ is the histogram estimator for $p_0$ based on observed data. An application of the log sum inequality~\cite[Theorem~2.7.1]{cover2006elements} shows that the weights that solve \cref{eq:finiteokle} have the appealing and natural property that $w_i = w_j$ whenever $x_i=x_j$.
Moreover, when the support of $p_0$ \new{is finite and} contains the support of $p_\theta$, $\finitehatI(\theta)$ converges to $I_\epsilon(\theta)$ at rate $n^{-1/2}$, as demonstrated by the following result (\Cref{sec:okl-estimation-finite}). \new{A modified proof can establish the weaker result of consistency of $\finitehatI(\theta)$ even when the support of $p_0$ is (countably) infinite, provided  $\supp(p_\theta) \subseteq \supp(p_0)$. The latter condition is mild and holds whenever the model family has a fixed support (e.g. exponential family) and $p_0$ is a Huber's contamination from the model.}
\begin{theorem}
\label{thm:finite-okl-convergence-tv}
Suppose that $I_{\epsilon_0}(\theta) < \infty$ for some $\epsilon_0 > 0$ and pick $\delta > 0$ and $\epsilon > \epsilon_0$. If $\supp(p_\theta) \subseteq \supp(p_0)$, \new{$\supp(p_0)$ is finite}, and $x_1,\ldots,x_n \iid p_0$, then with probability at least $1-\delta$, 
\[ |I_\epsilon(\theta) - \finitehatI(\theta)| \leq  O\left(\frac{ |\supp(p_0)|}{\epsilon - \epsilon_0} \sqrt{\frac{1}{n} \log \frac{|\supp(p_0)|}{\delta}} \right), \]
where $\supp(p) = \{ x \in \Xcal : p(x) > 0 \}$.
\end{theorem}

When $\cX=\R^d$, the above approximation strategy needs to be modified, since $\KL(Q|P_\theta)$ is unbounded whenever $P_\theta$ is a continuous distribution and $Q = \sum_{i=1}^n w_i \delta_{x_i}$ is a discrete distribution. 
However, using a suitable kernel $\K: \cX \times \cX \to [0,\infty)$ (\new{e.g.~$\K_h(x,y) = \frac{1}{(\sqrt{2\pi} h)^d}e^{-\frac{\|x-y\|^2}{2h^2}}$)} one can establish the continuous space approximation (\Cref{sec:okl-estimation-cont})
\begin{equation}
	\label{eq:okle}
	\hatI(\theta) = \inf_{\substack{w \in \hat{\Delta}_n \\ \frac{1}{2}\|w-o\|_1 \leq \epsilon}} \sum_{i=1}^n w_i \log \frac{n w_i \hat{p}(x_i)}{p_\theta(x_i)}
\end{equation}
where $\hat{p}$ is a density estimator for $p_0$ based on $x_1,\ldots, x_n$, $A$ is an $n \times n$ matrix with entries $A_{ij} = \frac{\K(x_i,x_j)}{n \hat{p}(x_i)}$, and $\hat{\Delta}_n = A \Delta_n$ is the image of the \ndsimplex\ under linear operator $A$. We will typically take $\hat{p}(\cdot) = \frac{1}{n} \sum_{j=1}^n \K(\cdot, x_j)$ to be the kernel-density estimate based on the same kernel $\K$, in which case $A$ is the stochastic matrix obtained by normalizing the rows of the kernel matrix $K = (\K(x_i,x_j))_{i,j \in [n]}$ to sum to one. 

The continuous space approximation in \cref{eq:okle} yields the finite space approximation in  \cref{eq:finiteokle} as a special case when $\K(x, y) = \I{x = y}$ is taken to be the indicator kernel. The weights vectors in $\hat{\Delta}_n$ are always non-negative and approximately sum to one for large values of $n$, since $\sum_{i=1}^n A_{ij} = \frac{1}{n}\sum_{i=1}^n \frac{\kappa(x_i, x_j)}{\hat{p}(x_i)} \approx \int \frac{\kappa(x, x_j)}{\hat{p}(x)} p_0(x) dx \approx \int \kappa(x, x_j) dx = 1$.
	
Now we briefly describe the large sample convergence of $\hatI(\theta)$ to $\okl$ shown in \Cref{sec:okl-estimation-cont}.
Proving this exact statement is technically challenging, as it requires proving uniform-convergence of the objective in \cref{eq:okle} when the weights are allowed to vary over the entire range $\hat{\Delta}_n$. To get around this, we restrict the optimization domain to $\hat{\Delta}_n^\beta = A \Delta_n^\beta$ where $\Delta_n^\beta = \{ (v_1, \ldots, v_n) \in \Delta_n \, : \,  v_i \in [\frac{\beta}{n}, \frac{1}{n\beta}] \}$ for a suitably small constant $\beta > 0$. Thus, the estimator that we theoretically study is given by
\begin{equation}
	\label{eq:okle-theoretical}
	\hat{I}_{\epsilon,\beta}(\theta) = \inf_{\substack{w \in \hat{\Delta}_n^\beta \\ \frac{1}{2}\|w-o\|_1 \leq \epsilon}} \sum_{i=1}^n w_i \log \frac{n w_i \hat{p}(x_i)}{p_\theta(x_i)}.
\end{equation}
Given this change, we can show the following result.
\begin{theorem}
\label{thm:continuous-okle-convergence}
Suppose $\supp(p_0) = \supp(p_\theta) = \Xcal$ is a compact subset of $\R^d$ and there exists a constant $\gamma > 0$ such that $p_0(x), p_\theta(x) \in [\gamma, 1/\gamma]$ for all $x \in \Xcal$. Suppose that we use the probability kernel $\kappa(x,y) = \frac{1}{h^d} \phi(\|x - y \|/h)$ having bandwidth $h > 0$, where $\kappa$ is positive semi-definite and $\phi$ is a bounded function with exponentially decaying tails. Assume that $p_0$ and $\log p_\theta$ are $\alpha$-H\"{o}lder smooth over $\Xcal$, and suppose that we use the clipped density estimator $\hat{p}(x) = \min(\max(\frac{1}{n}\sum_{i=1}^n \kappa(x_i, x), \gamma), 1/\gamma)$. Then for any constant $0 < \beta \leq \gamma^2/4$
\[ 
\left|\hat{I}_{\epsilon, \beta}(\theta) - I_\epsilon(\theta) \right| \leq \tilde{O}\left(n^{-1/2} h^{-d} + h^{\alpha/2} + \psi(\sqrt{h}) \right) ,  \]
with probability at least $1 - 1/n$, where $\psi(r) = \frac{\lambda (\Xcal \setminus \Xcal_{-r})}{\lambda(\Xcal)}$, $\lambda(\cdot)$ is the $d$-dimensional Lebesgue measure, $\Xcal_{-r} = \{ x \in \Xcal : B(x, r) \subseteq \cX \}$, and $\tilde{O}(\cdot)$ hides constants and logarithmic factors.
\end{theorem}
Here $\psi(r)$ measures the fraction of the volume of $\Xcal$ that is contained in the envelope of width $r$ closest to the boundary. For well-behaved sets, we expect $\psi(r)$ to decrease to 0 as $r \rightarrow 0$. For example, if $\Xcal$ is a $d$-dimensional ball of radius $r_0$, then $\psi(r) = 1 - (1 - \frac{r}{r_0})^d$.

We prove our theory with the truncated estimator $\hat{I}_{\epsilon, \beta}(\theta)$ instead of $\hat{I}_{\epsilon}(\theta)$ for technical reasons. By a suitable version of the sandwiching lemma (\Cref{sec:sandwiching}), under the assumptions of \Cref{thm:continuous-okle-convergence}, we conjecture that the optimal weights in   \eqref{eq:okle} lie  in the set $\hat{\Delta}^\beta_n$ for a small enough constant $\beta > 0$, in which case we will have $\hat{I}_{\epsilon, \beta}(\theta) = \hat{I}_{\epsilon}(\theta)$.

Given the large sample convergence of $\hatI(\theta)$ to $I_\epsilon(\theta)$ for every $\theta \in \Theta$, one expects that any estimator $\hat{\theta} \in \Theta$ that is a global minimizer of $\hatI$ will approach the set $\Theta_I$ (\Cref{ass:misspecification}) as $n \to \infty$. A qualitative version of this result, proved in \Cref{sec:owl-consistency}, is shown next.
\begin{proposition} 
\label{prop:owl-consistency}
Suppose that (i) $\cX$ and $\Theta$ are compact metric spaces, and $(x, \theta) \mapsto \log p_\theta(x)$ is a continuous function from $\cX \times \Theta \to \R$. Let $\epsilon \geq 0$ be such that (ii) \Cref{ass:misspecification} is satisfied. Let $x_1, \ldots, x_n \iid p_0$, and suppose that (iii) for any fixed $\theta \in \Theta$, we have $\hatI(\theta) \pconv \okl$ as $n \to \infty$. If an estimator $\hat{\theta} \in \Theta$ satisfies $|\hatI(\hat{\theta}) - \min_{\theta \in \Theta} \hatI(\theta) | \pconv 0$, then for every $\delta > 0$
$$
	\lim_{n \to \infty} \prob( \tv(P_{\hat{\theta}}, P_0)  > \epsilon + \delta) = 0.
$$
\end{proposition}
\Cref{thm:finite-okl-convergence-tv} and \Cref{thm:continuous-okle-convergence} (with the conjectured $\beta = 0$) provide sufficient conditions for condition (iii) in \Cref{prop:owl-consistency} to hold. Condition (i) is primarily used to show that the convergence in (iii) is uniform, i.e. $\sup_{\theta \in \Theta} |\hatI(\theta) - \okl| \pconv 0$. Finally, condition (ii) is satisfied when the data  $P_0 = (1-\epsilon) P_{\theta^*} + \epsilon C$ is a contamination. \Cref{prop:owl-consistency} and triangle inequality then  bound the model estimation error: $\tv(P_{\hat{\theta}}, P_{\theta^*}) \leq 2\epsilon + o_P(1)$.

\subsubsection{OWL Methodology}

\new{Based on the computable approximation $\hatI$ (or $\finitehatI$) to the OKL function, \Cref{alg:owl} uses an alternating optimization to minimize $\hatI(\theta)$ over $\theta$. The procedure simultaneously estimates a minimizer $\hat{\theta}$ and \emph{optimistic weights} $w_1, \ldots, w_n \geq 0$ that sum to $n$.} 

In \Cref{sec:comp}, we expand further on the computational details for the $\theta$-step and $w$-step in \Cref{alg:owl}. In particular, the $w$-step is a convex optimization problem which we can solve by using the 
 consensus ADMM algorithm \citep{boyd2011distributed, parikh2014proximal} based on  proximal operators that can be efficiently computed. Further, the $\theta$-step corresponds to maximizing a weighted likelihood, which can be performed for many models through simple modifications of procedures for the corresponding maximum likelihood estimation. As shown in \Cref{sec:comp}, the latter computation is easy to perform exactly for exponential families, while for mixture models we can use a weighted variant of the `hard EM' algorithm \citep{samdani2012unified}.

In \Cref{sec:owl-and-MM-DC}, we show that the iterates $\{\theta_t\}_{t \geq 1}$ of \Cref{alg:owl} will decrease the objective value $\theta \mapsto \hat{I}_\epsilon(\theta)$ at each step, and that \Cref{alg:owl} is an instance of the well-studied  Majorization-Minimization (MM) \citep{hunter2004tutorial} and Difference of Convex (DC) \citep{le2018dc} class of algorithms. We use this to analyze the convergence of $\{\theta_t\}_{t \geq 1}$ when $\{p_\theta\}_{\theta \in \Theta}$ is an exponential family (see \Cref{rem:owl-iterates-exponential-family}).

\begin{algorithm}
	\caption{OWL Methodology}
	\label{alg:owl}
	\begin{algorithmic}
		\STATE \textbf{Input:} Model $\{p_\theta\}_{\theta \in \Theta}$, radius $\epsilon \geq 0$, kernel $\K$, initial $\theta_1 \in \Theta$, and iteration limit $T$.
		\FOR{$t=1,\ldots,T$}
		\STATE \underline{$w$-step:} Find $w_t=(w_{t,1}, \ldots, w_{t,n}) \in \Rnn^n$ that minimizes \cref{eq:okle} for  $\theta=\theta_t$. 
		\STATE \underline{$\theta$-step:} Find $\theta_{t+1}$ that maximizes the weighted likelihood $\theta \mapsto \sum_{i=1}^n w_{t,i} \log p_\theta(x_i)$.
		\ENDFOR
		\STATE \textbf{Output:} The robust parameter estimate $\theta_T$ and the data weights $\new{n w_T}$.
	\end{algorithmic}
\end{algorithm}

In practice, when the data lie in a  continuous space, we often avoid using the kernel-based estimator \cref{eq:okle} to determine the weights in the $w$-step of \Cref{alg:owl} because it greatly slows down the computation (see \Cref{sec:admm-complexity})\deleted{, and the resulting weights are sensitive to the choice of kernel $\K$}. Instead, we often perform the $w$-step by solving the \emph{unkernelized optimization} problem using  $\kappa(x, y)=\I{x=y}$. We demonstrate via simulations (\Cref{sec:simul}) that the unkernelized version of the OWL procedure has equally good performance compared to the kernelized version with a suitably tuned bandwidth. \new{The theoretical impact of this approximation can be explored by studying the limiting behavior of the OKL estimator for suitably fixed choices of $\K$ as $n \to \infty$.} 

\subsubsection{Setting the corruption fraction $\epsilon$}
\label{sec:tune-epsilon}

So far, we have assumed that the parameter $\epsilon \in (0,1)$, which can be interpreted as the fraction of corrupted samples in the population distribution, is fixed at a known value that satisfies \Cref{ass:misspecification}. Now let us see how the population level analysis (Section \ref{sec:okl}) can inform our choice of $\epsilon$.  \Cref{ass:misspecification} is satisfied as long as $\epsilon \geq \varepsilon_0$, where
$$
\varepsilon_0 = \min_{\theta \in \Theta} \tv(P_0, P_\theta) = \left\{\epsilon \in [0,1] : \min_{\theta \in \Theta} I_\epsilon(\theta) = 0 \right\}.
$$

Hence, in principle, we could set $\epsilon = \varepsilon_0$ to use OWL to perform minimum-TV estimation \citep{yatracos1985rates}, which has the following advantages: (1) while directly minimizing TV distance is computationally intractable, the OWL methodology decomposes this problem into alternating convex optimization and weighted MLE steps, both of which are standard problems that often tend to be well-behaved, and (2) the OWL methodology provides us with weight vectors that can indicate outlying observations and relates minimum TV-estimation to  likelihood based inference. 

In order to choose $\epsilon \approx \varepsilon_0$ in practice, we define the function $\hat{g}(\epsilon) = \hatI(\hat{\theta}_\epsilon)$, where $\hat{\theta}_\epsilon$ is the parameter estimate computed by the OWL procedure for a given $\epsilon$. At the population level, the corresponding function $g(\epsilon) = \min_{\theta \in \Theta} \okl$ is monotonically decreasing in $\epsilon$ until $\epsilon = \varepsilon_0$, at which point it remains at 0. This introduces a kink, or elbow, at $\epsilon_0$ that we hope to identify in the sample estimate $\hat{g}$. Thus, our $\epsilon$-search procedure is to compute $\hat{g}$ over a fixed grid of $\epsilon$-values, smooth the resulting grid, and then select amongst the points of largest curvature (computed numerically), where the curvature of a twice-differentiable function $f$ at a point $x$ is given by $f''(x)/(1 + f'(x)^2)^{1.5}$~\citep{satopaa2011finding}. Despite the various approximations involved, our simulation results (\Cref{sec:simul}) show that the OWL procedure with such a tuned value of $\epsilon$ provides almost identical performance when compared with the OWL procedure with the true value of $\epsilon$.

\subsection{OWL extension to non-identically-distributed data}
\label{sec:niid}

While the population level analysis and theoretical results for the OKL estimator were derived under the assumption that data are generated i.i.d.~from a distribution $P_0$, the OWL procedure can be adapted to robustify likelihood based inference in the setting where the data are conditionally independent, but not necessarily identically distributed.

Suppose data $z_1,\ldots, z_n \in \cZ$ are conditionally independent, with the likelihood having the product form $p_\theta(z_{1:n}) = \prod_{i=1}^n p_{\theta,i}(z_i)$, for known functions $\{p_{\theta,i}\}_{i=1}^n$. For example, if $z_i = (y_i, x_i) \in \R \times \cX$ for $i=1,\ldots,n$, this includes the case of regression models $\{q_\theta(y|x)\}_{\theta \in \Theta}$ under the setup $p_{\theta, i}(z_i) = q_\theta(y_i|x_i)$. Another example of this setup includes mixture models if we expand the parameter space to also include cluster assignments (see \Cref{sec:weighted-likelihood}).

To robustify inference based on the product likelihood $p_\theta(z_{1:n}) = \prod_{i=1}^n p_{\theta, i}(z_i)$, we can replace the $w$- and $\theta$- steps in \Cref{alg:owl} by analogous steps in the product likelihood case. In particular, the modified $w$-step is given by
$$
w_t = \argmin_{\substack{w \in \Delta_n \\ \frac{1}{2}\|w-o\|_1 \leq \epsilon }} \left\{- \sum_{i=1}^n w_i \log p_{\theta_t, i}(x_i)  + \sum_{i=1}^n w_i \log w_i\right\}
$$
and the modified $\theta$-step is given by
$$
\theta_{t+1}= \argmin_{\theta \in \Theta} \sum_{i=1}^n w_{t,i} \log p_{\theta,i}(x_i).
$$
Despite our lack of theory in the non-identically-distributed case, we continue to see good empirical performance of OWL (see \Cref{sec:simul}, \Cref{sec:micro-credit} and \Cref{sec:application}).

\section{Asymptotic connection to coarsened inference}
\label{sec:coarsened-inference}

The development of the OWL methodology in \Cref{sec:methodology} followed from a presumed form of misspecification given by \Cref{ass:misspecification}. An alternative way to frame and address such misspecifications in a probabilistic framework was proposed by \cite{miller2018robust} who introduced a Bayesian methodology centered around the concept of a \emph{coarsened likelihood} defined as
\begin{equation}
    \label{eq:clike}
    L_\epsilon(\theta|x_{1:n}) \doteq   \prob_\theta\left(\D(\EmpDist{Z_{1:n}}, \EmpDist{x_{1:n}}) \leq \epsilon\right), 
\end{equation}
where $\D$ is a suitably chosen discrepancy between empirical probability measures.
Here, 
$\EmpDist{x_{1:n}} = n^{-1}\sum_{i=1}^n \delta_{x_i}$ denotes the empirical distribution of data $x_{1:n}$, and the probability is computed under $\prob_\theta$---the distribution underlying the artificial data $Z_1, \ldots Z_n \iid P_\theta$ from which the random measure $\EmpDist{Z_{1:n}} = n^{-1}\sum_{i=1}^n \delta_{Z_i}$ is constructed. 
The coarsened likelihood implicitly captures the likelihood of a probabilistic procedure in which idealized data are first generated by some model $\prob_\theta$ in the model class under consideration, but are then corrupted in such a way that the discrepancy between empirical measures of the idealized data and the observed data is bounded by $\epsilon$. 

When $\D$ is an estimator for the KL-divergence and an exponential prior is placed on $\epsilon$, \cite{miller2018robust} showed that the Bayes posterior based on $L_\epsilon(\theta|x_{1:n})$ could be approximated  by raising the likelihood to a power less than one in the formula for the standard posterior. 
However, to obtain a robustified alternative to maximum likelihood estimation, one may wish to maximize $\theta \mapsto L_\epsilon(\theta|x_{1:n})$ directly for a choice of $\D$ that guarantees robustness (e.g. Maximum Mean Discrepancy or the TV distance).
Such an approach would in general be quite challenging since evaluating \cref{eq:clike} corresponds to computing a high-dimensional integral. 

In this section, we show that for large $n$, the coarsened likelihood can be approximately maximized via the OWL methodology when $\D$ is an estimator for the TV distance. Specifically, if the observed data $x_1, \ldots, x_n$ are generated i.i.d.~from some distribution $P_0$ and $\D$ satisfies some regularity conditions, then the quantity $-\frac{1}{n} \log L_\epsilon(\theta|x_{1:n})$ asymptotically converges  as $n \to \infty$ to a variant of $I_\epsilon(\theta)$ based on $\D$. Hence, the OWL methodology asymptotically maximizes the coarsened likelihood $\theta \mapsto L_\epsilon(\theta|x_{1:n})$. 

We state this asymptotic connection first for finite spaces and then for continuous spaces. These results rely on Sanov's theorem from large deviation theory~\citep{demboLargeDeviationsTechniques2010} and are proved in \Cref{sec:coarsened-likelihood-asymptotics}.

\subsection{Asymptotic connection in finite spaces}
\label{sec:asymp-finite}

Let $\cX$ be a finite set and denote the space of probability distributions on $\cX$ by the simplex $\Delta_{\cX} \doteq \{q \in [0,1]^{\cX} | \sum_{x \in \cX} q(x) = 1\}$. Let $\{p_\theta\}_{\theta \in \Theta} \subseteq \Delta_{\cX}$ denote the collection of model distributions, and $p_0 \in \Delta_{\cX}$ denote the true data generating distribution. To connect the OKL with the coarsened likelihood in \cref{eq:clike}, we take $\D(p, q) = \frac{1}{2}\|p - q\|_1$. 

Given this setting, we can show that $-\frac{1}{n} \log L_\epsilon(\theta | x_{1:n})$ converges in probability to the OKL function $I_\epsilon(\theta)$ at a rate of $n^{-1/2}$, as demonstrated by the following theorem.
\begin{theorem}
\label{thm:finite-cposterior-tv}
Suppose that $I_{\epsilon_0}(\theta) < \infty$ for some $\epsilon_0 > 0$ and let $\delta > 0$. If $\epsilon > \epsilon_0$ and $x_1, x_2, \ldots, x_n \iid p_0$, then with probability at least $1-\delta$,
\[ \left|I_\epsilon(\theta) + \frac{1}{n} \log L_\epsilon(\theta|x_{1:n}) \right| \leq O\left( \frac{|\Xcal|}{\epsilon - \epsilon_0} \sqrt{\frac{1}{n} \log \frac{1}{\delta}}  + \frac{|\Xcal|}{n} \left( |\Xcal| + \log(n) + \frac{1}{\epsilon - \epsilon_0} \right)  \right).\]
\end{theorem}

\Cref{thm:finite-okl-convergence-tv} and \Cref{thm:finite-cposterior-tv} together show that, in the large sample limit, the OWL methodology and coarsened likelihood philosophy are two sides of the same coin: they both provide approximations of the OKL and, in turn, must approximate each other.

\subsection{Asymptotic connection in continuous spaces}
\label{sec:asymp-continuous}

Suppose $\Xcal$ is a Polish space (e.g. $\cX = \R^d$). Similar to the finite case, we can establish the following asymptotics for the coarsened likelihood for a suitable class of discrepancies $\D$, which includes the Wasserstein distance, Maximum Mean Discrepancy with  suitable choice of kernels \citep{simon2018kernel},  and the smoothed TV distance (\Cref{def:smoothed-tvd} in \cref{sec:smoothed-tvd}).

\begin{theorem}
\label{thm:clikelihood-asymptotics-density}
Suppose $I_{\epsilon_0}(\theta) < \infty$ for some $\epsilon_0 > 0$ and $\D: \cP(\cX) \times \cP(\cX) \to [0,\infty)$ is a pseudometric that is convex in its arguments and continuous with respect to the weak convergence topology on $\cP(\cX)$. If $\epsilon > \epsilon_0$ and $x_1, \ldots, x_n \iid P_0$, then
	\begin{equation*}
		-\frac{1}{n} \log L_\epsilon(\theta|x_{1:n}) \pconv \inf_{\substack{Q \in \cP(\cX) \\ \D(Q, P_0) \leq \epsilon}} \KL(Q|P_\theta) \quad 
		\text{ as $n \to \infty$.}
	\end{equation*}
\end{theorem}

Recall that the limiting expression in the above theorem has the same form as that of the OKL function given in \cref{eq:okl}. However, in order to establish connection between the OKL function and the coarsened likelihood, unlike in the finite case, we cannot merely take the discrepancy $\D$ in the coarsened likelihood to be the TV distance, since the TV distance between the two empirical distributions in \cref{eq:clike} will almost surely be equal to one. Instead, we will take $\D$ to be a smoothed version of TV distance calculated by first convolving the empirical measures with a smooth kernel function $K_h:\cX \times \cX \to [0,\infty)$ indexed by a bandwidth parameter $h > 0$. Further details can be found in \Cref{sec:smoothed-tvd}.

\section{Simulation Examples}
\label{sec:simul}
We now demonstrate the OWL methodology in simulated examples with artificially injected corruptions. In each simulation, we considered two methods for choosing the points to corrupt: (i) {\em max-likelihood corruption} where we fit a maximum likelihood estimate to the uncorrupted data and select the points with the highest likelihood; and (ii) {\em random corruption} where we choose the points to corrupt uniformly at random. \new{We ran each simulation with 50 initial seeds, plotting the mean performance and its 95\% confidence band}. For clarity and space, we present only the results for max-likelihood corruptions in this section, and we defer the results for randomly-selected corruptions to \Cref{app:more-simulations}. 

In all comparisons, OWL refers to our methodology with the data-based choice of the corruption fraction $\epsilon$ as described in \Cref{sec:tune-epsilon}, while OWL ($\epsilon$ known) refers to our methodology with $\epsilon$ equal to the true level of corruption in the data. MLE refers to the standard maximum likelihood estimate. \new{The Pearson residuals baseline is the method from \cite{markatou1998} based on the Hellinger residual adjustment function: a weighted-likelihood based method aimed at performing minimum Hellinger-distance estimation}. For settings where the outcomes were continuous, the Pearson residual density estimate was computed via kernel density estimation, with the bandwidth parameter selected to minimize the empirical Hellinger distance between the final model and the density estimate. For the linear regression baselines, ridge regression was performed with L2 penalty chosen via cross validation, and Huber regression used the commonly-accepted Huber penalty of 1.345~\citep{huber1964robust}. For the logistic regression baselines, L2-regularized MLE selected the L2 penalty via cross validation. For both regression settings, Random Sample Consensus (RANSAC) utilized the ground-truth corruption fraction.

\begin{figure}[ht]
    \centering
    \includegraphics[width=1.\textwidth]{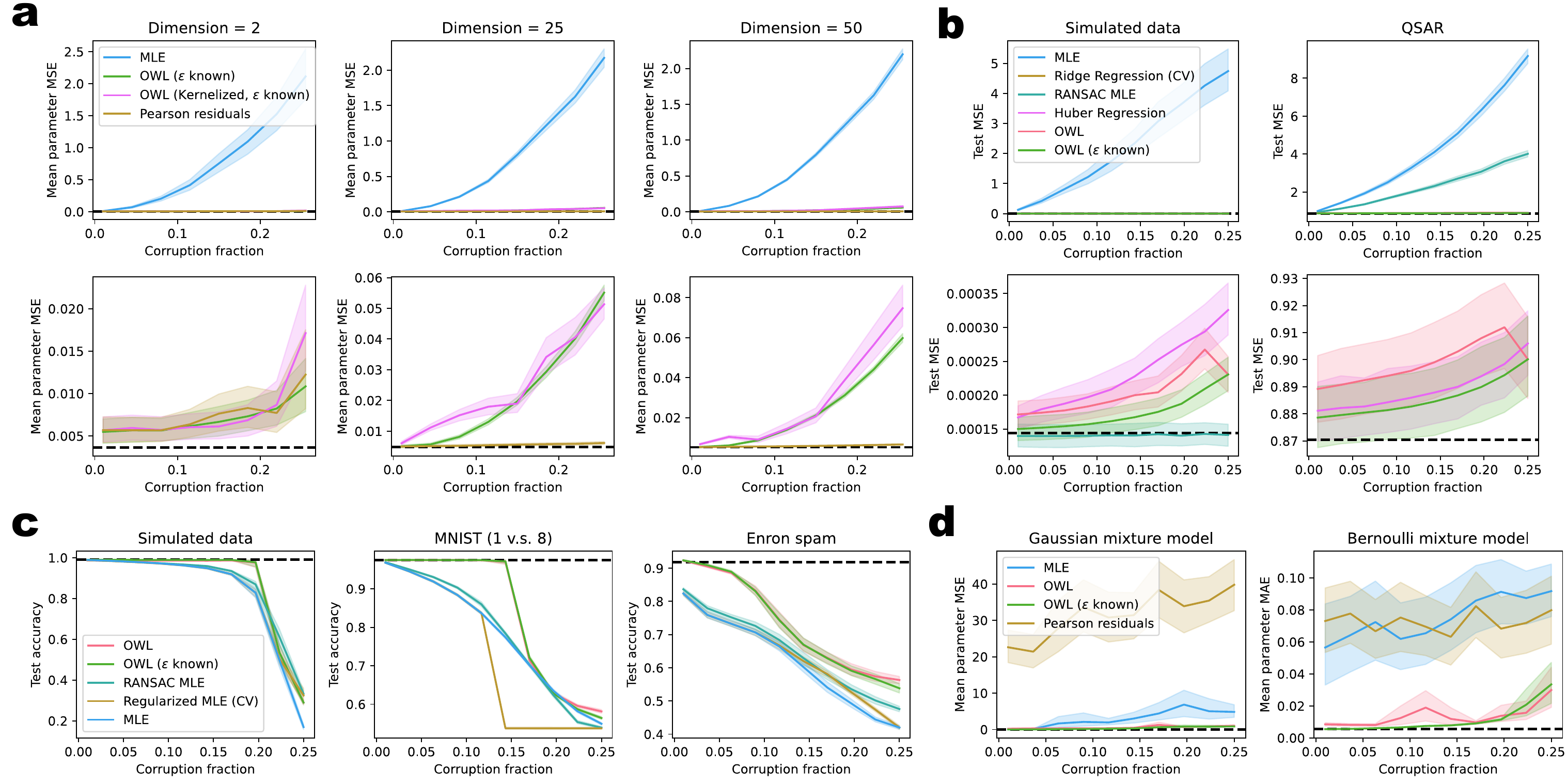}
    \caption{Simulation results for max-likelihood corruptions. (a) Mean parameter reconstruction for a multivariate normal model. (b) Test MSE for linear regression. (c) Test accuracy for logistic regression. (d) Mean parameter reconstruction for mixture models. In all figures, the dashed black line denotes median performance of MLE on full uncorrupted training set, and the shaded regions denote bootstrapped 95\% confidence intervals over 50 random seeds.}
    \label{fig:all-simulation}
\end{figure}

\paragraph{Gaussian modeling.} Our first simulation fit a multivariate normal distribution to data generated from a $p$-dimensional spherical normal distribution. The ground truth mean was drawn uniformly from $[-10,10]^p$, and the corrupted data points were also drawn uniformly from $[-10,10]^p$. The total number of points was 200, and we measured the dimension-normalized mean-squared error (MSE) of ground-truth mean parameter recovery. 

\new{\Cref{fig:all-simulation}a shows the results of the Gaussian simulation. In particular, we see that both the kernelized and the unkernelized versions of OWL had almost identical performance in this example, achieving substantial robustness to corruptions over the sample mean. In light of this similarity of performance, the rest of our simulation examples utilize the much more computationally efficient unkernelized version of OWL. We also observed that the kernelized version of OWL was not particularly sensitive to the choice of bandwidth parameter (\Cref{fig:kernel-choice-sim}). We further observe that on this example, the Pearson residuals method outperforms OWL, although the difference in performance is small relative to the gains of all methods over MLE.}

\paragraph{Linear regression.} We considered a homoscedastic model with normally distributed errors. We considered two datasets. The first is a simulated dataset with 10-dimensional i.i.d. standard normal covariates. The ground-truth regression coefficients were drawn independently from $\Ncal(0,4)$, the intercept was set to 0, and the residual standard deviation was 
$\sigma=1/4$. The training set consisted of 1,000 data points. For the test set, we drew 1,000 new data points and computed the MSE on the underlying mean response.

The second dataset was taken from a quantitative structure activity relationship (QSAR) dataset \citep{data:qsar} compiled by \cite{olier2018meta} from the ChEMBL database. It consists of 5012 chemical compounds whose activities were measured on the epidermal growth factor receptor protein erbB1. The activities were recorded as the negative log of the chemical concentration that inhibited 50\% of the protein target, i.e. the pIC$_{50}$ value. Each compound had 1024 associated binary covariates, corresponding to the 1024-dimensional FCFP4 fingerprint representation of the molecule~\citep{rogers2010extended}. We used PCA to reduce the dimension to 50. For every random seed, we computed a random 80/20 train/test split. The test MSE on this dataset is the standard MSE over the test responses.  In both datasets, for each data point selected to be corrupted, we corrupted the responses by fitting a least squares solution and observing the residuals: if the residual is positive, we set the response to be $3v$ where $v$ is the largest absolute value observed value in the training set responses, otherwise setting it to $-3v$.

\Cref{fig:all-simulation}b shows the results of the linear regression simulations for the max-likelihood corruptions. Across both datasets, we see that OWL is competitive with the best of the robust regression methods, whether that method is RANSAC or Huber regression. 

\paragraph{Logistic regression.} For the logistic regression setting, we have parameters $b \in \R$, $w \in \R^p$ and observations $(x_i, y_i) \in \R^{p + 1}$ assumed to follow the distribution
\[ y_i \sim \text{Bernoulli} \left( \frac{1}{1 + \exp\left( -\langle x_i, w \rangle - b \right)} \right) . \]
We considered three datasets. The first is a simulated dataset using the same parameters as the linear regression setting. The training labels are created according to the generative model. For test accuracy, we computed accuracy against the true sign-values, i.e. $\I{\langle w^\star , x \rangle \geq 0}$.

The second dataset is taken from the MNIST handwritten digit classification dataset~\citep{lecun1998gradient}. We considered the problem of classifying the digit `1' v.s. the digit `8,' resulting in a dataset with 14702 data points and 784 covariates,  representing pixel intensities. The third dataset is a collection of 5172 documents from the Enron spam classification dataset, preprocessed to contain 5116 covariates, representing word counts~\citep{metsis2006spam}. For both the MNIST and the Enron spam datasets, we reduced the dimensionality to 10 via PCA and used a random 80/20 train/test split.

\Cref{fig:all-simulation}c shows the results of the logistic regression simulations. Across all datasets, OWL again outperforms the other approaches in the presence of corruption.

\paragraph{Gaussian mixture models.} Recall the standard Gaussian mixture modeling setup: there are a collection of means $\mu_1, \ldots, \mu_K \in \R^p$, standard deviations $\sigma_1, \ldots, \sigma_K > 0$, and mixing weights $\pi \in \Delta^{K}$. Data points $x_i \in \R^p$ are drawn i.i.d. according to
\[ x_i \sim \sum_{k=1}^K \pi_k \Ncal(\mu_k, \sigma_k^2 I_p). \] 
For our simulations, we generated a synthetic dataset of 1000 points in $\R^{10}$ by first drawing $K=3$ means $\mu^\star_1, \ldots, \mu^\star_K$ whose coordinates are i.i.d.~Gaussian with standard deviation 2. 
We set $\sigma_1 = \sigma_2 = \sigma_3 = 1/2$ and $\pi_1 = \pi_2 = \pi_3 = 1/3$. To corrupt a data point, we randomly selected half of its coordinates and set them randomly to either a large positive value or a large negative value (here, 5 and -5).
As a metric, we measured the average mean squared Euclidean distance between the means of the fitted model and the ground truth model. 

For all methods, we used random restarts, choosing the final model based on the method's criterion: likelihood for MLE, the OKL estimate for OWL, and empirical Hellinger distance for Pearson residuals. We see that OWL remains robust against varying levels of corruptions, whereas both MLE and Pearson residuals perform significantly worse (left panel of \Cref{fig:all-simulation}d).

\paragraph{Bernoulli product mixture models.} Consider the following model for $p$-dimensional binary data: there are a collection of probability vectors $\lambda_1, \ldots, \lambda_K \in [0,1]^p$ and mixing weights $\pi \in \Delta^{K}$. Each data point $x_i$ is drawn i.i.d.~according to the process 
\begin{align*}
z_i &\sim \text{Categorical}(\pi) \\
x_{ij} &\sim \text{Bernoulli}( \lambda_{z_i j} ) \text{ for } j=1,\ldots, p.
\end{align*}
For our simulations, we generated a synthetic dataset of 1000 points in $\{0,1\}^{100}$ by first drawing $K=3$ means $\lambda^\star_1, \ldots, \lambda^\star_K$ whose coordinates are i.i.d.~from a Beta$(1/10, 1/10)$ distribution. The mixing weights were chosen to be uniform over the components. To corrupt a data point, we flipped each zero coordinate with probability 1/2.
As a metric, we measured the average $\ell_1$-distance between the $\lambda$ parameters of the fitted model and the ground truth model.

 The right panel of \Cref{fig:all-simulation}d shows the results of the Bernoulli mixture model simulations. We see that OWL remains robust against varying levels of corruptions, whereas both MLE and Pearson residuals perform significantly worse.
 
\section{Application to micro-credit study}
\label{sec:micro-credit}
In this section we apply OWL to data from a micro-credit study for which maximum likelihood estimation (MLE) is shown to be brittle. In \cite{angelucci2015microcredit}, the authors worked with one of the largest micro-lenders in Mexico to randomize their credit rollout across 238 geographical regions in the Sonora state. Within 18-34 months after this rollout, the authors surveyed $n=16,560$ households for various outcome measures.

Following \cite{broderick2020automatic}, here we focus on the \emph{Average Intention to Treat effect} (AIT) of the rollout on household profits. For $i \in \{1, \ldots, n\}$, let $Y_i$ denote the profit of the $i$th household during the last fortnight (in USD PPP units), and let $T_i \in \{0,1\}$ be a binary variable that is one if the household $i$ falls in the geographical region where the rollout happened.  The AIT on household profits is defined as the coefficient $\beta_1$ in the  model:
\begin{equation}
\label{eq:AITmodel}
Y_i = \beta_0 + \beta_1 T_i + \varepsilon_i,\quad \varepsilon_i \iid N(0, \sigma^2),  \quad i \in \{1,\ldots, n\}. \end{equation}
In \Cref{sec:brittleness}, we reproduce the brittleness in estimating $\beta_1$ using the MLE as  demonstrated in \cite{broderick2020automatic} by removing an handful of observations.

Here we compare OWL to the above data deletion approach. We fit the model \eqref{eq:AITmodel} to the full data set using 50 $\log_{10}$-spaced $\epsilon$-values between $10^{-4}$ and $10^{-1}$, and used the tuning procedure in \Cref{sec:tune-epsilon} to select the value $\epsilon_0 = 0.005$ where the minimum-OKL versus epsilon plot (\Cref{app:micro-credit}, \Cref{fig:micro-okl-plot}) has its  most prominent kink. We also calculate the MLE, which corresponds to OWL with $\epsilon = 0$. The AIT on household profit estimated by OWL as a function of $\epsilon$ can be seen in the left panel of \Cref{fig:micro_fig}. For  values of $\epsilon$ below $\epsilon_0$, the AIT estimates change rapidly as $\epsilon$ changes, while for  values of $\epsilon$ above $\epsilon_0$,  the AIT estimates are quite stable with changes in $\epsilon$. This is due to OWL automatically down-weighting the outlying observations, as seen in the right panel of  \Cref{fig:micro_fig}. We  quantify uncertainty in our estimates by using a variant of the Bootstrap (\Cref{sec:os-bootstrap}).

\begin{figure}[h]
	\centering
	\includegraphics[width=0.47\textwidth]{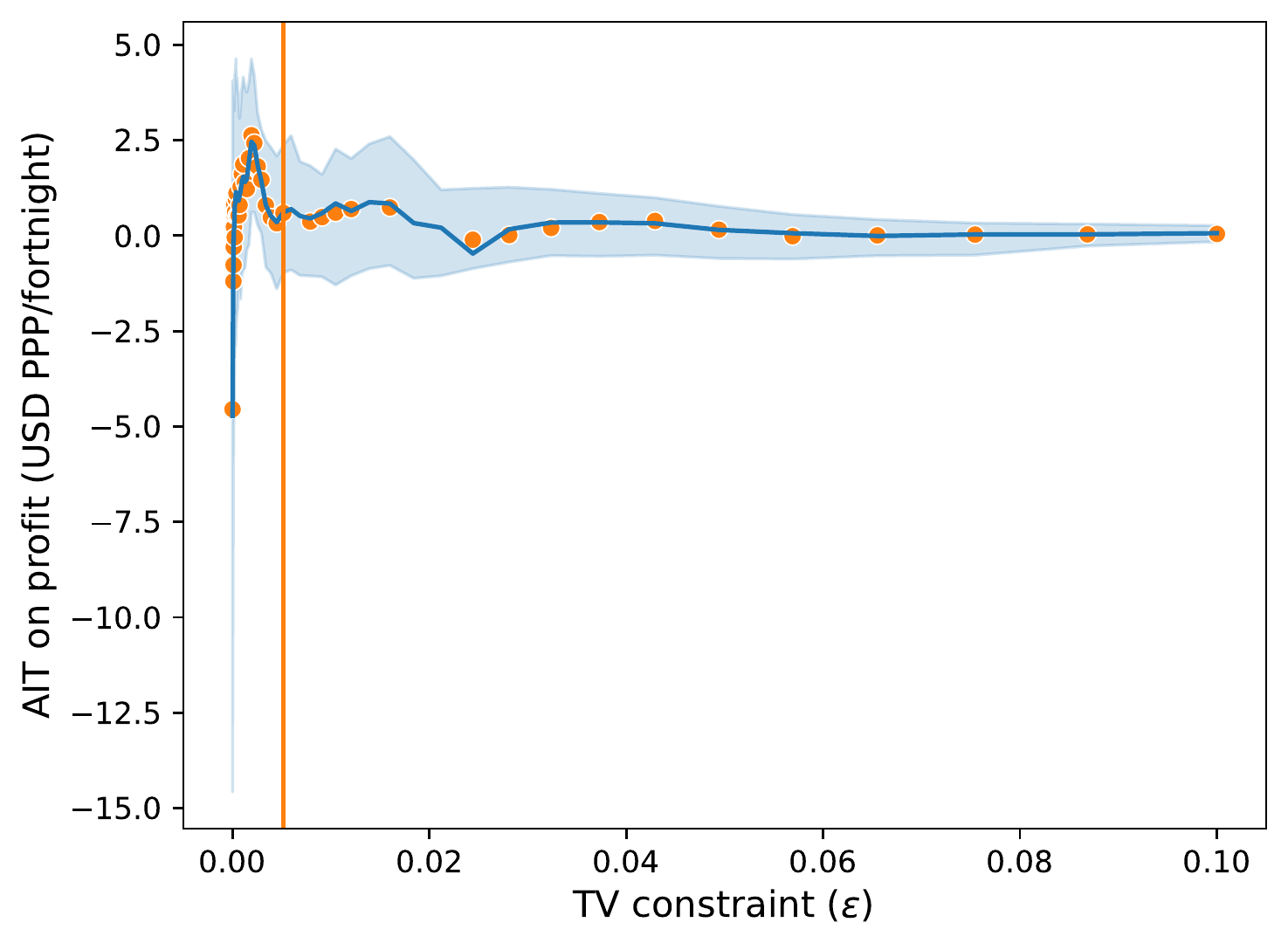}
	\includegraphics[width=0.47\textwidth]{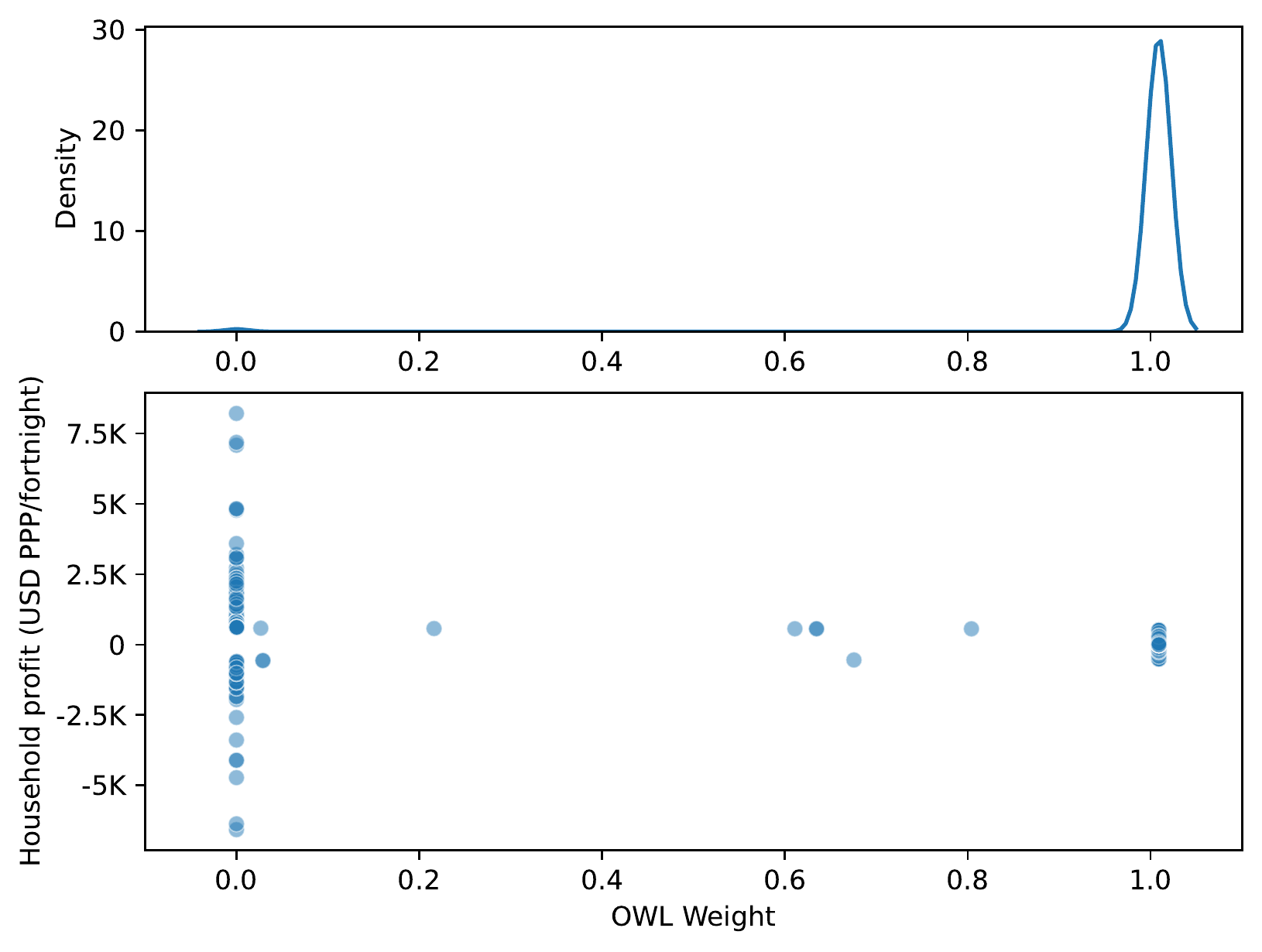}
	\caption{Estimating the Average Intention to Treat (AIT) effect on household profits in the micro-credit study \cite{angelucci2015microcredit} in the presence of outliers. Left: the AIT estimates using OWL for various values of $\epsilon$ along with $90\%$  bootstrap vertical confidence bands. The vertical line is drawn at the value $\epsilon_0 = 0.005$ obtained by the tuning procedure in \Cref{sec:tune-epsilon}, and roughly coincides with the $\epsilon$ beyond which the AIT estimates stabilize and the size of the confidence bands shrinks (see \Cref{app:micro-credit}). Right: shows that the weights estimated by OWL procedure at $\epsilon = \epsilon_0$ down-weight roughly  $1\%$ of the households that have outlying profit values (for visual clarity, we omit a down-weighted household with profit less that $-40K$ USD PPP); see also \Cref{fig:micro-outlier-hist-plot} in \Cref{app:micro-credit}.}
	\label{fig:micro_fig}
\end{figure}

In summary, the OWL procedure chose to down-weight roughly 1\% of the households with extreme profit values and estimated an AIT of $\beta_1 = 0.6$ USD PPP per fortnight based on the selected value of $\epsilon_0 = 0.005$. The value $\epsilon = \epsilon_0$, tuned using the procedure in \Cref{sec:tune-epsilon}, roughly coincides with the point at which the AIT estimates become stable with respect to  $\epsilon$ and also with the point at which the 90\% confidence bands for AIT become narrower --- both suggesting that OWL with the choice $\epsilon = \epsilon_0$ has identified and down-weighted outliers that may be causing brittleness in estimating AIT.
 
\section{Discussion}
In this paper, we introduced the optimistically weighted likelihood (OWL) methodology, motivated by brittleness issues arising from misspecification in statistical methodology based on standard likelihoods. On the theoretical side, we established the consistency \new{and robustness} of our approach and showed its asymptotic connection to the coarsened inference methodology. We also proposed a feasible alternating optimization scheme to implement the methodology and demonstrated its empirical utility on both simulated and real data.

The OWL methodology opens up several interesting future directions. One practical open problem is how to scale to larger datasets. As a weighted likelihood method, OWL requires solving for a weight vector whose dimension is the size of the dataset. While we can solve the resulting convex optimization problem for thousands of data points, the procedure becomes significantly more complicated when the size of the dataset exceeds computer memory. How do we maintain a feasible solution, i.e.~one that lies in the intersection of the simplex and some probability ball, when the entire vector cannot fit in memory?

Another practical question is how to apply the OWL approach in more complex models; for example, involving dependent data. This may be relatively straightforward for models in which the likelihood can still be written in product form due to conditional independence given random effects. This would open up its  application to nested, longitudinal, spatial and temporal data, as random effects models are routinely used in such settings.

\new{
Finally, it will be fruitful to implement the OWL method for a choice of $\D$ other than the TV distance. Indeed, as reflected in \Cref{ass:misspecification} and our robustness results, the choice of $\D$ controls the nature of robustness offered by OWL. An important choice for $\D$ may be the Wasserstein distance because it metrizes the weak convergence topology on bounded spaces \cite[Section 2.4]{huber2009} and, as argued in Chapter 1 of \cite{huber2009}, can capture robustness to both small changes in all observations (e.g.~due to rounding or grouping) and large changes in a few observations (e.g.~due to contamination or outliers). Indeed, while several of our theoretical results already hold for the Wasserstein metric, appropriate modifications to the OKL estimator and the OWL algorithm will be needed to make this approach feasible.}

\paragraph{Acknowledgement}
\if1\blind
{
The authors acknowledge funding from grants N00014-21-1-2510-P00001 from the Office of Naval Research (ONR) and R01ES027498, U54 CA274492-01 and R37CA271186  from the National Institutes of Health, as well as helpful discussions with Sayan Mukherjee and Amarjit Budhiraja.} \fi
\new{They authors also wish to thank the two anonymous referees for their comments and for pointing out the relevant literature.}

\paragraph{Supplementary Materials}
The accompanying supplementary materials contain a real data application to cluster single-cell RNAseq data, along with additional details referenced in the article. \if1\blind
See \url{https://github.com/cjtosh/owl} for code to reproduce all analyses.
\fi

\putbib[refs]
\end{bibunit}

\appendix
\part{} 

\begin{bibunit}[abbrvnat]
\pagebreak
    \begin{center}
    \textbf{\large Supplementary Materials: Robustifying Likelihoods by Optimistically Re-weighting Data}
    \end{center}
    \parttoc \setcounter{equation}{0}
    \setcounter{figure}{0}
    \setcounter{table}{0}
    \setcounter{page}{1}
    \setcounter{theorem}{0}
    \setcounter{lemma}{0}
    \setcounter{proposition}{0}

\renewcommand{\theequation}{S\arabic{equation}}
    \renewcommand{\thefigure}{S\arabic{figure}}
    \renewcommand{\thetable}{S\arabic{table}}
    \renewcommand{\thesection}{S\arabic{section}}
    \renewcommand{\thepage}{S\arabic{page}}
    \renewcommand{\thetheorem}{S\arabic{theorem}}
    \renewcommand{\thelemma}{S\arabic{lemma}}
    \renewcommand{\theproposition}{S\arabic{proposition}}
    \renewcommand{\bibnumfmt}[1]{[S#1]} \renewcommand{\citenumfont}[1]{S#1} 

\section{Comparison with classical robust statistics}
\label{sec:more-classical-robustness}
\new{
The systematic study of statistical robustness  that originated in the pioneering work of \cite{huber1964robust} and \cite{hampel1968contributions} has since extensively been developed into a practical methodology, with several book-length treatments now available on this topic like \cite{huber2009,Hampel2005, maronna2019robust}. Over the years, there have been several surveys articles that attempt to provide an overview of this field like \cite{huber1972,hampel2001robust, rousseeuw2011robust, avella2015robust}. Arguably, statistical robustness should be of primary concern in today's world of large datasets, where no interpretable (i.e.~parametric) statistical model will be sufficient to exactly describe all aspects of the data, even though such models may be adequate to describe well many important aspects of the data.

    The classical approach to construct robust estimators based on influence functions in  \cite{Hampel2005} or the minimax theory in \cite{huber2009} provide a way to choose optimal estimators from a class of estimators that is (a) Fisher consistent at the model, and (b) robust to arbitrary contamination in a small neighborhood away from the model. Indeed, the focus on properties (a) and (b) is both technical and practical, since it is possible to construct optimal M-estimators satisfying properties (a) and (b) that: have a bounded influence function, can be computed using an iteratively re-weighted maximum likelihood procedure, and are asymptotically efficient at the model \citep[see e.g.][]{avella2015robust}.

    In this work, our focus has been to systematically develop a model fitting methodology that can handle general    misspecification beyond the classical Huber's contamination model considered in (b), which was motivated by the presence of outliers (called   \emph{gross errors} in \cite{Hampel2005, huber2009}). Motivated by the problem of averting brittleness in model selection using the coarsened inference methodology of \cite{miller2018robust} (see \Cref{sec:coarsened-inference}), we particularly wanted to allow for small changes to the shape of the distribution that might not be a Huber's contamination (\Cref{sec:model-selection-illustration}). Further, unlike methods based on M-estimation, the OWL procedure outputs a set of parameter values rather than a unique optimum parameter value, which we think better represents the uncertainty of parameter estimation when our models are misspecified. Indeed, even as $n \to \infty$, it seems odd for parameter estimation methods to be increasingly confident about the parameter value when we know that our models are wrong. This also makes us wary about the use of confidence intervals based on the asymptotic normality of M-estimators \cite[Chapter 5]{van2000asymptotic} when dealing with  large datasets (i.e.~$n \to \infty)$, since these intervals shrink at the rate of $n^{-1/2}$ and may be overconfident.

A practical advantage of classical robustness methodology based on M-estimators is that such estimators reduce to solving weighted likelihood equations with weights which are in turn determined using the data \cite[see e.g.][Section 2.8]{maronna2019robust}. The computed weights provide a measure of influence that each observation has on the final procedure, explaining the behavior of such procedures in sophisticated data models (e.g.~\cite{mancini2005optimal} handles the case of dependent data and \cite{avella2018robust} handles variable selection in high-dimensional regression). Alternatively, data weights can also be chosen as a function of the cumulative distribution function or density,  aimed at down-weighting observations that lie in the tails of the model \citep{field1994,dupuis2002,windham1995}. However, while methods like \cite{field1994} can offer robustness to data contamination, they are not designed to handle small errors in the shape of the model distribution (e.g.~increased skewness) that are not attributed to outliers. On the other hand, the \emph{optimistic re-weighting} proposed in this work can correct for arbitrary changes in the shape of the distribution within a small total-variation neighborhood (see \Cref{sec:model-selection-illustration}).}

\section{Illustrating brittleness of model selection}
\label{sec:model-selection-illustration}
\new{Even minor model misspecification can have a dramatic impact on the problems of model assessment and model selection, particularly for large sample sizes. Indeed, since models are never exactly true, for a sufficiently large sample size one can always find evidence to reject any model under consideration. This has motivated the consideration of model credibility indices \citep{lindsay2009model} and tolerance regions around models \citep{liu2009building} to determine the adequacy of a model for the task at hand.}  

\new{Here we focus on the problem of selecting the number of components in kernel mixture models. Suppose that data $x_1, \ldots, x_n$ are sampled from a mixture of kernels $\sum_{j=1}^{k_0} \pi_j^* f(\cdot|\theta^*_j)$ with $k_0 < \infty$  components such that $\pi_j^* > 0$ for every $j \in [k_0]$. Estimation of $k_0$ is done using a model selection criteria like AIC or BIC  \citep{mclachlan2019finite}. Particularly, let $l_k(\bpi; \btheta) = \sum_{i=1}^n \log(\sum_{j=1}^k \pi_j f(x_i|\theta_j))$ denote the log-likelihood associated with the mixture model  with $k$ components, weights $\bpi=(\pi_1, \ldots, \pi_k) \in \Delta_k$, and parameters $\btheta = (\theta_1, \ldots, \theta_k)$. The number of mixture components $\hat{k}_0$  is estimated as the minimizer of
\begin{equation}
\label{eq:selection-criteria}
\textrm{Selection-Criterion}(k) = 2\kappa_{k,n} - 2\max_{\bpi \in \Delta_k, \btheta \in \Theta^k} l_k(\bpi, \btheta),
\end{equation}
where the AIC uses $\kappa_{k,n} = k\times(\dim(\Theta)+1)$ (the number of parameters to fit),  while the BIC uses $\kappa_{k,n} = \frac{k(\dim(\Theta)+1)}{2} \ln(n)$.
When the model is indeed correct, BIC criterion can be shown to consistently select the true number of components $k_0$ \citep{keribin2000consistent} as $n \to \infty$.}

\new{In practice, however, the data may not be sampled exactly from a mixture of the (typically Gaussian) kernels under consideration. In such cases,  mixture models with selection tend to over-cluster the data, particularly for large sample size \citep{miller2018robust,cai2021finite,mclachlan2019finite}.  In \Cref{fig:model-selection-example},  we provide a simple illustration to demonstrate the failure of AIC and BIC  to select two mixture components using Gaussian mixture models when the data consists of two well-separated clusters, one of which has a slightly skewed density, which is thus not Gaussian. This example illustrates the problem associated with using mixture models for clustering and sub-population identification when model assumptions are slightly violated. Interestingly, we are able to select the correct number of components in this example if we suitably modify the selection criteria using optimistic weights from the OWL procedure \eqref{eq:owl-selection-criteria}.  We recommend further investigation of such approaches for robust model selection.}

\new{
\paragraph{Details of the data used in \Cref{fig:model-selection-example}.} The observed data $x_1, \ldots, x_n$ in  \Cref{fig:model-selection-example} were obtained as a random sample of $n=1000$ observations drawn from the mixture distribution  $P_0 = \frac{1}{4}N(0,1) + \frac{3}{4} N\chi^2_d(5,1)$, where a sample $Y = \mu + \sigma \frac{Z - d}{\sqrt{2d}}$ from our skewed distribution $N\chi^2_d(\mu, \sigma^2)$ is an appropriately scaled and translated Chi-squared random variable $Z \sim \chi^2_d$, satisfying $\E[Y] = \mu$ and $\textrm{Var}(Y) = \sigma^2$. As $d \to \infty$, note that $N\chi^2_{d}(\mu, \sigma^2)$ converges in distribution to $N(\mu, \sigma^2)$ due to the central limit theorem. Here we considered  $d=10$ so that $N\chi^2_d$ has a noticeable  skew. 

\paragraph{Details of the selection criteria used in \Cref{fig:model-selection-example}.}
We then estimated AIC and BIC for four candidate model fitting methods: (a) Soft EM algorithm, (b) Hard EM algorithm, (c) OWL $\epsilon = 0.05$, and (d) OWL $\epsilon = 0.1$. Particularly, method (a) refers to the regular EM algorithm \citep[e.g.][]{mclachlan2019finite} used to perform the maximization  in \Cref{eq:selection-criteria}.  For method (b), the likelihood term $l_k(\bpi, \btheta)$ in \Cref{eq:selection-criteria} is replaced by $\tilde{l}_k(\bpi, \bz,\btheta) = \sum_{i=1}^n \sum_{j=1}^k \I{z_i = j} \log (\pi_j f(x_i|\theta_j))$, which is then maximized over (hard) assignments $\bz \in [k]^n$, $\bpi \in \Delta_k$, and $\btheta \in \Theta^k$ using the hard EM algorithm \citep[e.g.][]{satopaa2011finding}.  While both the soft and hard EM algorithms are designed to maximize similar likelihoods $l_k$ and $\tilde{l}_k$, they can have slightly different behaviors \citep{kearns1998information,celeux1993comparison}.  We included both of them for this analysis since the OWL procedure is based on maximizing a weighted version of $\tilde{l}_k$. For methods (c) and (d), we fit the mixture model  using OWL with the chosen value of $\epsilon$ to obtain optimal parameters $\hat{\btheta}_\epsilon \in \Theta^k$,  assignment $\hat{\bz}_\epsilon \in [k]^n$, and optimistic weights $n^{-1} \hat{w}_\epsilon \in\Delta_n$. We then modify the selection criteria \eqref{eq:selection-criteria} using optimistic weights as follows:
\begin{equation}
\label{eq:owl-selection-criteria}
\textrm{OWL-Selection-Criteria}_\epsilon(k) = 2\kappa_{n,k} - 2 \sum_{i=1}^n \hat{w}_{\epsilon,i} \sum_{j=1}^k \I{\hat{z}_{\epsilon,i} = j} \log (\pi_j f(x_i|\hat{\theta}_{\epsilon,j})).
\end{equation}
}

\begin{figure}[h]
    \centering
    \begin{subfigure}{0.8\textwidth}
    \includegraphics[width=\textwidth]{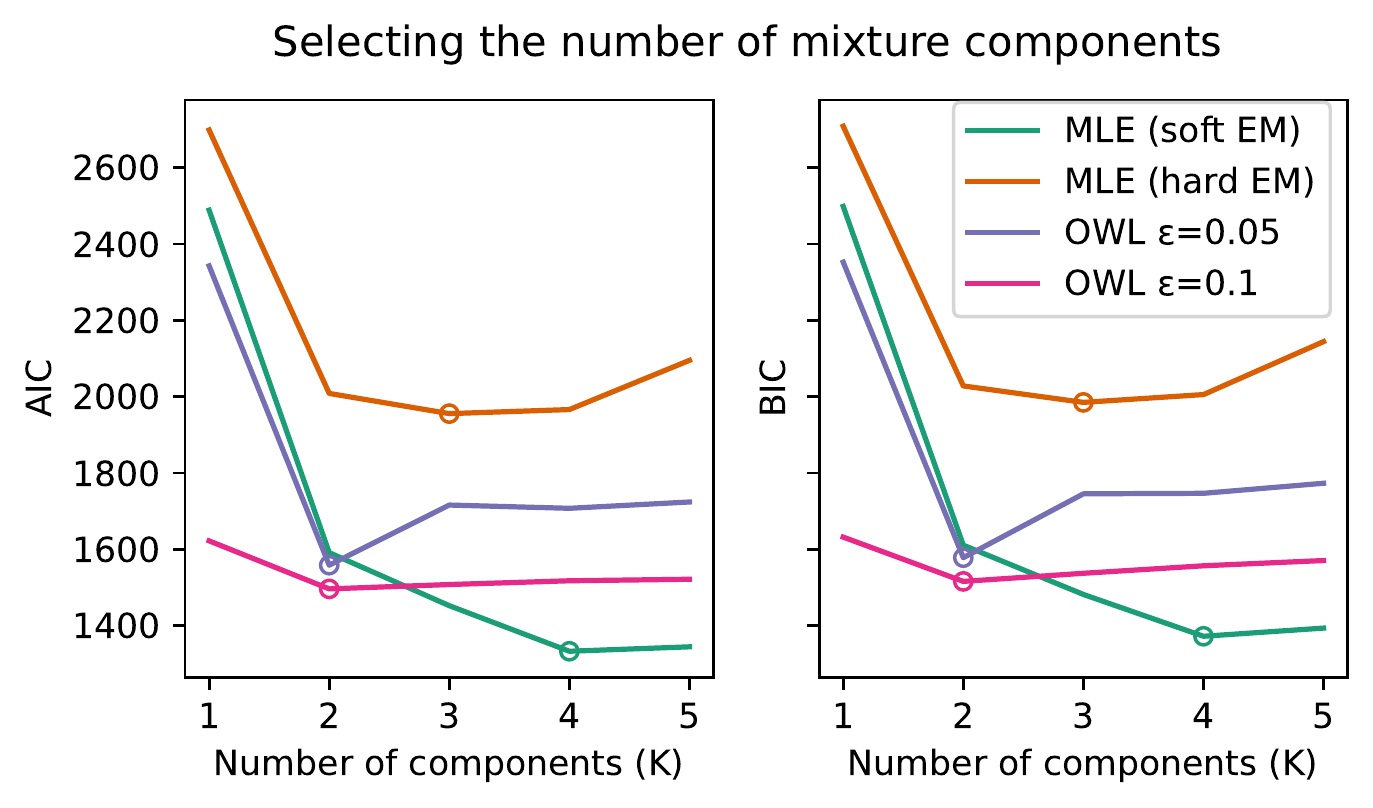}
    \end{subfigure}
    \begin{subfigure}{0.8\textwidth}
    \includegraphics[width=\textwidth]{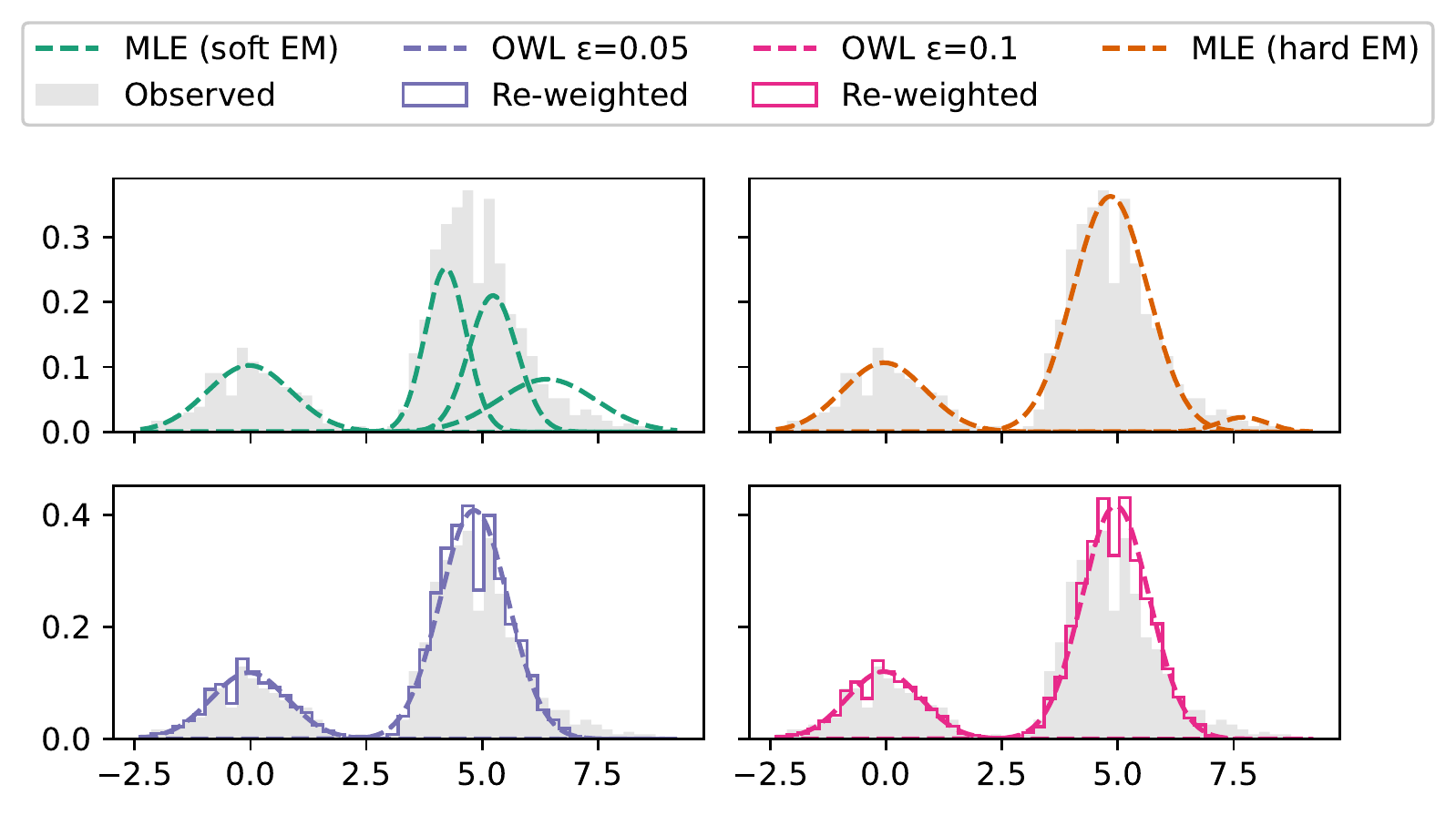}
    \end{subfigure}
    \caption{\new{Brittleness of model selection in mixture models. We fit a Gaussian mixture model using AIC and BIC criteria to  $n=1000$  observations drawn from a  well-separated mixture of a Gaussian and  a non-Gaussian component. The presence of the non-Gaussian component leads the usual AIC or BIC to select more than two components to better match the overall density. But if we modify the selection criteria using optimistic weights from the OWL procedure \eqref{eq:owl-selection-criteria} this seems to prevent the inflation in the number of components. The model fit from the various methods, along with the re-weighted data under OWL are shown.}}
    \label{fig:model-selection-example}
\end{figure}
\FloatBarrier 
\section{Properties of the OKL}
\label{sec:useful-lemmas}

This section covers some technical lemmas about the OKL function that will be used in the theoretical results in the remainder of the appendix. When showing the asymptotic connection between coarsened likelihood and OKL (\Cref{sec:coarsened-likelihood-asymptotics}), we will allow for more general distances than the total variation (TV) distance. To this end, let $\D(P, Q)$ denote a general distance between probability measures $P, Q \in \cP(\cX)$. We define the OKL function for a general distance $\D$ by
\begin{align}
	\label{eqn:okl-general-distance}
	I_\epsilon(\theta) = \inf_{\substack{Q \in  \cP(\cX) \\ \D(Q, P_0) \leq \epsilon}} \KL(Q | P_\theta). 
\end{align}

Although such extended analysis with  general distances might be possible, for simplicity, we will restrict to the case of $\D=\tv$ while proving the sandwiching property of the I-projection (\Cref{sec:sandwiching}) and the consistency of the OKL estimator in continuous spaces (\Cref{sec:okl-estimation-cont}).

\subsection{Continuity of OKL in the coarsening radius}

The following lemma shows when we can expect the OKL function to be continuous in $\epsilon$.
\begin{lemma}
\label{lem:okl-continuity}
Let $0 \leq \epsilon_0 < \epsilon$ and $\alpha > 0$. Suppose the function $Q \mapsto \D(Q, P_0)$ is convex, then
\[ 0 \leq I_{\epsilon}(\theta) - I_{\epsilon+\alpha}(\theta) \leq \frac{\alpha}{\epsilon - \epsilon_0 + \alpha} I_{\epsilon_0}(\theta). \]
If we additionally have $\alpha \leq \epsilon - \epsilon_0$, then 
\[ 0 \leq I_{\epsilon-\alpha}(\theta) - I_{\epsilon}(\theta) \leq \frac{\alpha}{\epsilon - \epsilon_0} I_{\epsilon_0}(\theta). \]
\end{lemma}
\begin{proof}
We will prove the first statement, as the second follows identically.
First observe that we always have $I_{t}(\theta) - I_{t'}(\theta) \geq 0$ for all $t' \geq t \geq 0$. Moreover, if $I_{ \epsilon_0}(\theta)$ is infinite, then the above holds trivially. Thus, we may assume that $I_{\epsilon_0}(\theta) < \infty$. 

Pick $\delta > 0$. By the definition of OKL, for any $r>0$ there exists $Q_{r} \in \cP(\cX)$ such that $\D(Q_{r}, P_0) \leq r$ and 
\[ \KL(Q_{r} | P_\theta) \leq I_{r}(\theta) + \delta. \]
Take $Q = (1-\lambda)Q_{\epsilon + \alpha} + \lambda Q_{\epsilon_0} \in \cP(\cX)$ for $\lambda = \frac{\alpha}{\epsilon - \epsilon_0 + \alpha}$. By the convexity of $\D$, we have
\[ \D(Q, P_0) \leq (1-\lambda) \D(Q_{\epsilon + \alpha}, P_0) + \lambda \D (Q_{\epsilon_0}, P_0) \leq (1-\lambda)(\epsilon + \alpha) + \lambda \epsilon_0 \leq \epsilon.  \]
Moreover, by the convexity of KL divergence (see e.g.~\cite[Lemma~2.4]{budhiraja2019analysis})
\begin{align*}
I_{\epsilon}(\theta) 
\leq \KL(Q | P_\theta)
&\leq (1-\lambda)\KL(Q_{\epsilon + \alpha} | P_\theta) + \lambda \KL( Q_{\epsilon_0} | P_\theta) \\
&\leq (1-\lambda) (I_{\epsilon + \alpha}(\theta) + \delta) + \lambda (I_{\epsilon_0}(\theta) + \delta) \\
&\leq I_{\epsilon + \alpha}(\theta) + \frac{\alpha}{\epsilon - \epsilon_0 + \alpha} I_{\epsilon_0}(\theta) + \delta.
\end{align*}
Rearranging and noting that $\delta > 0$ was chosen arbitrarily, gives us the result in the lemma. The second statement is derived in an identical manner by substituting $\epsilon - \alpha$ in place of $\epsilon$.
\end{proof}

\subsection{Sandwiching property of I-projections}
\label{sec:sandwiching}

In parts of this appendix, we will be dealing with cases where $\Xcal \subseteq \R^d$ and $P_0$ and $P_\theta$ have corresponding densities $p_0, p_\theta \in \Den$ with respect to the Lebesgue measure $\lambda$. 
In this setting, if $I_\epsilon(\theta) < \infty$, then by \cite{csiszar1975divergence} there is a ($\lambda$-almost everywhere) unique density  $\iproj \in \Den$ that we will call the information ($I$-)projection such that $\tv(\iproj, p_0) \leq \epsilon$ and $\KL(\iproj|p_\theta) = I_\epsilon(\theta)$. We will show that $\iproj$ satisfies the following \emph{sandwiching} property relative to $p_0$ and $p_\theta$ for any value of $\epsilon > 0$. 

\begin{definition}
For probability vectors $p, q, r \in \Delta_n$, we say that $r$ is \emph{sandwiched} between $p$ and $q$ if $\min(p_i, q_i) \leq r_i \leq \max(p_i, q_i)$ for all $i=1,\ldots, n$. Similarly, if $p, q, r \in \Den$ are probability densities, then we say that $r$ is sandwiched between $p$ and $q$ if the condition  $\min(p(x),q(x)) \leq r(x) \leq \max(p(x),q(x))$ holds for $\lambda$-almost every $x$.
\end{definition}

The following proposition will be important in proving the sandwiching property for the I-projection.
\begin{proposition}
\label{prop:sandwich-kl-tv}
For probability vectors (or densities), if $r$ is sandwiched between $p$ and $q$, then $\tv(r,p) \leq \tv(q,p)$ and $\KL(r|p) \leq \KL(q|p)$.
\end{proposition}

In fact, will prove the above result for any $\phi$-divergence $D_\phi(p,q) = \int \phi(p/q) q d\lambda$ when $\phi$ is a convex function $\phi(1)=0$. The total variation distance ($\phi(x) = |x - 1|$) and KL-divergence ($\phi(x) = x \log x$) will emerge as special cases.

\begin{lemma}
\label{lem:sandwich-phi-div}
Let $\phi: \R \to (-\infty, \infty]$ be a proper convex function with $\phi(1) = 0$. If a density $r$ is sandwiched between densities $p$ and $q$, then $D_\phi(r,q) \leq D_\phi(p,q)$.
\end{lemma}
\begin{proof}
The sandwiching property implies that there is a function $t: \cX \to [0,1]$ such that $r = (1-t) p + t q$.  Hence	
\begin{align*}
    D_\phi(r,q) &= \int \phi((1-t) p/q + t q/q ) q d\lambda \leq \int (1-t) \phi(p/q) q d\lambda + \int t \phi(1) q d\lambda \\
    &= D_\phi(p,q) - \int t \phi(p/q) q d\lambda \leq D_\phi(p,q) - \xi\int t q(p/q - 1) d\lambda = D_\phi(p,q). 
\end{align*}
where the two inequalities follow from the convexity of $\phi$ noting that there is $\xi \in \mathbb{R}$ (called the sub-gradient) such that 
$\phi(x) \geq \phi(1) + \xi(x-1) = \xi(x-1)$ for all $x \in \R$, and the last equality holds since $\int t(p-q) d\lambda = 0$, since $p$ and $r$ are assumed to integrate to one.
\end{proof}

Now, we will need a simple lemma about total variation distance before we can prove the sandwiching property of the I-projection. For brevity, we will use notations like $\{p \leq q\}$ and $\{q > \max(p_0, p_\theta)\}$ to denote the sets $\{x \in \cX : p(x) \leq q(x)\}$ and $\{x \in \cX: q(x) > \max(p_0(x), p_\theta(x))\}$ respectively. Note that for two densities $p, q \in \Den$, the total variation distance can be expressed as $\tv(p,q) = \int_{p > q} (p-q) d\lambda = \int_{q > p} (q-p) d\lambda = \frac{1}{2} \int |p-q| d\lambda$.

\begin{lemma}
\label{lem:tv-transform}
Let $p,q,r \in \Den$. If $\{r < q\} \subseteq \{p \leq r\}$ or $\{r > q\} \subseteq \{p \geq r\}$ then $\tv(p,r) \leq \tv(p,q)$.  \end{lemma}
\begin{proof} 
	Suppose $\{r < q\} \subseteq \{p \leq r\}$, then
	$$
	\tv(p,r) = \int_{\{p > r\}} (p-r) d\lambda \leq \int_{\{p > r\}} (p-q) d\lambda \leq \int_{\{p > q\}} (p-q) d\lambda = \tv(p,q)
	$$
	where the two inequalities follow from the inclusion  $\{p > r\} \subseteq \{r \geq q\} \cap \{p > q\}$.
	Similarly if $\{r > q\} \subseteq \{p \geq r\}$ then, 
	$$
	\tv(p,r) = \int_{\{ p < r\}} (r-p) d\lambda \leq \int_{\{p < r\}} (q-p) d\lambda \leq \int_{\{p < q\}} (q-p) d\lambda = \tv(p,q)
	$$
	since $\{p < r\} \subseteq \{r \leq q\} \cap \{p < q\}$.	
\end{proof}

\begin{lemma}
\label{lem:info-proj-sandwich}
Let $p_0, p_\theta$ be probability densities satisfying $I_\epsilon(\theta) < \infty$, and let $\iproj$ denote the I-projection of $p_\theta$ onto the set $\{ q \in \Den : \tv(q, p_0) \leq \epsilon \}$.
Then $\iproj$ is sandwiched between $p_0$ and $p_\theta$.
\end{lemma}

\begin{proof}
	
It suffices to show that for any density $q \in \Den$, there is a density $\bar{q} \in \Den$ sandwiched between $p_0$ and $p_\theta$ such that $\tv(\bar{q},p_0) \leq \tv(q, p_0)$ and $\KL(\bar{q}|p_\theta) \leq \KL(q|p_\theta)$. We can then complete the proof by  noting that $\bar{q} = \iproj$ if we set $q = \iproj$. 

Let $q \in \Den$ be given, and suppose that the set $S^+ = \{ q  > \max(p_0,p_\theta) \}$ has non-zero Lebesgue measure. Letting  
\[ v = \int_{S^+} (q - \max(p_0, p_\theta)) \, d\lambda, \]
note that
\[  \tv(q, p_\theta)  = \int_{\{q > p_\theta \}} (q - p_\theta)  \, d\lambda 
\geq v . \]
Define the density
\[ \bar{q}(x) = 
\begin{cases} 
    \max(p_0(x), p_\theta(x)) & \text{ if } x \in S^+ \\   
    q(x) + \frac{v}{\tv(q, p_\theta)}(p_\theta(x) - q(x)) & \text{ if } p_\theta(x) > q(x) \\
    q(x) & \text{ otherwise }
\end{cases}. \]   
Then it is not hard to verify that $\bar{q}$ integrates to one and satisfies $\bar{q}(x) \leq \max(p_0(x), p_\theta(x))$ everywhere. Next,
applying \Cref{lem:tv-transform}
with $p=p_0$, $q=q$ and $r=\bar{q}$, we obtain $\tv(\bar{q}, p_0) \leq \tv(q, p_0)$ since $\{\bar{q} < q\} = S^+ \subseteq \{p_0 \leq \bar{q} \}$ holds.
Additionally, $\bar{q}$ is sandwiched between $p_\theta$ and $q$. Next, \Cref{prop:sandwich-kl-tv} implies that $\KL(\bar{q} | p_\theta) \leq \KL(q | p_\theta)$.

Now let $q \in \Den$ such that $S^+$ is empty but the set $S^- = \{ q < \min(p_0, p_\theta) \}$ is non-empty. Letting 
\[ v = \int_{S^-} (\min(p_0, p_\theta) - q) \, d\lambda, \]
note that $0 < v \leq \tv(q,p_\theta)$, and define the density
\[ \bar{q}(x) = 
\begin{cases} 
    \min(p_0(x), p_\theta(x)) & \text{ if } x \in S^- \\   
    q(x) + \frac{v}{\tv(q, p_\theta)}(p_\theta(x) - q(x)) & \text{ if } p_\theta(x) < q(x) \\
    q(x) & \text{ otherwise }
\end{cases}. \]  
Then observe that $\bar{q}$ is a density and it is sandwiched between $q$ and $p_\theta$. Similar arguments as above using \Cref{lem:tv-transform} and \Cref{prop:sandwich-kl-tv} show that $\tv(\bar{q}, p_0) \leq \tv(q, p_0)$ and $\KL(\bar{q} | p_\theta) \leq \KL(q | p_\theta)$.
\end{proof} 
\section{Quantifying robustness of the OWL functional}
\label{sec:robustness-okl-minimizer}
\new{
This section studies the robustness offered by the population (i.e.~$n \to \infty$) version of our OWL estimator, henceforth called the \emph{OWL functional}.  Similar to the approach used for proving robustness of minimum distance functionals \citep{donoho1988automatic}, here we analyze the \emph{bias-distortion curve} \citep{huber2009,donoho1988automatic} that captures the largest bias in the estimated parameter when the data distribution $\Pop$ is distorted in terms of the metric $\D$ underlying the OWL functional. 

Under regularity assumptions, our two main results in this section show a lower bound of $\epsilon$ on the breakdown point of the \emph{OWL functional}  (\Cref{sec:lower-bound-breakdown}) and a finite upper bound on a notion of  sensitivity under infinitesimal perturbations (\Cref{sec:upper-bound-sensitivity}). By using the Hausdorff metric to calculate distance between sets, our results here handle complexities that arise because the \emph{OWL functional} can output an arbitrary parameter value that minimizes the OKL function, rather than a unique parameter value.

\subsection{Mathematical Setup}
\label{sec:robustness-setup}
Suppose our data space $\cX$ is a Polish space and our parametric model $\{P_\theta\}_{\theta \in \Theta} \subseteq \cP(\cX)$  is a subset of the probability measures $\cP(\cX)$ on $\cX$. We assume that we are provided with a distance $\D$  on $\cP(\cX)$ that captures realistic departures of the true data distribution $P_0$ away from the model family $\{P_\theta\}_{\theta \in \Theta}$.  Following \Cref{sec:okl}, this means that for suitable values of  $\epsilon > 0$ the set of \emph{identifiable parameters} $\Theta_\epsilon(\Pop) \doteq \{ \theta \in \Theta :  \D(\Pop, P_\theta) \leq \epsilon \}$ is suitably small but non-empty.  Formally this happens when $\epsilon > \epsPop \doteq \inf_{\theta \in \Theta} \D(\Pop, P_\theta)$.  Denoting the power set of $\Theta$  by $2^{\Theta}$, we now define our OWL functional.

\begin{definition}
\label{def:owl-functional}
The \emph{OWL functional} $\OF: \cP(\cX) \to 2^{\Theta}$ is defined as: 
$$
\OF(\Pop)=\argmin_{\theta \in \Theta} \OKL[\theta]{\Pop},
$$
where $\OKL[\theta]{\Pop} = \min_{\substack{Q \in \cP(\cX)\\ \D(\Pop, Q) \leq \epsilon}} \KL\left(Q \mid P_\theta\right)$ is the OKL function.  \end{definition}

\begin{remark}  When $\D=\tv$ and $\Pop$ is the unknown data generating distribution, our OWL methodology (\Cref{sec:methodology}) can find a \emph{local} minimizer of the estimator $\hatI$ of the OKL function $\theta \mapsto \OKL[\theta]{\Pop}$. However, in general, the OKL function will have multiple global minimizers (see \Cref{lem:okl-minimizers-characterization} below). The OWL functional $\OF(\Pop)$ denotes the set of all such minimizers.
\end{remark}

Next, when $\D$ is suitably regular (\Cref{ass:metric-condition-for-I-projection}) and $\epsilon \geq \epsPop$, we provide a simple characterization for $\OF(\Pop)$.  \Cref{ass:metric-condition-for-I-projection} is satisfied by many common integral probability metrics  on $\cP(\cX)$, including the TV distance, the Wasserstein distance when the space $\cX$ is bounded, and suitable versions of the Maximum Mean discrepancy \citep{Sriperumbudur2009}.
\begin{assume}
\label{ass:metric-condition-for-I-projection}
The distance $\D$  is a metric on $\cP(\cX)$  such that the ball $\{Q \in \cP(\cX): \D(Q, \Pop) \leq \epsilon\}$  is convex and closed with respect to the TV distance for every $\Pop \in \cP(\cX)$ and $\epsilon > 0$.
\end{assume}

\begin{lemma}
    \label{lem:okl-minimizers-characterization}
    Suppose $\epsilon > \epsPop$ and  \Cref{ass:metric-condition-for-I-projection} holds. Then $\OF(\Pop)$ is equal to the set of identifiable parameters $\Theta_\epsilon(\Pop)\doteq \{ \theta \in \Theta :  \D(\Pop, P_\theta) \leq \epsilon \}$. 
\end{lemma}
\begin{proof}

When $\epsilon > \epsPop$ let us note that $\min_{\theta \in \Theta} \OKL[\theta]{\Pop} = 0$.  Indeed, this follows since $\OKL[\theta]{\Pop}  \geq 0$ for any $\theta \in \Theta$ and we can take  $Q=P_{\theta'}$ for any $\theta' \in \Theta$ such that $\D(P_{\theta'}, P_0) < \epsilon$ to conclude that $\OKL[\theta']{\Pop} = 0$.  

To conclude the proof, we now justify the following equalities:
\[
     \OF(\Pop) = \{\theta \in \Theta : \OKL[\theta]{\Pop} = 0\} = \{\theta \in \Theta : \D(P_\theta, \Pop) \leq \epsilon\} \doteq \Theta_\epsilon(\Pop) 
    \]

Indeed, the first equality follows from \Cref{def:owl-functional} and our result that  $\min_{\theta \in \Theta} \OKL[\theta]{\Pop} = 0$.  The second equality will follow if we can show  $\OKL[\theta]{\Pop} = 0$ if and only if $\D(P_\theta, \Pop) \leq \epsilon$.  The `if' direction follows by taking $Q=P_\theta$ in the definition of the OKL. Let us now show the `only if' direction.  Indeed since \Cref{ass:metric-condition-for-I-projection} is satisfied, we can apply Theorem 2.1 in  \cite{csiszar1975divergence}  to conclude that there is a unique $Q^\theta $  such that $\D(Q^\theta, P_0) \leq \epsilon$ and $\KL(Q^\theta|P_\theta)=\OKL[\theta]{\Pop}=0$. This shows that  $Q^\theta = P_\theta$ and thus  $\D(P_\theta, P_0) \leq \epsilon$.
\end{proof}

\begin{remark} 
\label{rem:minimum-distance-functional}
Suppose \Cref{ass:metric-condition-for-I-projection} holds. Then \Cref{lem:okl-minimizers-characterization} shows that the map $\MD: \cP(\cX) \to 2^{\Theta}$ defined as $\MD(\Pop) \doteq  \cap_{\epsilon > \epsPop} \OF(\Pop)$ is the minimum distance functional $\MD(\Pop) = \argmin_{\theta \in \Theta} \D(P_{\theta}, \Pop)$ from \cite{donoho1988automatic}.
\end{remark}

Following the theory of robustness of minimum distance functionals, we will establish robustness of the  OWL functional by analyzing a suitable \emph{bias-distortion (BD) curve} \citep{donoho1988automatic,huber2009}. In order to define the BD curve for a set-valued functional like $\OF$, we lift a metric on $\Theta$ to the Hausdorff distance on its  subsets $2^{\Theta}$.  
For simplicity, suppose that $\Theta  \subseteq \R^p$  is equipped with the usual Euclidean norm $\|\cdot\|$. 
The Hausdorff distance between sets $A, B \subseteq \Theta$ is defined as  
\begin{equation}
    \label{eq:def-hausdorff}
    \dH(A,B) \doteq \max\{\rho(A|B), \rho(B|A)\}
\end{equation}
where $\rho(C|D) \doteq \sup_{x \in C} \inf_{y \in D} \|x-y\|$ for any $C, D \subseteq \Theta$.

\begin{definition} The \emph{bias-distortion curve} $\maxbias(\cdot|\Pop): \R_+ \to \R_+$ for the OWL functional at $\Pop \in \cP(\cX)$ is defined as:
    \[
    \maxbias(\delta|\Pop) \doteq \sup_{\substack{P \in \cP(\cX) \\ \D(P, \Pop) \leq \delta}} \dH(\OF(P), \OF(\Pop)).
\]
\label{def:bd}
\end{definition}
The value $\maxbias(\delta|\Pop)$ quantifies the  worst-case effect of a perturbation in the input distribution $\Pop$ of order $\delta$ (measured in terms of $\D$) on the output of the OWL functional. 
A similar BD-curve has been used in \citet[Equation 1.23]{huber2009} and \citet[Section 3]{donoho1988automatic} to study robustness of  the asymptotic form of estimators with respect to suitable neighborhoods around the input distribution $\Pop$. However, in contrast to earlier work, our formulation in \Cref{def:bd} uses  the Hausdorff distance because the OWL functional is allowed to output a set of parameter values.

\begin{remark} Similar to \cite{donoho1988automatic}, both the OWL functional (\Cref{def:owl-functional}) and the BD-curve (\Cref{def:bd}) are based on the same distance $\D$. Different choices of $\D$ capture differing forms of robustness. For example:
\begin{enumerate}
    \item If $\D$ is the bounded Lipschitz metric (i.e.~Wasserstein distance based on a bounded metric on $\cX$) then $\maxbias$ characterizes robustness of $\OF$ in terms of weak(-star) neighborhoods around $\Pop$  \cite[see][Section 2.4]{huber2009}. Such neighborhoods provide robustness against both large changes in a few observations (known as \emph{gross errors}) and small changes in many observations (like rounding and grouping) in the data \cite[Chapter 1]{huber2009}.
    \item If $\D$ is the TV distance then $\maxbias(\delta|\Pop)$ bounds the worst case bias of $\OF$ under  \emph{gross errors}  captured by $\delta$-contamination neighborhoods \cite[Section 1.4]{huber2009}. This follows since $\dtv(\Pop, (1-\delta) \Pop + \delta H) \leq \delta$ for any $H \in \cP(\cX)$. 
\end{enumerate}
\end{remark}

As discussed in \cite{donoho1988automatic}, when $\D$ is the Huber's discrepancy or the Prokhorov distance, key robustness indicators of the functional like \emph{gross-error sensitivity}, asymptotic \emph{breakdown}, and qualitative robustness can all be visualized in terms of a corresponding BD-curve \cite[see][Figure 1]{donoho1988automatic}. Following \cite{donoho1988automatic}, we extend the definition of these quantities to the OWL functional by using our BD-curve in \Cref{def:bd}.

\begin{definition} 
\label{def:breakdown}
The \emph{breakdown} point $\BD \in (0,\infty]$ of the OWL functional $\OF$ at $\Pop$ is defined as 
$$
    \BD(\Pop) \doteq \sup\{ \delta > 0 : \maxbias(\delta|\Pop) < \infty\}.
$$
\end{definition}

\begin{definition}
\label{def:sensitivity}
The \emph{sensitivity} $\sensitivity(\Pop) \in (0,\infty]$ of the OWL functional $\OF$ at $\Pop$ is defined as 
$$
    \sensitivity(\Pop) \doteq \limsup_{\delta \downarrow 0} \maxbias(\delta|\Pop)/\delta.
$$
\end{definition}

\begin{definition}
\label{def:qualitative-robustness}
The OWL functional $\OF$ is said to be \emph{robust} at $\Pop$ if 
$$
    \lim_{\delta \to 0} \maxbias(\delta|\Pop) = 0.
$$
\end{definition}

These notions continue to have the same intuition as in \cite{donoho1988automatic}. Namely, we say that the OWL functional is robust at $\Pop$ if small distortions of $\Pop$ in the metric $\D$ do not affect the output value very much. If the OWL functional has a finite \emph{sensitivity} at $\Pop$, then for a small distortion $P$ of $\Pop$, the bias in the output is at most linear in the degree of distortion, i.e.
$$
\dH(\OF(P), \OF(\Pop)) \leq \sensitivity(\Pop) (1+o(1)) \D(P, \Pop). 
$$
Finally, if the breakdown point $\BD(\Pop)$  is far from zero then a distortion of degree at least  $\BD(\Pop)$ is  needed to make the $\OF$ estimator blow-up completely.

We will now use the above notions to establish and discuss robustness properties of the OWL functional.

\subsection{Quantitative Robustness of the OWL functional}
\label{sec:growth-of-bd-curve}

Following the mathematical setup of \Cref{sec:robustness-setup}, we will now discuss and establish robustness of our OWL functional (\Cref{def:owl-functional}) by studying the \emph{bias-distortion} curve in \Cref{def:bd}. Particularly, under suitable regularity conditions, we will show that the OWL functional has (i)  finite sensitivity (i.e. $\sensitivity(\Pop) < \infty$) if $\epsilon > \epsPop$, (ii) a breakdown point $\BD(\Pop)$ of at least $\epsilon$ when $\Pop$ lies in the model family. Property (i) also establishes that the OWL functional is robust in the sense of \Cref{def:qualitative-robustness}.

Throughout this section we will assume that $\Theta$ is a closed subset of $\R^p$. We want to study robustness of the OWL functional when $\Pop=P_{\theta_0}$ is located on the model family and corresponds to the  parameter $\theta_0 \in \Theta$.  

\subsubsection{Lower-bound on the breakdown point of the OWL functional}
\label{sec:lower-bound-breakdown}

In this section, we establish a lower-bound of $\epsilon$ on the breakdown point (\Cref{def:breakdown}) of our OWL functional $\OF$.
It will be illustrative to state our result in terms of the following \emph{gauge} function introduced by \cite{donoho1988automatic} to study the breakdown point of the minimum distance (MD) functional:
$$
b_0(\delta) \doteq \sup\{\|\theta - \theta_0\|: \D(P_\theta, P_{\theta_0}) \leq \delta, \theta \in \Theta\}.
$$
Let $\epsilon^* \doteq \sup\{ \delta > 0 | b_0(\delta) < \infty \}$. It is shown in \cite{donoho1988automatic}  that (i) $\epsilon^*$ is at least as large as the largest breakdown point among all Fisher consistent functionals, and that (ii) the  MD functional always has a breakdown point  of at least $\epsilon^*/2$ .   In comparison, our following lemma shows that the breakdown point of the OWL functional $\OF$ is at least $\epsilon$ as long as $\epsilon$ is less than $\epsilon^*/2$.

\begin{lemma} 

    Suppose  \Cref{ass:metric-condition-for-I-projection} holds and let $\epsilon < \epsilon^*/2$. Then the breakdown point of the OWL functional $\OF$ is at least $\epsilon$, i.e. $\BD(\Pop) \geq \epsilon$.
    \label{lem:owl-breakdown}
\end{lemma}
\begin{proof}
    Consider any $P \in \cP(\cX)$ such that $\D(P, \Pop) < \epsilon$.  Since $\epsPop[P] \leq \D(P,\Pop) < \epsilon$, we have 
    \[
        \OF(P) = \{\theta \in \Theta : \D(P_\theta, P) \leq \epsilon\} \subseteq \{\theta \in \Theta : \D(P_\theta, \Pop) \leq 2\epsilon\}
    \]
where the equality follows from \Cref{lem:okl-minimizers-characterization} and  \Cref{ass:metric-condition-for-I-projection}, while the inclusion follows from the triangle inequality. We further note using \Cref{lem:okl-minimizers-characterization} that $\theta_0 \in \OF(P) \cap \OF(\Pop)$ and $\OF(P) \cup \OF(\Pop)  \subseteq \bar{B}_\Theta(\theta_0, b_0(2\epsilon))$, where $\bar{B}_\Theta(\theta_0, r) = \{\theta \in \Theta: \|\theta - \theta_0\| \leq r\}$ is the closed ball of radius $r$ around $\theta_0$ in $\Theta$. 
This shows that
 \[
 \dH(\OF(\Pop), \OF(P)) \leq b_0(2\epsilon).
 \]
Finally, note that $b_0(2\epsilon) < \infty$ since $\epsilon < \epsilon^*/2$. Since $P \in \cP(\cX)$ is an arbitrary distribution satisfying $\D(P, \Pop) < \epsilon$, we have  in fact shown that $\maxbias(\delta|\Pop) \leq b_0(2\epsilon) <  \infty$ for every $\delta < \epsilon$, allowing us to conclude $\BD(\Pop) \geq \epsilon$.

\end{proof}

\subsubsection{Upper-bound on the sensitivity of the OWL functional}
\label{sec:upper-bound-sensitivity}

We now show that our OWL functional has finite sensitivity (\Cref{def:sensitivity}) under some regularity conditions. Let $\Fpop: \Theta \to [0,\infty)$ denote the mapping  $\Fpop(\theta) = \dist(\Pop, P_\theta)$ and let $\level[\alpha] \doteq \{\theta \in \Theta : \Fpop(\theta) \leq \alpha \}$ be the lower-level set based on level $\alpha \in \R$.

\begin{assume}[Regularity assumption on level-sets] We will assume that there is a $0 < \bar{\delta} < \epsilon$ such that
    \label{ass:assumptions-for-level-set-regularity}
\begin{enumerate}
\item $\Theta$ is a closed subset of $\R^p$ with a non-empty interior $\Theta^{\circ}$,
\item $\level[\epsilon + \bar{\delta}]$ is a compact subset of $\Theta^{\circ}$,
\item $\Fpop$ is continuously differentiable on the set $\level[\epsilon +  \bar{\delta}] \setminus \level[\epsilon - \bar{\delta}]^{\circ}$, and 
$$
\inf\{\|\nabla \Fpop(\theta)\| : \theta \in \level[\epsilon +  \bar{\delta}] \setminus \level[\epsilon - \bar{\delta}]^{\circ}\} > 0.
$$ \label{item:non-zero-derivative-band}
\end{enumerate}
\end{assume}

\begin{remark}
\Cref{ass:assumptions-for-level-set-regularity} will typically be satisfied when $\D=\tv$. For example, suppose that our model consists of a family of densities $\{p_\theta\}_{\theta \in \Theta}$ (with respect to measure $\lambda$ on $\cX$) such that the map $\theta \mapsto p_\theta(x)$ is continuously differentiable at almost every $x$, and the norm of the  gradient is bounded above by an integrable function, i.e. $\|\nabla_\theta p_\theta(x) \| \leq h(x)$ with $\int h(x) d\lambda(x) < \infty$. Additionally, suppose that $\epsilon < \epsilon^*$ from \Cref{sec:lower-bound-breakdown}, so that there is a $\bar{\delta} > 0$ such that $0 < r_0 = b_0(\epsilon + \bar{\delta}) < \infty$. Further suppose that $B(\theta_0, r_0) \subseteq \Theta^{\circ}$ where $B(\theta_0, r)=\{\theta \in \R^p : \|\theta - \theta_0\| < r\}$ is the open ball in $\R^{p}$  with center $\theta_0$ and radius $r$. Then
    \begin{enumerate}
        \item  By Scheffé's lemma, $\Fpop(\theta) = (1/2)\int |p_\theta(x) - p_{\theta_0}(x)| \lambda(dx)$ is a continuous function of $\theta$ because the density $p_\theta(x)$ is continuous in $\theta$ at every $x$.

        \item $\level[\epsilon + \bar{\delta}]$ is a compact subset of $\Theta$ since it is a closed and bounded subset of $\R^p$. Indeed, $\level[\epsilon + \bar{\delta}]$ is closed since $\Fpop$ is continuous and $\Theta$ is assumed to be a closed subset of $\R^p$. Additionally, $\level[\epsilon + \bar{\delta}]$ is bounded since it is contained in $B(\theta_0, r_0)$.
        \item Finally,  suppose that $\lambda(\{x: p_\theta(x) = p_{\theta_0}(x)\}) = 0$ at any $\theta \neq \theta_0$. Then interchanging the order of the derivative and the integral \cite[Theorem 6.28]{klenkeProbabilityTheoryComprehensive2014a} and using the dominated convergence theorem shows that $\Fpop$ is continuously differentiable at $\theta$ whenever $\theta \neq \theta_0$, with 
        $$
        \nabla \Fpop(\theta) = \int \operatorname{sign}(p_\theta(x) - p_0(x)) \nabla_\theta p_\theta(x) d\lambda(x).
        $$
        Thus \Cref{item:non-zero-derivative-band} in \Cref{ass:assumptions-for-level-set-regularity} will be satisfied as long as $\nabla \Psi(\theta) \neq 0$ for every $\theta$ in the compact set $\level[\epsilon +  \bar{\delta}] \setminus \level[\epsilon - \bar{\delta}]^{\circ}$. (Note that $\nabla \Psi(\theta_0) = 0$ since $\theta_0$ minimizes $\Fpop$, but this does not cause a problem since $\theta_0 \in \level[\epsilon - \bar{\delta}]^{\circ}$ as long as $\bar{\delta} < \epsilon$.)
    \end{enumerate}
\end{remark}

The assumed regularity conditions allow us to conclude the following Lipschitz continuity of the level-sets around the level $\epsilon$.

\begin{lemma} Under \Cref{ass:assumptions-for-level-set-regularity} there are finite positive constants $C_1$ and  $\delta_1$ such that
    \[
        \dH(\level[\lambda_1], \level[\lambda_2]) \leq C_1 |\lambda_2 - \lambda_1|
    \]
for any $\lambda_1, \lambda_2 \in (\epsilon - \delta_1, \epsilon + \delta_1)$.
\label{lem:level-set-hausdorff}
\end{lemma}
\begin{proof}
    \Cref{ass:assumptions-for-level-set-regularity} ensures that Lemma 2.1 in \cite{li2021posterior} is satisfied with $f = \Fpop$, $c = \epsilon$, and $\epsilon_1 = \bar{\delta}$. This shows that there are constants $\delta_1, C_1 > 0$ such that for any $\lambda \in (\epsilon - \delta_1, \epsilon + \delta_1)$ and $\theta \in \Theta$ such that $|\Fpop(\theta) - \lambda| \leq 2\delta_1$ it holds that
    \[
        \rho(\theta|\level[\lambda]) \leq \rho(\theta| \{\Fpop = \lambda\}) \leq C_1 |\Fpop(\theta) - \lambda|,
    \]
    where recall $\rho(\theta|A) \doteq \rho(\{\theta\}|A) = \inf_{\theta' \in A} \|\theta - \theta'\|$ for $A \subseteq \Theta$.

    Now we can bound the Hausdorff distance between $\level[\lambda_1]$ and $\level[\lambda_2]$ for $\lambda_1, \lambda_2 \in (\epsilon - \delta_1, \epsilon + \delta_1)$. Without loss of generality, assume $\lambda_1 < \lambda_2$. Then $\rho(\level[\lambda_1]|\level[\lambda_2]) = 0$ since $\level[\lambda_1] \subseteq \level[\lambda_2]$. On the other hand, since $|\lambda_2 - \lambda_1| < 2\delta_1$, the previous result with $\lambda = \lambda_1$ shows that whenever $\theta \in \level[\lambda_2] \setminus \level[\lambda_1]$,  we have
    \[
        \rho(\theta|\level[\lambda_1]) \leq C_1 |\Fpop(\theta) - \lambda_1| \leq  C_1 |\lambda_2 - \lambda_2|.
    \]
    This shows $\rho(\level[\lambda_2]|\level[\lambda_1]) = \sup_{\theta \in \level[\lambda_2] \setminus \level[\lambda_1]} \rho(\theta|\level[\lambda_1]) \leq C_1 |\lambda_2-\lambda_1|$. Now one may use the definition of Hausdorff distance \eqref{eq:def-hausdorff} to complete the proof.

\end{proof}

Now we will establish that the OWL functional has finite sensitivity.

\begin{lemma} 
 \label{lem:bd-small-bounds}
Suppose \Cref{ass:metric-condition-for-I-projection,ass:assumptions-for-level-set-regularity} are satisfied and fixed values $\epsilon > 0$ and $\Pop=P_{\theta_0}$ are given. Then there are finite positive constants $C_1$ and $\delta_1$ such that
    \[
        \maxbias(\delta| \Pop) \leq C_1 \delta \qquad \text{ for each }  \delta \in [0, \min(\epsilon, \delta_1)).
    \]
\end{lemma}
\begin{proof}
    Consider the constants $C_1, \delta_1 > 0$ as defined in  \Cref{lem:level-set-hausdorff}. Choose any $P \in \cP(\cX)$ such that $\D(P,\Pop) \leq \delta$ for $\delta \in (0, \min(\epsilon, \delta_1))$. 
    Since $\epsPop[P] = \inf_{\theta \in \Theta} \D(P, P_\theta) \leq \D(P, \Pop) < \epsilon$, \Cref{lem:okl-minimizers-characterization} shows that $\OF(Q) = \{\theta: \dist(Q, P_\theta) \leq \epsilon \}$ for $Q \in \{P, \Pop\}$. Using the triangle inequality for the metric $\D$ we can note the inclusion
    \[
        \OF[\epsilon-\delta](\Pop) \subseteq \OF(P) \subseteq \OF[\epsilon+\delta](\Pop).
    \] 
This shows that
    \begin{align*}
         \dH(\OF(P), \OF(\Pop)) &\leq \max(\dH(\OF(\Pop), \OF[\epsilon-\delta](\Pop)), \dH(\OF(\Pop), \OF[\epsilon+\delta](\Pop))) \\
                               &\leq C_1 \delta
    \end{align*}
    where the last line follows from \Cref{lem:level-set-hausdorff}. Since $P$ was an arbitrary distribution such that $\dist(P, \Pop) \leq \delta$, the proof is now complete.
\end{proof}

\begin{corollary} Under conditions of \Cref{lem:bd-small-bounds}, the OWL functional has finite sensitivity at $\Pop$ in the sense of \Cref{def:sensitivity}, i.e.
$\sensitivity(\Pop) < \infty$.
\label{cor:bounded-sensitivity}
\end{corollary}

\begin{corollary} Under conditions of \Cref{lem:bd-small-bounds}, the OWL functional is robust at $\Pop$ in the sense of \Cref{def:qualitative-robustness}.
\label{cor:qualitative-robustness}
\end{corollary}
}

\section{Convergence of OKL estimator in discrete spaces}
\label{sec:okl-estimation-finite}
In this section, we will work with a discrete data space but allow for a more general distance function. To this end, let $\D(p, q)$ denote the distance between probability vectors $p, q \in \Delta_\Xcal$. Then the OKL function for a general distance in \cref{eqn:okl-general-distance} translates to this setting as
\begin{equation}
	\label{eq:okle-general-distance-finite}
	 I_\epsilon(\theta) = \inf_{\substack{q \in {\Delta}_{\Xcal} \\ \D(p, p_0) \leq \epsilon}} \KL(q | p_\theta). 
\end{equation}

Given a dataset $x_1, \ldots, x_n \in \Xcal$, the finite approximation to the OKL for general distance $\D$ is given by
\begin{align}
\label{eqn:okl-estimator-general-distance}
\finitehatI(\theta) = \inf_{\substack{q \in {\Delta}_{\Xcal} \\ \supp(q) \subseteq \hX_n \\ \D(q, \pfinite) \leq \epsilon}} \sum_{x \in \hX_n} q(x) \log \frac{q(x)}{p_\theta(x)},
\end{align}
where $\hX_n = \{x_1, \ldots, x_n \}$ is the observed support and $\pfinite(y) = \frac{|\{i \in [n] | x_i = y\}|}{n}$. When $\D(p, q) = \frac{1}{2}\|p - q\|_1$ is the total variation distance, the above is equivalent to the form given in \cref{eq:finiteokle} as shown in the following lemma.
\begin{lemma}
\label{lem:alternate-finite-okl-form}
For any $\epsilon > 0$ and $x_1, \ldots, x_n \in \Xcal$,
\begin{align*}
	\inf_{\substack{q \in {\Delta}_{\Xcal} \\ \supp(q) \subseteq \hX_n \\ \frac{1}{2}\|q - \pfinite\|_1 \leq \epsilon}} \sum_{x \in \hX_n} q(x) \log \frac{q(x)}{p_\theta(x)}  
	\ = \ 
	\inf_{\substack{w \in {\Delta}_{n} \\ \frac{1}{2}\|w - o\|_1 \leq \epsilon}} \sum_{i=1}^n w_i \log \frac{w_i n \pfinite(x_i)}{p_\theta(x_i)}. 
\end{align*}
where $o=(1/n, \ldots, 1/n) \in \Delta_n$.
\end{lemma}
\begin{proof}
For convenience, let $n_i = n \pfinite(x_i)$ and let $n(x) = n \pfinite(x)$. For any $q \in {\Delta}_{\Xcal}$ satisfying $\supp(q) \subseteq \hX_n$, let $w_{q} \in \Delta_n$ satisfy $w_{q,i} = q(x_i)/n_i$. Then we have
\begin{align*}
&\| w_q - o \|_1 = \sum_{i=1}^n \left|\frac{q(x_i)}{n_i} - \frac{1}{n} \right| = \sum_{x \in \hX_n} n(x) \left|\frac{q(x)}{n(x)} - \frac{1}{n} \right| = \| q - \pfinite \|_1 \text{ and } \\
&\sum_{i=1}^n w_{q,i} \log \frac{w_{q,i} n_i}{p_\theta(x_i)} 
= \sum_{i=1}^n \frac{q(x_i)}{n_i} \log \frac{q(x_i)}{p_\theta(x_i)} 
= \sum_{x \in \hX_n} q(x) \log \frac{q(x)}{p_\theta(x)}.
\end{align*}
The above two equalities imply that
\begin{align*}
	\inf_{\substack{q \in {\Delta}_{\Xcal} \\ \supp(q) \subseteq \hX_n \\ \frac{1}{2}\|q - \pfinite\|_1 \leq \epsilon}} \sum_{x \in \hX_n} q(x) \log \frac{q(x)}{p_\theta(x)}  
	\ \geq \ 
	\inf_{\substack{w \in {\Delta}_{n} \\ \frac{1}{2}\|w - o\|_1 \leq \epsilon}} \sum_{i=1}^n w_i \log \frac{w_i n \pfinite(x_i)}{p_\theta(x_i)}. 
\end{align*}

Now for any $w \in \Delta_n$, let $q_w \in {\Delta}_{\Xcal}$ satisfy $q_w(x) = \sum_{i: x_i = x} w_i$. Note that we trivially must have $\supp(q_w) \subseteq \hX_n$. Moreover, we have
\begin{align*}
&\|q_w - \pfinite \|_1 = \sum_{x \in \hX_n} \left| \left( \sum_{i: x_i = x} w_i \right) - \frac{n(x)}{n} \right| 
\leq \sum_{x \in \hX_n} \sum_{i: x_i = x}  \left| w_i - \frac{1}{n} \right|
= \| w - o \|_1 \text{ and } \\
&\sum_{x \in \hX_n} q_w(x) \log \frac{q_w(x)}{p_\theta(x)}
= \sum_{x \in \hX_n} \left( \sum_{i: x_i = x} w_i \right) \log \frac{ \left( \sum_{i: x_i = x} w_i \right)}{p_\theta(x)} 
\leq  \sum_{i=1}^n w_i \log \frac{w_i n_i}{p_\theta(x_i)},
\end{align*}
where the first inequality is the triangle inequality and the second inequality is the log sum inequality~\cite[Theorem~2.7.1]{cover2006elements}.
Together, the above implies that
\begin{align*}
	\inf_{\substack{q \in {\Delta}_{\Xcal} \\ \supp(q) \subseteq \hX_n \\ \frac{1}{2}\|q - \pfinite\|_1 \leq \epsilon}} \sum_{x \in \hX_n} q(x) \log \frac{q(x)}{p_\theta(x)}  
	\ \leq \ 
	\inf_{\substack{w \in {\Delta}_{n} \\ \frac{1}{2}\|w - o\|_1 \leq \epsilon}} \sum_{i=1}^n w_i \log \frac{w_i n \pfinite(x_i)}{p_\theta(x_i)}. \ \ \ \qedhere
\end{align*}
\end{proof}

To prove convergence of the estimator in \cref{eqn:okl-estimator-general-distance}, we will make the following assumptions on the space $\Xcal$ and distance $\D$.
\begin{assume}
\label{assum:finite-continuous-distance}
$\Xcal$ is a discrete set, $\D$ is a metric on $\Delta_{\cX}$ that is jointly convex in its arguments, and there exists a constant $C \geq 1$ for which $\D(p, q) \leq C \|p - q\|_1$ for all $p,q \in \Delta_\Xcal$.
\end{assume}
Observe that \Cref{assum:finite-continuous-distance} fulfills the conditions of \Cref{lem:okl-continuity}, implying that $I_{\epsilon}(\theta)$ is continuous in $\epsilon$. To get quantitative bounds, we make the following assumption.
\begin{assume}
\label{assum:finite-okl}
There exist constants $V, \epsilon_0 > 0$ such that $I_{\epsilon_0}(\theta) \leq V$.
\end{assume}
Finally, we will also require some conditions on the support of $p_\theta$ and $p_0$.
\begin{assume}
\label{assum:lower-bounded-on-support}
There exists a finite set $S \subseteq \Xcal$ and constant $\gamma_0 > 0$ such that $\supp(p_\theta) = S \subseteq \supp(p_0)$ and $p_0(x) \geq \gamma_0$ for all $x \in S$.
\end{assume}
Given the above assumptions, we have the following convergence result for $\finitehatI(\theta)$.

\begin{theorem}
\label{thm:finite-okl-convergence}
Pick $\epsilon > \epsilon_0$ and let $n \geq \max \left\{ \frac{1}{\gamma_0} \log \frac{2|S|}{\delta}, \frac{1}{2} \left( \frac{C |S|}{\epsilon - \epsilon_0} \right)^2 \log \frac{4|S|}{\delta} \right\}$, and suppose \Cref{assum:finite-continuous-distance,assum:lower-bounded-on-support,assum:finite-okl} hold. If $x_1,\ldots,x_n \iid p_0$, then with probability at least $1-\delta$, 
\[ |I_\epsilon(\theta) - \finitehatI(\theta)| \leq  \frac{CV |S|}{\epsilon - \epsilon_0} \sqrt{\frac{1}{2n} \log \frac{2|S|}{\delta}}. \]
\end{theorem}
\begin{proof}
If $n \geq \frac{1}{\gamma_0} \log \frac{2|S|}{\delta}$, then with probability $1-\delta/2$, we have $S \subseteq \hX_n$. Moreover, an application of Hoeffding's inequality tells us that with probability at least $1-\delta/2$, we have
\[ \| \pfinite - p_0 \|_1 \leq |S| \sqrt{\frac{1}{2n} \log \frac{4|S|}{\delta}}. \]
By a union bound, both of these events occur with probability at least $1-\delta$. Condition on these two events occurring.

Now observe that any $q \in \Delta_{\Xcal}$ that achieves $\KL_{\Xcal}(q | p^\theta) < \infty$ must satisfy $q(x) = 0$ for all $x \not \in S$. Thus, we may rewrite the OKL and our finite estimator as
\begin{align*}
I_\epsilon(\theta) &= \inf_{\substack{q \in \Delta_{\Xcal}\\ \supp(q) \subseteq S \\ \D(q, p_0) \leq \epsilon}} \KL(q|p_\theta) \\
\finitehatI(\theta) &= \inf_{\substack{q \in \Delta_{\Xcal}\\ \supp(q) \subseteq \hX_n \cap S \\ \D(q, \pfinite) \leq \epsilon}} \KL(q|p_\theta) = \inf_{\substack{q \in \Delta_{\Xcal}\\ \supp(q) \subseteq S \\ \D(q, \pfinite) \leq \epsilon}} \KL(q|p_\theta),
\end{align*}
where we have used the fact that we are conditioning on $S \subseteq \hX_n$.

Moreover, by \Cref{assum:finite-continuous-distance} and our bound on $\| \pfinite - p_0 \|_1$, we have:
\[ \D(q, \pfinite) \leq \D(q, p_0) + \D(\pfinite, p_0) \leq \D(q, p_0) + C \|\pfinite - p_0\|_1 \leq \D(q, p_0) + \alpha_n,  \]
where $\alpha_n = C |S| \sqrt{\frac{1}{2n} \log \frac{2|S|}{\delta}}$. Similarly, we also can conclude $\D(q, \pfinite) \geq \D(q, p_0) - \alpha_n.$

Applying \Cref{lem:okl-continuity}, we have
\begin{align*}
\finitehatI(\theta) &= \inf_{\substack{q \in \Delta_{\Xcal}\\ \supp(q) \subseteq S \\ \D(q, \pfinite) \leq \epsilon}} \KL(q|p_\theta) 
\geq \inf_{\substack{q \in \Delta_{\Xcal}\\ \supp(q) \subseteq S \\ \D(q, p_0) \leq \epsilon + \alpha}} \KL(q|p_\theta) \\
&= I_{\epsilon + \alpha_n}(\theta) \geq I_{\epsilon}(\theta) - \frac{\alpha_n}{\epsilon - \epsilon_0 + \alpha_n} V.
\end{align*}
On the other hand, if $n \geq \frac{1}{2} \left( \frac{C |S|}{\epsilon - \epsilon_0} \right)^2 \log \frac{4|S|}{\delta}$, then $\alpha_n \leq \epsilon - \epsilon_0$, and we can again apply \Cref{lem:okl-continuity} to see that
\begin{align*}
\finitehatI(\theta) &= \inf_{\substack{q \in \Delta_{\Xcal}\\ \supp(q) \subseteq S \\ \D(q, \pfinite) \leq \epsilon}} \KL(q|p_\theta) 
\leq \inf_{\substack{q \in \Delta_{\Xcal}\\ \supp(q) \subseteq S \\ \D(q, p_0) \leq \epsilon - \alpha}} \KL(q|p_\theta) \\
&= I_{\epsilon - \alpha_n}(\theta) \leq I_{\epsilon}(\theta) + \frac{\alpha_n}{\epsilon - \epsilon_0} V.
\end{align*}
Rearranging the above inequalities gives us the lemma statement.
\end{proof}

\new{
\begin{remark} \label{rem:consistency-for-discrete} When $\supp(p_0)$ is countably infinite, the proof of \Cref{thm:finite-okl-convergence} can suitably be modified to  establish a weaker consistency result for the OKL estimator as long as the support condition $\supp(p_\theta) \subseteq \supp(p_0)$ is satisfied. This allows for commonly used models like Geometric and Poisson distributions on the space $\cX$ of non-negative integers.  The following details can be used:
\begin{enumerate}
    \item Suppose that all observations $x_1, x_2, \ldots, $ are defined on the same space. Then (almost surely) we have $\lim_{n \to \infty} \|\pfinite - p_0\|_1 \to 0$ as $n \to \infty$ by the law of large numbers and Scheff\'e's lemma.
    \item Since \Cref{lem:alternate-finite-okl-form}  does not need $\cX$ to be finite, we have 
$$
\finitehatI(\theta) = \inf_{\substack{q \in \Delta_{\Xcal}\\ \supp(q) \subseteq \hX_n \\ \D(q, \pfinite) \leq \epsilon}} \KL(q|p_\theta) \geq \inf_{\substack{q \in \Delta_{\Xcal}\\ \supp(q) \subseteq \hX_n \\ \D(q, p) \leq \epsilon + \alpha_n}} \KL(q|p_\theta) \geq I_{\epsilon + \alpha_n}(\theta),
$$
where $\alpha_n \doteq \D(p_0,\pfinite) \leq C \|\pfinite - p_0\|_1 \to 0$. This allows us to use \Cref{lem:okl-continuity} to obtain the inequality  $\finitehatI(\theta) \geq I_\epsilon(\theta) - \frac{\alpha_n}{\epsilon - \epsilon_0 +\alpha_n} V$, establishing that $\liminf_{n \to \infty} \finitehatI(\theta) \geq I_\epsilon(\theta)$ almost surely.
    \item To establish a reverse inequality, let $q^* \in \Delta_{\Xcal}$ be such that $\D(q^*,p_0) \leq \epsilon - \delta$ and $\KL(q^*|p_\theta) = I_{\epsilon - \delta}(\theta)$ for a fixed value of $0 < \delta < \epsilon - \epsilon_0$ to be chosen later.  Note that $\supp(q^*) \subseteq \supp(p_\theta)$ since $I_{\epsilon - \delta}(\theta) \leq I_{\epsilon_0}(\theta) < \infty$.
    \item Consider the notation $q(A) \doteq \sum_{x \in A} q(x)$ and $\KL_A(p|q) \doteq \sum_{x \in A} p(x) \log \frac{p(x)}{q(x)}$ for any $p,q \in \Delta_{\cX}$ and $A \subseteq \cX$. Then, we may consider $\hat{q}_n \in \Delta_{\Xcal}$ such that $\hat{q}_n(x) = \frac{q^*(x)}{q^*(\hX_n)} 1_{x \in \hX_n}$ and, similar to \Cref{lem:restrict-to-S}, establish: 
$$
\KL_{\hX_n}(q^*|p_\theta) = (1-q^*(\cX \setminus \hX_n))( \KL(\hat{q}_n|p_\theta) + \log(1-q^*(\cX \setminus \hX_n)))
$$
and $\|\hat{q}_n - q^*\|_1 \leq 2 q^*(\cX \setminus \hX_n)$.  
    \item Note that $\cX_n \uparrow S \doteq \supp(p_0)$ (almost surely) by the law of large numbers. The dominated convergence theorem now shows that $\lim_{n \to \infty} q^*(\cX \setminus \hX_n) = q^*(\cX \setminus S) = 0$ and $\lim_{n \to \infty} \KL_{\hX_n}(q^*|p_\theta) = \KL_{S}(q^*|p_\theta) = \KL(q^*|p_\theta)$. By the triangle inequality, this shows  that $\limsup_{n \to \infty}|\D(\hat{q}_n, \pfinite) - \D(q^*, p_0)| = 0$ since $\D(\hat{q}_n, q^*)$ and $\D(p_0, \pfinite)$ both converge to zero.

    \item Now we can obtain the reverse inequality via the following chain of arguments. For any fixed $0 < \delta < \epsilon - \epsilon_0$, we have shown that there is $\hat{q}_n \in \Delta_{\cX}$ for each $n \in \nat$ with support contained in $\hX_n$ such that $\lim_{n \to \infty} \KL(\hat{q}_n|p_\theta) = \lim_{n \to \infty} \KL_{\hX_n}(q^*|p_\theta) = \KL_{S}(q^*|p_\theta) = I_{\epsilon-\delta}(\theta)$ and $\limsup_{n \to \infty} \D(\hat{q}_n, \pfinite) \leq \D(q^*,p_0) < \epsilon$. This is sufficient to establish that
    $$
    \limsup_{n \to \infty} \finitehatI(\theta) \leq I_{\epsilon-\delta}(\theta).
    $$
    Since the map $\delta \mapsto I_{\epsilon - \delta}(\theta)$ is continuous in a neighborhood of zero, it follows that $\lim_{n \to \infty} \finitehatI(\theta) = I_\epsilon(\theta)$.
\end{enumerate}
\end{remark}}

\section{Convergence of OKL estimator in Euclidean spaces}
\label{sec:okl-estimation-cont}

In this section, we show consistency (and convergence rates) for the OKL estimator when $\cX = \R^d$ is the Euclidean space. To avoid technical complications with tail estimation, we will restrict our analysis to the case when both the data and the model family are supported on a compact set $S \subseteq \cX$. The results and notation in this section are  self-contained, and can be read independently of other sections. 

\subsection{Introduction}
Consider the data space $\cX = \R^d$ equipped with the Euclidean norm $\|\cdot\|$ and its Borel sigma algebra $\cB$. We will let $\Den = \{f: \cX \to [0,\infty) \mid \int f(x) dx = 1, f \text{ is $\cB$-measurable}  \}$ denote the set of densities on $\cX$ with respect to the Lebesgue measure $\lambda$. Given the data density $p_0 \in \Den$, model family $\{p_\theta\}_{\theta \in \Theta} \subseteq \Den$ and coarsening radius $\epsilon > 0$, the central object of interest in this section is the OKL function defined as
\begin{equation}
	\label{eqn:okl}
	I_\epsilon(\theta) = \inf_{\substack{q \in \Den \\
			\tv(q,p_0) \leq \epsilon 
	}} \KL(q|p_\theta),
\end{equation}
where, given two densities $p,q \in \Den$, the total variation distance between them is $\tv(q, p) =\frac{1}{2}\int |q(x)-p(x)|dx$, and the KL-divergence is $\KL(p|q) = \int  p(x) \log \frac{p(x)}{q(x)}  dx$ if the absolute continuity condition $\int p(x) \I{q(x) = 0} dx = 0$ is satisfied, otherwise $\KL(p|q) = \infty$.

Given samples $x_1,\ldots, x_n \in \cX$ drawn i.i.d.~from the distribution with density $p_0$, and a suitable kernel $K_h: \cX \times \cX \to \R$, 
we will approximate $I_\epsilon(\theta)$ with the value of a finite-dimensional optimization problem given by
\begin{equation}
	\label{eqn:okle1}
	\hatI(\theta) = \inf_{\substack{w \in \hat{\Delta}_n \\ \frac{1}{2}\|w-o\|_1 \leq \epsilon}} \sum_{i=1}^n w_i \log \frac{n w_i \hat{p}(x_i)}{p_\theta(x_i)},
\end{equation}
where $\hat{p}$ is a suitable density estimator for $p_0$,
$\Delta_n = \{(v_1,\ldots, v_n)| \sum_{i=1}^n v_i = 1, v_i \geq 0 \}$ is the \ndsimplex , $o=(1/n, \ldots, 1/n) \in \Delta_n$ is the constant probability vector, $A$ is an $n \times n$ matrix with entries $A_{ij} = \frac{K_h(x_i,x_j)}{n \hat{p}(x_i)}$, $\hat{\Delta}_n = A \Delta_n$ is the image of the set $\Delta_n \subseteq \R^n_+$ under the linear operator $A$, and $\|(x_1, \ldots, x_d)\|_1 = \sum_{i=1}^d |x_i|$ is the $\ell_1$ norm of the vector.

\subsubsection{Assumptions and statement of the result}

To describe the formal statement quantifying the approximation between \cref{eqn:okl} and \cref{eqn:okle1}, we will introduce a series of assumptions. Our result will handle the case when densities $p_0$ and $p_\theta$ are smooth densities supported on a bounded set $S \subseteq \cX$ whose boundary has measure zero, i.e.~$\lambda(S \setminus S^\circ) = 0$, where $S^\circ$ denotes the interior of $S$. We fix a value $\theta \in \Theta$ throughout this section.

\begin{assume}
	\label{assump:bounded-support}
	Suppose $S \subseteq \cX$ is a subset of the closed unit ball $\bar{B}(0,R)$ of radius $R$, with finite Lebesgue measure $V_S = \lambda(S) \geq 1$, and has a boundary measure functional $\phi(r) = \frac{\lambda (S \setminus S_{-r})}{\lambda(S)}$ that satisfies $\lim_{r \rightarrow 0} \phi(r) = 0$. 
\end{assume}

Here $S_{-r} = \{x \in S : B(x, r) \subseteq S\}$ denotes the set of points for which the unit ball of radius $B(x,r) = \{y \in \cX | \|x-y\| \leq r\}$ is also contained inside $S$. For example, when $S=B(x_0,r_0)$ is the unit ball of radius $r_0$ centered at the point $x_0$, the boundary measure functional is given by
$$
\phi(r) = 1 - \frac{\lambda(B(x, r_0-r))}{\lambda(B(x, r_0))} = 1 - \left(1 - \frac{r}{r_0} \right)^d
$$
since $S_{-r} = B(x_0, r_0 - r)$. 

Next we have the following smoothness and support condition on densities $p_0, p_\theta \in \Den$. Here the support of a density $p \in \Den$ is defined as $\supp(p) = \{x \in \cX : p(x) > 0\}$.

\begin{assume}
	\label{assump:bounded-densities}
	The supports satisfy $\supp(p_0) = \supp(p_\theta) = S$, and there exists $\gamma \in (0,1]$ such that $p_0(x), p_\theta(x) \in [\gamma, 1/\gamma]$ for all $x \in S$.
\end{assume}

\begin{assume}
	\label{assump:smooth-densities}
	$p_0$ and $\log p_\theta$ are $\alpha$-H\"{o}lder smooth on $S$. More precisely, there are constants $\alpha, C_\alpha > 0$ such that $|p_0(x) - p_0(y)|, |\log p_\theta(x) - \log p_\theta(y)| \leq C_\alpha \|x - y \|^\alpha $ for all $x, y \in S$.
\end{assume}

Finally, we will require the following assumption on the kernel $K_h: \cX \times \cX \to \R$.

\begin{assume}
\label{assump:kernel-properties}
The kernel satisfies $K_h(x, y) = \frac{1}{h^d}\kappa\left( \frac{\| x - y \|_2}{h} \right)$ for a non-increasing continuous function $\kappa: \R_+ \rightarrow \R_+$ satisfying (i) $\int_{\R^d} \kappa(\|x \|_2)\, dx = 1$, (ii) $\kappa(0) \leq c_0$ for some fixed constant $c_0 >0$, and (iii) there exist constants $t_0, C_\rho, \rho$ such that $\rho \leq 1$ and $t_0 \geq 1$ so that for all $t > t_0$, $k(t) \leq C_\rho \exp(-t^\rho)$. Further, suppose that (iv) for every $h > 0$, the kernel $K_h$ is a positive semi-definite kernel, i.e. the $m \times m$ matrix $(K_h(y_i, y_j))_{i,j \in [m]}$ for any $y_1, \ldots, y_m \in \R^d$ and $m \geq 1$ is positive semi-definite. 
\end{assume}

Parts (i)-(iii) in \Cref{assump:kernel-properties} are standard assumptions in the kernel density estimation literature (see, e.g.~\cite{jiang2017uniform}), while part (iv) allows us to use techniques from reproducing kernel Hilbert spaces. By Shoenberg's theorem (see e.g.~\cite{ressel1976short}) part (iv) of \Cref{assump:kernel-properties} is satisfied whenever $\kappa$ is a completely monotone function (see e.g.~\cite{merkle2014completely}). In particular, note that \Cref{assump:kernel-properties} is satisfied by the Gaussian kernel given as $\kappa(t) = \frac{1}{(2\pi)^{d/2}}\exp(-t^2)$.  

\begin{definition} Suppose \Cref{assump:kernel-properties} holds and observations $x_1, \ldots, x_n \iid p_0$ are given. Define a mapping $L: \Delta_n \mapsto \Den$ from the \ndsimplex\ $\Delta_n$ to the space of densities $\Den$ given by $L(w) = \q{w}$, where
	$$
	\q{w}(\cdot) = \sum_{i=1}^n w_i K_h(\cdot, x_i).
	$$
	\label{defn:qhatw}
\end{definition}

For the probability density estimator $\hat{p}$ in the Monte Carlo estimate of the objective in \cref{eqn:okl}, we use the kernel density estimate $\hat{p}=\q{o}$, where $o=(1/n, \ldots, 1/n) \in \Delta_n$.

Using the parameterization $w = A v$ for $v \in \Delta_n$, one can rewrite \cref{eqn:okle1} using the above notation as:
$$
\hatI(\theta) = \inf_{\substack{v \in \Delta_n \\ \frac{1}{2n} \sum_{i=1}^n \left|\frac{\q{v}(x_i)}{\hat{p}(x_i)} - 1\right| \leq \epsilon}} \frac{1}{n}\sum_{i=1}^n \frac{\q{v}(x_i)}{\hat{p}(x_i)} \log \frac{\q{v}(x_i)}{p_\theta(x_i)}.
$$

In order to ensure stability of the optimization objective in \cref{eqn:okle1} and guarantee regularity in the optimal weights, we will restrict the optimization over weights that do not deviate too far from uniformity and replace $\hat{p}$ by the `clipped' version of our estimator: $\hat{p}_\gamma(x) = \min(\max( \hat{p}(x), \gamma), 1/\gamma)$. To that end, define $\Delta_n^r = \{ v \in \Delta_n : n \cdot v_i \in [r,1/r] \text{ for all } i \}$ for $r \leq 1$. Then for our theoretical result, we will consider a version of the estimator 
\begin{align}
	\label{eqn:I-hat-defn}
	\hatI(\theta) = \inf_{\substack{v \in \Delta_n^{\gamma^2/4} \\ \frac{1}{2n}\sum_{i=1}^n \left|  \frac{\q{v}(x_i)}{\hat{p}_\gamma(x_i)} - 1\right| \leq \epsilon}} \frac{1}{n} \sum_{i=1}^n \frac{\q{v}(x_i)}{\hat{p}_\gamma(x_i)} \log \frac{\q{v}(x_i)}{p_\theta(x_i)},
\end{align}
where $\gamma$ is the value from \Cref{assump:bounded-densities}.

Under the above assumptions, we have the following result:

\begin{theorem}
	\label{thm:okl-convergence-formal}
	Suppose \Cref{assump:bounded-support,assump:bounded-densities,assump:smooth-densities,assump:kernel-properties} hold and $n \geq n_0, h \leq \bar{h}, and \frac{\log n}{\sqrt{n} h^d} \leq \bar{\eta}$ for suitable constants $\bar{h}, n_0, \bar{\eta} > 0$ that depend on the quantities in the assumptions. Fix a $\beta > 0$, and suppose we obtain an i.i.d.~sample $x_1, \ldots, x_n$ of size $n$ from density $p_0$. Then with probability at least $1-e^{-(\log n)^{2\beta}}$,
	\[ \left| \hat{I}_\epsilon(\theta) - I_\epsilon(\theta) \right| \leq O\left( \frac{\log n}{\sqrt{n} h^d} + h^{\min(\alpha/2,1)} \log \frac{1}{h} + \phi(\sqrt{h}) \log \frac{1}{\phi(\sqrt{h})} \right),   \]
	where the $O(\cdot)$ hides constant factors depending on the quantities in Assumptions~\ref{assump:bounded-support}-\ref{assump:kernel-properties}, and on the choice of $\epsilon, \beta > 0$.
\end{theorem}

When $S=B(0,R)$ is a Euclidean ball of radius $R$, Taylor expansion shows that $\phi(\sqrt{h}) = \frac{\sqrt{h}d}{R} + O(h)$.  Suppose, for simplicity that $\alpha \geq 1$, then taking $h=n^{-\frac{1}{2d+1}}$ and $\beta=1$, we see that $|\hatI(\theta) - \okl| = \tilde{O}(n^{-\frac{1}{4d+2}})$ upto a logarithmic factor in $n$, with probability at least $1-1/n$. 

Note that the choice of bandwidth parameter  $h=o(n^{-\frac{1}{2d}})$ required to shrink our error terms to zero decreases more slowly with $n$ than the optimal rate $h=o(n^{-1/d})$ for density estimation \citep{rigollet2009optimal}. This suggests that our error bounds can potentially be improved.
 Note also that the   above result holds for any fixed value of  $\theta \in \Theta$; extensions to uniform convergence over all $\theta \in \Theta$ may be possible with further assumptions on the complexity of the class $\Theta$ and continuity of the map $\theta \mapsto \log p_\theta(\cdot)$, but we do not pursue this direction here.

\begin{proof}[Proof outline of \Cref{thm:okl-convergence-formal}]
We sketch the proof here. 
\begin{description}
    \item[\Cref{sec:uniform-convergence}] First, \Cref{lem:uniform-convergence} shows that with high probability for large $n$, the estimators of KL divergence and total variation in \cref{eqn:I-hat-defn} are close to their population counterparts $\KL(\q{v}|p_\theta)$ and $\tv(\q{v}, p_0)$ for any $v \in \Delta_n$ such that $\q{v}$ is suitably bounded. See  \Cref{cor:all-the-bounds} for the most useful version of this statement. 
\item[\Cref{sec:kernele-density-estimation-bounds}] Next, \Cref{lem:existence-of-good-estimator} shows that for any density $q$ over $S$ satisfying Assumption~\ref{assump:bounded-densities}, with high probability for large $n$, there exists a $v \in \Delta_n^\gamma$ such that $\q{v}$ is a pointwise accurate estimator of $K_h \star q$, where $K_h \star q$ denotes the convolution of the density $q$ with the probability kernel $K_h$. Combined with \Cref{lem:info-proj-sandwich}, which establishes the boundedness of the minimizer $\iproj$ in \cref{eqn:okl}, we have that there exists a pointwise accurate estimator $\q{v^*}$ of $K_h \star \iproj$ for some $v^* \in \Delta_n^{\gamma^2/4}$. \item[\Cref{sec:okl-smoothed-approx}] \Cref{lem:smoothed-densities-kl-tv-approx} shows that under \Cref{assump:smooth-densities}, $\KL(K_h \star q | p_\theta)$ and $\tv(K_h \star q, p_0)$ must not be much greater than $\KL(q | p_\theta)$ and $\tv(q, p_0)$, respectively, when $h$ is sufficiently small. \end{description}
We combine the above arguments (see \Cref{sec:proof-of-theorem}), to show that the optimal objective value in \cref{eqn:okl} corresponding to the minimizer $\iproj$, must be close to the optimal objective value in \cref{eqn:I-hat-defn}, thus showing that $\hat{I}_\epsilon(\theta)$ can well approximate $I_\epsilon(\theta)$ when $h$ is small, and $n$ is large.
\end{proof}

While \Cref{sec:proof-of-approx-theorem} covers the main steps of proof, there are useful results in \Cref{sec:useful-lemmmas} that provide upper bounds on the tail probability and second moment of the kernel $K_h$ based on \Cref{assump:kernel-properties}. These results are used in \Cref{sec:okl-smoothed-approx} and \Cref{sec:proof-of-theorem}.

\paragraph{Additional notation}
For $p \in \Den$ and $A \subseteq \cX$, we will use the shorthand $p(A)$ to denote the quantity $\int_A p(x) dx$.
For a  probability density kernel $K_h$, let \[\smooth{p}(\cdot) = (K_h \star p)(\cdot) = \int p(y) K_h(\cdot, y) dy \in \Den \] denote the kernel-smoothed version of a density $p \in \Den$. Similarly, for a measure $\mu$ on $\cX$, we can define $(K_h \star \mu) (x) = \int K_h(x, y) \mu(dy)$. For a set $S \subseteq \R^d$ and $p, q \in \Den$, let $\KL_S(q| p) = \int_S q(x) \log \frac{p(x)}{q(x)} \, dx$. The constant $v_d$ will denote the volume of the Euclidean ball in $\R^d$. The notation $\|f\|_1 = \int |f(x)| dx$ and $\|f\|_{\infty} = \sup_{x \in \cX} |f(x)|$ will denote the $L_1$ and $L_\infty$ norm of the function $f$. For $p, q \in \Den$, recall that $\tv(p,q) = \frac{1}{2}\|p-q\|_1.$

\subsection{Kernel tail bounds}
\label{sec:useful-lemmmas}

In this sub-section, we derive tail bounds for a class of probability kernels $K_h(x,y) = \frac{1}{h^d}\kappa(\|x-y\|/h)$ on $\R^d$ indexed by parameter $h > 0$ used for density estimation, when the function $\kappa: [0, \infty) \to [0,\infty)$ has exponentially decaying  tails (see parts (i)-(iii) of \Cref{assump:kernel-properties}). The following two lemmas bound the tail distribution and the variance of the random variable $Z$ having density $x \mapsto \kappa(\|x\|)$. These results will then be used to obtain tails bounds for the kernel $K_h$.

\begin{lemma} 
	\label{lem:exp-tail-bound}
	Suppose \Cref{assump:kernel-properties} holds. If $Z$ is an $\R^d$ valued random vector with probability density $x \mapsto \kappa(\|x\|)$, then 
	\begin{equation}
		\prob(\|Z\| \geq t) = \int_{\|x\| \geq t} \kappa(\|x\|) dx \leq C_1 t^{d+1-\rho} e^{-t^\rho}
	\end{equation}
	for each $t \geq  t_1 = \max(\Gamma(a+1)^{1/{\rho(a-1)}}, t_0) \geq 1$ where $\Gamma$ is the Gamma function, $a = \frac{d+1}{\rho}$, and $C_1 > 0$ is a constant depending on constants $t_0, C_\rho, \rho > 0$ in  \Cref{assump:kernel-properties} and dimension $d$. Explicitly, $C_1=\rho^{-1}C_\rho v_d 2^{a-1}$, where $v_d = \frac{\pi^{d/2}}{\Gamma(d/2+1)}$ is the volume of the unit ball in $\R^d$.    
\end{lemma}
\begin{proof}
	Using \cite[Lemma~4]{jiang2017uniform} and the upper bound on tails of $\kappa$ since $s \geq t \geq t_0$, 
	\begin{align*}
		\int_{\|x\| \geq t} \kappa(\|x\|) dx = v_d \int_{t}^\infty \kappa(r) r^d dr \leq C_\rho v_d \int_t^\infty e^{-r^\rho} r^{d} dr.
	\end{align*}
	Using the substitution $s = r^\rho$, we obtain
	$$
	\int_t^\infty e^{-r^\rho} r^{d} dr = \frac{1}{\rho}\int_{t^\rho}^\infty e^{-s} s^{\frac{d+1}{\rho} - 1} ds = \frac{1}{\rho}\Gamma\left(a, t^\rho\right)
	$$
	where $\Gamma(a,y) = \int_y^\infty e^{-x} x^{a-1} dx$ is the incomplete Gamma function and $a = \frac{d+1}{\rho} \geq 2$. Taking $b_a = \Gamma(a+1)^{1/(a-1)}$, we have the following bound on the incomplete Gamma function \cite[Theorem~1.1]{pinelis2020exact} for $a \geq 2$:
	$$
	\Gamma(a,y) \leq \frac{((b_a + y)^a - y^a)e^{-y}}{a b_a} = \frac{e^{-y}\int_{0}^{b_a} a (x + y)^{a-1} dx}{a b_a} \leq (b_a + y)^{a-1} e^{-y}.
	$$   
	Since $t^\rho \geq b_a$, the proof can be completed by combining the above displays along with the bound
	$$
	\Gamma(a, t^\rho) \leq 2^{a-1} t^{\rho(a-1)} e^{-t^\rho}  = 2^{a-1} t^{d+1-\rho} e^{-t^\rho}.
	$$
\end{proof}

\begin{lemma}
	\label{lem:bounded-second-moment}
	Suppose \Cref{assump:kernel-properties} holds. If $Z$ is an $\R^d$ valued random vector with probability density $x \mapsto \kappa(\|x\|)$, then
	\[\E_{Z}[\|Z\|_2^2] \leq v_d \left(c_0 t_0^{d+3} + \frac{C_\rho}{\rho} \Gamma\left( \frac{d+3}{\rho}  \right) \right) ,\] 
	where $v_d$ is the volume of the unit ball in $d$-dimensions and $\Gamma(\cdot)$ is the Gamma function.
\end{lemma}
\begin{proof}
	We can write
	\begin{align*}
		\E_{Z}[\|Z \|_2^2] = \int_{\R^d} \|z \|_2^2 \kappa(\|z \|_2) \, dz
		= v_d \int_0^\infty \kappa(t) t^{d+2} \, du,
	\end{align*}
	where the second equality follows from \cite[Lemma~4]{jiang2017uniform}. Taking $c_0, t_0, \rho, C_\rho$ as the constants from Assumption~\ref{assump:kernel-properties}, we have
	\begin{align*}
		\int_0^\infty \kappa(t) t^{d+2} \, du &= \int_{0}^{t_0} \kappa(t) t^{d+2} \, du + \int_{t_0}^{\infty} \kappa(t) t^{d+2} \, du \\
		&\leq c_0 t_0^{d+3} + C_\rho \int_{t_0}^\infty \exp(- t^\rho) t^{d+2} \, du \\
		&\leq c_0 t_0^{d+3} + \frac{C_\rho}{\rho} \int_0^\infty \exp(-u) u^{\frac{d+3}{\rho} - 1} \, du \\
		&= c_0 t_0^{d+3} + \frac{C_\rho}{\rho} \Gamma\left( \frac{d+3}{\rho}  \right),
	\end{align*}
	where we have used the substitution $u = t^\rho$ in the second inequality as well as the fact that the integrands are non-negative from 0 to $\infty$, and the last line follows from the definition of the Gamma function.
\end{proof}

We now use the above lemmas to prove properties about the kernel $K_h$.

\begin{lemma}
	\label{lem:kernel-support-concentration}
	Suppose that \Cref{assump:kernel-properties} holds. For any $x \in \R^d$, 
	\[ 
	\int_{B(x,r)} K_h(x,y) \, dy \geq 1 -  \frac{C_1 k!}{(r/h)^4} 
	\]
	whenever $r \geq t_1 h$, where $B(x,r)$ denotes the open ball of radius $r$ around $x$, constants $C_1, t_1 > 0$ are as defined in \Cref{lem:exp-tail-bound}, and $k = \lceil \frac{5+d}{\rho}\rceil - 1$. 
In particular, for any bounded set $S \subset \R^d$ and any $x \in S$, we have
	\[ \int_S K_h(x,y) \, dy \geq 1 - \frac{C_1 k!}{(r_x/h)^4}  \]
	whenever $r_x \geq t_1 h$, where $r_x = \inf_{y \in S^c} \| x -y \|_2$. 
If $0 < h \leq 
	h_0 = \min(\frac{1}{C_1 k!}, \frac{1}{t_1^2})$, we can always take $r = \sqrt{h}$, and the bound on the right hand side simplifies to $1-\frac{C_1 k!}{(r/h)^4} \geq 1- h$. 
\end{lemma}
\begin{proof}
	Let $Z$ be an $\R^d$-valued random vector with probability density $z \mapsto \kappa(\|z\|_2)$, then 
	\begin{align*}
		1 - \int_{B(x,r)} K_h(x,y) \, dy &= \Pr( x + hZ \notin B(x,r)) \\
		&= \Pr\left(\|Z\|_2 \geq \frac{r}{h}\right) \leq C_1 (r/h)^{(d+1-\rho)}e^{-(r/h)^\rho} \\
		&\leq \frac{C_1 k!}{(r/h)^{k \rho - (d+1-\rho)}} \leq \frac{C_1 k!}{(r/h)^4}.
	\end{align*}
	where the inequality in the second line follows from \Cref{lem:exp-tail-bound} whenever $r \geq t_1 h$, and the inequality on the last line follows by using  $e^y \geq \frac{y^k}{k!}$ and $k \rho - (d+1-\rho) \geq 4$. To get the statement for set $S$, we observe that $1- \int_S K_h(x,y)  \, dy \geq 1 - \int_{B(x,r_x)} K_h(x,y) \, dy$. 
	
\end{proof}

For any probability measure $\mu$ on $\cX$, we define the density $f_{h,\mu}(y) = \int K_h(y, x) \mu(dx) \in \Den$ obtained by convolving $\mu$ with the kernel $K_h$.
The next result bounds the TV-distance $\tv(f_{h,\mu}, f_{h,\nu})$ in terms of the supremum norm $\|f_{h,\mu}-f_{h,\nu}\|_{\infty}$ when measures $\mu$ and $\nu$ are supported on $S$.

\begin{corollary} 
	\label{cor:tv-uniform-estimates}
	Suppose \Cref{assump:bounded-densities,assump:kernel-properties} hold and $\mu$ and $\nu$ are two probability measures supported on $S$. Then whenever $h \leq 1/t_1$,
	$$
	\tv(f_{h,\mu}, f_{h, \nu}) \leq \frac{v_d  (R + 1)^d}{2} \|f_{h,\mu} - f_{h,\nu}\|_\infty  + C_1 h^{d+1-\rho}{e^{-h^{-1/\rho}}}
	$$ where $v_d$ is the volume of the unit ball in $\R^d$, and $t_1$ is as defined in \Cref{lem:exp-tail-bound}.
\end{corollary}
\begin{proof}
We first claim that
\begin{align}
\label{eqn:kernel-smoothed-density-tails}
\int_{\|y\| \geq R+1} f_{h,\mu}(y) dy \leq C_1 (1/h)^{d+1-\rho} e^{-(1/h)^\rho} \text{ for any $h \leq 1/t_1$} .
\end{align}
To see this, we note that $f_{h,\mu}(y) = \int_S K_h(x, y) \mu(dx)$, since $\mu$ is supported on $S$ and $K_h$ is symmetric. Thus by Fubini's theorem
$$
\int_{\|y\| \geq R+r} f_{h,\mu}(y) dy = \int_S \left(\int_{\|y\| \geq R+ r} K_h(x, y) dy \right) \mu(dx).
$$
Next, note that the proof of \Cref{lem:kernel-support-concentration} shows that 
\[ \int_{\|y\| \geq R+ 1} K_h(x, y) dy \leq 1 - \int_{B(x,1)} K_h(x, y) dy \leq C_1 (1/h)^{d+1-\rho} e^{-(1/h)^\rho}\] for any $x \in S$ since $B(x,1) \subseteq B(0,R + 1)$ whenever $x \in S \subseteq B(0,R)$, proving \cref{eqn:kernel-smoothed-density-tails}.

To finish the proof of the lemma, observe that
\begin{align*}
    \tv(f_{h,\mu}, f_{h, \nu}) &=\frac{1}{2} \int_{B(0,R+1)} |f_{h,\mu}(x) - f_{h, \nu}(x)| dx + \frac{1}{2} \int_{\|x\| \geq R+1} |f_{h,\mu}(x) - f_{h, \nu}(x)| dx \\
    &\leq  \frac{v_d}{2}(R+1)^d \|f_{h,\mu} - f_{h, \nu}\|_{\infty} + C_1 h^{d+1-\rho}{e^{-h^{-1/\rho}}}
\end{align*}
where the bound on the second term follows from \cref{eqn:kernel-smoothed-density-tails}.
\end{proof}

Next, we can see how well the convolved density $\smooth{q} = (K_h \star q) \in \Den$ approximates the density $q \in \Den$. For instance, if the density $q$ is supported on a bounded set $S \subseteq \R^d$ with non-zero Lebesgue measure $V_S$, with boundary functional $\phi(r) = \frac{\lambda(S \setminus S_{-r})}{\lambda(S)}$ (see \Cref{assump:bounded-support}), then the following Lemma provides an upper bound on the mass placed by the convolved density  $\smooth{q}$ outside the set $S$.

\begin{lemma}
	\label{lem:smooth-support-concentration}
	Suppose that \Cref{assump:bounded-support} and \Cref{assump:kernel-properties} hold. For any probability density $q \in \Den$ and density bounded above by $1/\gamma$ and supported on the set $S$,
	\[ \smooth{q}(S^c) \doteq \int_{S^c} \smooth{q}(x) dx \leq  h + \frac{V_S}{\gamma} \phi(\sqrt{h}) \]
	whenever $h \leq h_0$, where $h_0$ is as given in \Cref{lem:kernel-support-concentration}.
\end{lemma}
\begin{proof}
	We use the following chain of arguments:
	\begin{align*}
		\smooth{q}(S^c)  &= \int_{S^c} \smooth{q}(x) dx = \int_{S^c} \int_S q(y) K_h(x, y) dy dx  \\
		&= \int_S q(y) \int_{S^c} K_h(x,y) \, dx \, dy \\
		&= \int_{S_{-\sqrt{h}}} q(y) \int_{S^c} K_h(x,y) \, dx \, dy + \int_{S\setminus S_{-\sqrt{h}}} q(y) \int_{S^c} K_h(x,y) \, dx \, dy \\
		&\leq h + \int_{S\setminus S_{-\sqrt{h}}} q(y) \, dy
		\leq h + \frac{V_S}{\gamma} \phi(\sqrt{h}),
	\end{align*}
	where the first inequality uses $\int_{S^c} K_h(x,y) dx = 1 - \int_{S} K_h(y,x) dx$ is bounded above by $h$ if $y \in S_{-\sqrt{h}}$  (\Cref{lem:kernel-support-concentration}) or by one if $y \in S \setminus S_{-\sqrt{h}}$, and the second inequality holds from the definition of boundary functional $\phi$ and the upper bound on the density $q$.
\end{proof}

For a set $S \subseteq \R^d$, let $\KL_S(q| p) = \int_S q(x) \log \frac{p(x)}{q(x)} \, dx$. The following lemma provides a useful trick to translate between $\KL$ and $\tv$ expressions for densities that are not supported on $S$ to densities that are supported on $S$.

\begin{lemma}
	\label{lem:restrict-to-S}
	Suppose $q \in \Den$ is a density that may not be supported on $S$. Then there is a density $\bar{q}$ supported on $S$ such that
	$$
	\KL_S(q|p_\theta) = (1-q(S^c))(\KL(\bar{q}|p_{\theta}) + \log (1-q(S^c)))
	$$
	and $\tv(\bar{q}, p_0) \leq \tv(q, p_0) + q(S^c).$
\end{lemma}
\begin{proof} We will take $\bar{q}(x) = \frac{q(x)}{q(S)} \I{x \in S}$ to be the restriction of the density $q$ to $S$. It is straightforward to see that the  equality for the KL terms hold. The TV bound follows from the triangle inequality, noting that $\tv(\bar{q}, q) = \frac{1}{2} \int_{\R^d} \abs{\bar{q}(x) - q(x)} dx = q(S^c)$.
\end{proof}

\subsection{Key Lemmas}
\label{sec:proof-of-approx-theorem}

\subsubsection{Uniform convergence of KL and TV estimators}
\label{sec:uniform-convergence}

The main result of this subsection is the following.

\begin{lemma}
	\label{lem:uniform-convergence}
	There exists an absolute constant $c_1 > 0$ such that the following holds. Suppose observations $x_1, \ldots, x_n$ are drawn i.i.d. from $p_0$, and that \Cref{assump:bounded-densities} and \Cref{assump:kernel-properties} hold. Let $\q{w}$ be as defined in \Cref{defn:qhatw}. Let $\hat{p} = \q{o}$, where $o=(1/n, \ldots, 1/n) \in \Delta_n$, be the kernel-density estimator for $p_0$  based on observations $x_1, \ldots x_n$, and let $\hat{p}_\gamma(\cdot) = \min(\max( \hat{p}(\cdot), \gamma), 1/\gamma)$ be its truncated version. With probability at least $1-\delta$, we have
	\begin{align*}
		\left| \left(\| \q{w} - p_0 \|_1 - \q{w}(S^c)\right) - \frac{1}{n} \sum_{i=1}^n \left| \frac{\q{w}(x_i)}{\hat{p}_\gamma(x_i)} - 1 \right| \right| 
		&\leq \frac{c_1 u}{\gamma^3} \left(\frac{c_0}{\sqrt{n} h^{d}} + \sqrt{\frac{2}{n}\log \frac{6}{\delta}} + \tv(\hat{p}, p_0) \right) \hspace{1em} \text{ and }  \\
		\left| \KL_S(\q{w} | p_\theta) - \frac{1}{n} \sum_{i=1}^n \frac{\q{w}(x_i)}{\hat{p}_\gamma(x_i)} \log \frac{\q{w}(x_i)}{p_\theta(x_i)} \right| 
		&\leq  \left(\frac{c_1 u}{\gamma^3} \log \frac{u}{\ell \gamma}\right) \left(\frac{c_0}{\sqrt{n} h^d} + \sqrt{ \frac{2}{n}\log \frac{6}{\delta}} + \tv(\hat{p}, p_0) \right).
	\end{align*}
	uniformly over all $w \in \Wcal_{\ell, u} = \{ w \in \Delta_n \, : \,  \q{w}(x_i) \in [\ell, u] \text{ for all } i \in [n] \}$, where thresholds $u, l$ are chosen so that $u\geq \max(1, \gamma e) \geq 1 \geq \ell$ and $\|f\|_1 = \int |f(x)| dx$.
\end{lemma}

To establish this result, we will use techniques from uniform law of large numbers over a class of functions $\Fc \subseteq \{ f: S \rightarrow \R \}$. For a fixed class $\Fc$, given i.i.d. samples $x_1, \ldots, x_n$ from density $p_0$ supported on $S \subseteq \cX$, we consider the empirical Rademacher complexity (e.g.~\cite{boucheron2005theory}) of the function class $\Fc$ defined as:
\begin{equation}
	\Rad(\Fc) = \E_{\epsilon} \left[ \sup_{f \in \Fc} \left| \frac{1}{n} \sum_{i=1}^n \epsilon_i f(x_i)\right| \right]
\end{equation}
where the expectation is over Rademacher random variables $\epsilon_1, \ldots \epsilon_n \iid \operatorname{Uniform}(\{-1,1\})$. The Rademacher complexity is a measure of richness of a function class that can be used to obtain a uniform law of large numbers result like the following standard result.

\begin{lemma}[c.f. Theorem~3.2 of \citep{boucheron2005theory}]
	\label{lem:bbl-rademacher-bounds}
	Let $\Fc \subset \{ f: S \rightarrow [-a,a] \}$ be a fixed class of bounded functions and suppose $\delta > 0$. If $x_1, \ldots, x_n \sim p_0$, then with probability at least $1-\delta$
	\[ 
	\sup_{f \in \Fc}\left| \frac{1}{n} \sum_{i=1}^n f(x_i) - \int f(x) p_0(x) dx \right| \leq 2 \Rad(\Fc) + a\sqrt{ \frac{2}{n} \log \frac{2}{\delta}}. 
	\]
\end{lemma}

Thus a typical task to obtain uniform laws of large numbers is to obtain upper bounds on $\Rad(\Fc)$ for various function classes $\Fc$. To this end, the following results will be useful: 

\begin{lemma}[c.f. Theorem~3.3 of \citep{boucheron2005theory}]
\label{lem:bbl-rademacher-lipschitz}
Suppose $g: S \rightarrow \R$ is a bounded function, $\Fc \subset \{ f: S \rightarrow \R \}$, and $\psi: \R \rightarrow \R$ is a function with Lipschitz constant $L$. Then
\begin{align*}
    \Rad( \{ x \mapsto g(x) f(x) : f \in \Fc \}) &\leq \Rad(\Fc) \sup_{x \in S} |g(x)| \\
    \Rad( \{ x \mapsto  \psi(f(x) ): f \in \Fc \}) &\leq L  \Rad(\Fc) + \frac{|\psi(0)|}{\sqrt{n}}.
\end{align*}
\end{lemma}
\begin{proof}
	The above statements with $\psi(0)=0$ follow from Theorem~3.3 of \citep{boucheron2005theory}. To handle the case $\psi(0) \neq 0$, define $\tilde{\psi} = \psi - \psi(0)\boldsymbol{1}$ where $\boldsymbol{1}$ denotes the function taking the constant value one, and note that 
	\begin{equation*}
		\begin{aligned}
			\Rad( \{ \psi \circ f : f \in \Fc \}) &= \Rad( \{ \tilde{\psi} \circ f + \psi(0) \boldsymbol{1} : f \in \Fc \}) \\
			&= \Rad( \{ \tilde{\psi} \circ f : f \in \Fc \}) + \Rad(\{\psi(0)\boldsymbol{1}\}) \\
			&\leq L \Rad(\Fc) + |\psi(0)| \Rad(\boldsymbol{1})
		\end{aligned}
	\end{equation*}
	where we have used that $\Rad(\{\boldsymbol{1}\}) \leq \E_\epsilon \abs{\frac{1}{n}\sum_{i=1}^n \epsilon_i} \leq \sqrt{\E_\epsilon \left(\frac{1}{n}\sum_{i=1}^n \epsilon_i \right)^2} = n^{-1/2}.$
\end{proof}

\begin{lemma}[Lemma~22 of \citep{bartlett2002rademacher}]
	\label{lem:kernel-rademacher-bounds}
	Let $k : \cX \times \cX \rightarrow \R$ be a positive definite kernel satisfying $k(x,x)\leq c$ for all $x \in \cX$. Then, for the function class 
\[ \Fc_a = \bigg\{ \cdot \mapsto \sum_{i=1}^m w_i k(\cdot, y_i) \, : \, m \geq 1, y_i \in \cX, w_i \in \R, \text{ and } \sum_{i,j} w_i w_j k(y_i, y_j) \leq a^2 \bigg\},  \]
	we have $\Rad(\Fc_a) \leq 2a \sqrt{\frac{c}{n}}$. \end{lemma}

We will now use the above results to bound the Monte Carlo estimation error uniformly over functions generated from a data-dependent function family $\{\q{w} :  w \in \Delta_n \}$.

\begin{lemma}
	\label{lem:uniform-bounds-over-qclass}
	Suppose bounded functions $\alpha, \beta, \eta : S \to \R$ and a Lipschitz function $\Phi: \R \to [-M, M]$ with Lipschitz constant $L$ are given, and suppose that \Cref{assump:kernel-properties} holds. Given a function $q: S \to \R$, denote $\psi(q, x) = \eta(x) \Phi(\alpha(x)q(x) + \beta(x))$. Then for any $\delta \in (0, 2/e)$, with probability $1-\delta$ it holds that:
	$$
	\sup_{w \in \Delta_n} \abs*{\frac{1}{n}\sum_{i=1}^n \psi(\q{w}, x_i) - \int \psi(\q{w}, x) p_0(x) dx } \leq \frac{C_1}{h^d\sqrt{n}} + C_2 \sqrt{\frac{2}{n}\log \frac{2}{\delta}}
	$$ 
	where $\q{w}$ is as in \Cref{defn:qhatw}, $C_1 = 4\|\eta\|_\infty \|\alpha\|_\infty L c_0$, and  $C_2 = \sqrt{2}\|\eta\|_\infty(L \|\beta\|_\infty + |\Phi(0)|) + \|\eta\|_\infty M$.
\end{lemma}
\begin{proof}
	Using $\hat{\Qcal}$ to denote the data dependent class of functions $\{\q{w} \mid w \in \Delta_n \}$, and taking $\kappa = K_h$, $c=c_0/h^d$, and $a=\sqrt{c}$ in \Cref{lem:kernel-rademacher-bounds}, let us first note that $\hat{\Qcal} \subseteq \Fc_a$. Indeed, this holds since for any $w = (w_1, \ldots, w_n) \in \Delta_n$ we have $\q{w}(\cdot) = \sum_{i=1}^n K_h(\cdot, x_i) w_i$ and 
	$$
	\sum_{i,j=1}^n w_i w_j K_h(x_i,x_j) \leq c \sum_{i=1}^n \sum_{j=1}^n w_i w_j = c, 
	$$
	where the inequality follows from the Cauchy Schwarz bound $K_h(x, y) \leq \sqrt{K_h(x, x) K_h(y, y)}$ using the positive definiteness of the kernel $K_h$, and the bound $K_h(x,x) \leq c_0 /h^d = c$ from \Cref{assump:kernel-properties}. Further, note by \Cref{lem:kernel-rademacher-bounds} that $\Rad(\Fc_a) \leq \frac{2c_0}{h^d\sqrt{n}}$.
	
	Next, let $\Gc$ denote the fixed function class $\{x \mapsto \psi(f, x) | f \in \Fc_a\}$. Repeated application of \Cref{lem:bbl-rademacher-lipschitz} shows:
	\begin{align*}
		\Rad(\Gc) &\leq  \|\eta\|_\infty \|\alpha\|_\infty L \Rad(\Fc_a) + \frac{ \|\eta\|_\infty(L \|\beta\|_\infty + |\Phi(0)|)}{\sqrt{n}}\\
		&\leq \frac{2\|\eta\|_\infty \|\alpha\|_\infty L c_0}{h^d\sqrt{n}} + \frac{ \|\eta\|_\infty(L \|\beta\|_\infty + |\Phi(0)|)}{\sqrt{n}}.
	\end{align*}
	
	Thus we can complete the proof by an application of \Cref{lem:bbl-rademacher-bounds} for the bounded class of functions $\Fc = \Gc$, and noting the inclusion $\hat{\Qcal} \subseteq \Fc_a$.
\end{proof}

With these uniform convergence results in place, now we can prove the main result of this section:

\begin{proof}[Proof of Lemma~\ref{lem:uniform-convergence}]
	By assumption
	\[ \frac{1}{n} \sum_{i=1}^n \left| \frac{1}{p_0(x_i)} - \frac{1}{\hat{p}_\gamma(x_i)} \right| \leq \frac{1}{\gamma^2 n} \sum_{i=1}^n \left|p_0(x_i) - \hat{p}_\gamma(x_i)\right| \leq \frac{1}{\gamma^2 n} \sum_{i=1}^n (\left|p_0(x_i) - \hat{p}(x_i)\right| \wedge \gamma^{-1}), 
	\]
	where the last inequality uses the fact that $\abs{p_0(\cdot) - \hat{p}_\gamma(\cdot)} \leq \min(1/\gamma, \abs{p_0(\cdot) - \hat{p}(\cdot)})$ since $p_0(\cdot), \hat{p}_\gamma(\cdot) \in [1/\gamma, \gamma]$. Recall $\hat{p} = \q{o}$, and hence we may apply \Cref{lem:uniform-bounds-over-qclass} with $\Phi(x)=|x|\wedge \gamma^{-1}$, $\beta = p_0$, $\alpha(\cdot)= \eta(\cdot) = 1$, and $w=o$, to show with probability $1-\delta/3$:
	\begin{align*}
		\frac{1}{n}\sum_{i=1}^n \left|p_0(x_i) - \hat{p}(x_i)\right|\wedge \gamma^{-1}
		&\leq \int (\gamma^{-1} \wedge |p_0(x) - \hat{p}(x)|) p_0(x) dx + \left(\frac{4 c_0 }{ h^d \sqrt{n}} +   \frac{3}{\gamma}\sqrt{\frac{2}{n} \log \frac{6}{\delta}}\right) \\
		&\leq \frac{2}{\gamma}\tv(\hat{p}, p_0)  + \left(\frac{4 c_0 }{ h^d \sqrt{n}} +  \frac{3}{\gamma}\sqrt{\frac{2}{n} \log \frac{6}{\delta}}\right) = \frac{\nu}{\gamma},
	\end{align*}
	where $\nu = 2\tv(\hat{p}, p_0) + \frac{4 \gamma c_0 }{h^d \sqrt{n}} +  3\sqrt{\frac{2}{n}  \log \frac{6}{\delta}}$.

Turning to the $\ell_1$-bound, for any $w \in \Wcal_\gamma$, we have
	\begin{align*}
		\left| \frac{1}{n} \sum_{i=1}^n \left| \frac{\q{w}(x_i)}{\hat{p}_\gamma(x_i)} - 1 \right| - \frac{1}{n} \sum_{i=1}^n \left| \frac{\q{w}(x_i)}{{p}_0(x_i)} - 1 \right| \right|
		&\leq \frac{1}{n} \sum_{i=1}^n \left| \left| \frac{\q{w}(x_i)}{\hat{p}_\gamma(x_i)} - 1 \right| - \left| \frac{\q{w}(x_i)}{{p}_0(x_i)} - 1 \right| \right| \\
		&\leq \frac{1}{n} \sum_{i=1}^n \left| \frac{\q{w}(x_i)}{\hat{p}_\gamma(x_i)} -  \frac{\q{w}(x_i)}{{p}_0(x_i)} \right| \leq  \frac{\nu u}{\gamma^3}.
	\end{align*}
	Here the first inequality follows from the triangle inequality, the second follows from the reverse triangle inequality, and the last follows from the upper bound on $\q{w}$ and the bounds at the beginning of the proof. 
Observe that $\int \abs{\frac{q(x)}{{p_0}(x)} - 1} p_0(x) dx = \| q - p_0 \|_1 - q(S^c)$ for any density $q \in \Den$. 
Next, we apply \Cref{lem:uniform-bounds-over-qclass} with $\Phi(x) = |x| \wedge M$, $M = u/\gamma$, $\alpha = 1/p_0$, $\beta = -1$, $\eta = 1$, noting that:
	$$
	\psi(\q{w}, \cdot) = \abs*{\frac{\q{w}(\cdot)}{p_0(\cdot)} - 1} \wedge M  = \abs*{\frac{\q{w}(\cdot)}{p_0(\cdot)} - 1}
	$$ 
	for each $w \in \Wcal_{\ell, u}$. Thus we have
	\begin{align*}
		\left|\frac{1}{n} \sum_{i=1}^n \left| \frac{\q{w}(x_i)}{{p_0}(x_i)} - 1 \right| - \| \q{w} - p_0 \|_1  - \q{w}(S^c)  \right| \leq  \frac{4c_0}{\gamma \sqrt{n} h^{d}} + \left(3 + \frac{u}{\gamma}\right)\sqrt{ \frac{2}{n} \log \frac{6}{\delta}},
	\end{align*}
	with probability at least $1-\delta/3$ for all $w \in \Wcal_{\ell, u}$ simultaneously.

	Turning to the KL-bound, a similar chain of reasoning gives us
	\begin{align*}
		\left| \frac{1}{n} \sum_{i=1}^n \frac{q_w(x_i)}{\hat{p}(x_i)} \log \frac{q_w(x_i)}{p_\theta(x_i)} -  \frac{1}{n} \sum_{i=1}^n \frac{q_w(x_i)}{p_0(x_i)} \log \frac{q_w(x_i)}{p_\theta(x_i)}   \right|
		\leq \frac{\nu u}{\gamma^3} \log \frac{u}{\gamma},
	\end{align*}
	where we have additionally used the lower bound $p_\theta(\cdot) \in [\gamma, 1/\gamma]$. Next, we will invoke \Cref{lem:uniform-bounds-over-qclass} with $\alpha=1/p_\theta, \beta = 0, \eta = p_\theta/p_0$ and $\Phi (x) = k(x) \log k(x)$ where $k(x) = \min(\max(x, u/\gamma), l \gamma)$, noting that the function $\Phi$ is bounded within radius $M=\frac{u}{\gamma} \log \frac{u}{\gamma}$ and is a Lipschitz function with Lipschitz constant bounded by $L = \sup_{x \in [l\gamma, u/\gamma]} |1 + \log x| \leq 2\log \frac{u}{l \gamma}$. Then, using the bounds $C_1 \leq \tilde{C}_1 = \frac{8c_0}{\gamma^3} \log \frac{u}{l\gamma}$ and $C_2 \leq \tilde{C}_2 = \frac{3u}{\gamma^3} \log \frac{u}{\gamma}$, \Cref{lem:uniform-bounds-over-qclass} shows that with probability $1-\delta/3$,  
	$$
	\abs*{\frac{1}{n}\sum_{i=1}^n \frac{\q{w}(x_i)}{p_0(x_i)} \log \frac{\q{w}(x_i)}{p_\theta(x_i)} - \KL_S(\q{w} | p_\theta)} \leq \frac{\tilde{C_1}}{\sqrt{n} h^d} + \tilde{C_2} \sqrt{\frac{2}{n}\log \frac{6}{\delta}}
	$$
	for each $w \in \Wcal_{\ell, u}$ simultaneously, since $\psi(\q{w}, \cdot) = \frac{\q{w}(\cdot)}{p_0(\cdot)} \log \frac{\q{w}(\cdot)}{p_\theta(\cdot)}$ and $\KL_S(\q{w}|p_\theta) =$ \\ $\int \psi(\q{w}, x) p_0(x) dx$ whenever $w \in \Wcal_{\ell, u}$.
	Finally, we obtain the lemma statement by using the union bound and putting all the display equations together, noting that $u/\gamma \geq e$ and $\gamma \leq 1$.
\end{proof}

\subsubsection{Sup-norm approximation of kernel density estimates}
\label{sec:kernele-density-estimation-bounds}

Recall our notation for the truncated probability simplex $\Delta_n^\beta = \{(w_1, \ldots, w_n) \in \Delta_n : \frac{\beta}{n} \leq w_i \leq \frac{1}{n \beta}\}$. The next lemma provides approximation results for the subset of densities $\{\q{w} : w \in \Delta_n^{\gamma^2/4}\} \subseteq \Den$ that hold with arbitrary high probability when $n$ is large enough. In particular, Item 1 shows that $\hat{p}$ approximates well the convolved density $K_h \star p_0$, Item 2 provides upper and lower bounds on the density $\q{w}$, Item 3 bounds the mass of $\q{w}$ outside the set $S$, and Item 4 shows that any convolved density of the form $K_h \star q$ can be approximated well by a density of the form $\q{w}$, whenever $q(\cdot) \in [\gamma, 1/\gamma]$ is a density supported on $S$.

\begin{lemma}
	\label{lem:existence-of-good-estimator}
	Suppose that \Cref{assump:bounded-densities} and \Cref{assump:kernel-properties} hold and $\delta \in (0,1)$ is such that $n \geq \frac{2}{\gamma^4} \log \frac{12}{\delta}$. Further suppose that $q \in \Den$ is a density supported on $S$ satisfying $\gamma \leq q(x)  \leq 1/\gamma$ for each $x \in S$. If $x_1, \ldots, x_n \sim p_0$, with probability $1-\delta$, the following items hold:
	\begin{enumerate}
		\item  $\|\hat{p} - K_h \star p_0\|_{\infty} \leq \frac{6c_0}{\sqrt{n} h^d} \sqrt{\log \frac{6}{\delta}}$. \item  $\frac{\gamma^2}{4} \left( \gamma - \frac{6c_0}{\sqrt{n} h^d} \sqrt{\log \frac{2}{\delta}} \right)\leq \q{w}(x) \leq \frac{4}{\gamma^2} \left( \gamma^{-1} + \frac{6c_0}{\sqrt{n} h^d} \sqrt{\log \frac{2}{\delta}} \right)$ uniformly over $w \in \Delta_n^{\gamma^2/2}$ and $x \in \R^d$.
		\item  $\q{w}(S^c) \leq  \frac{4}{\gamma^2} \left( h + \frac{V_S}{\gamma}\phi(\sqrt{h}) +  \sqrt{\frac{1}{2n} \log \frac{3}{\delta}} \right)$
		uniformly over $w \in \Delta_n^{\gamma^2/2}$, whenever $h \leq h_0$, where $h_0$ is as defined in \Cref{lem:kernel-support-concentration}.
		
		\item  Denote $\smooth{q} = K_h \star q$. There is a $w \in \Delta_n^{\gamma^2/2}$ such that
		$\| \q{w} - \smooth{q} \|_{\infty} \leq
		\frac{8 c_0}{\gamma^4 h^d} \sqrt{\frac{1}{n} \log \frac{12}{\delta}}$. In particular, 
		$$
		\abs{\KL_S(\q{w}|p_\theta) - \KL_S(\smooth{q}|p_\theta)} \leq \frac{C}{ h^d \sqrt{n}} \sqrt{\log \frac{12}{\delta}
		}$$ 
		where $C = \frac{8 c_0}{\gamma^4} V_S(1 + \log(\frac{u}{\gamma l}))$ and $u$ and $l$ are the upper and lower bounds on $\q{w}$ in Item 2. We take $C = \infty$ if $l  \leq 0.$
	\end{enumerate}
\end{lemma}

\begin{proof}[Proof of \Cref{lem:existence-of-good-estimator}]
	
	Combining \Cref{lem:bbl-rademacher-bounds} with \Cref{lem:kernel-rademacher-bounds} shows that with probability $1-\delta/3$:
	\begin{align*}
		\left|\frac{1}{n} \sum_{i=1}^n K_h(x_i, x)  - \int K_h(y, x) p_0(y) dy\right|
		&\leq  \frac{4c_0}{\sqrt{n} h^d} + \frac{c_0}{\sqrt{n} h^d}\sqrt{2\log \frac{6}{\delta}} \leq \frac{6c_0}{\sqrt{n} h^d} \sqrt{\log \frac{6}{\delta}}  .
	\end{align*}
	for all $x \in \R^d$ simultaneously.
	Noting that $\hat{p}(x) = \frac{1}{n} \sum_i K_h(x_i, x)$ and $(K_h \star p_0)(x) = \int K_h(y,x) p_0(y) dy$, the uniform bound in Item 1 follows. 
	
	To show Item 2, note for each $w \in \Delta_n^{\gamma^2/4}$ that
	\[ 
	\frac{\gamma^2 \hat{p}(x)}{4} \leq \q{w}(x) = \sum_{i=1}^n w_i K_h(x_i, x) \leq \frac{4  \hat{p}(x)}{\gamma^2} \quad \text{for each } x \in \R^d. 
	\]
	Item 2 now follows from the bound in Item 1 and as $(K_h \star p_0)(\cdot) \in [\gamma, 1/\gamma]$, since $\int K_h(y,x) dy = 1$ and $p_0(\cdot) \in [\gamma, 1/\gamma]$.

	Now we will show Item 3. For arbitrary $w \in \Delta_n^{\gamma^2/4}$ observe that 
	$$
	\q{w}(S^c) \leq \frac{4  \hat{p}(S^c)}{\gamma^2}. 
	$$
	An application of Hoeffding's inequality to the function $f(\cdot) = \int_{S^c} K_h(\cdot, x) dx \in [0,1]$ implies that with probability $1-\delta/3$,
	\begin{align*}
		\hat{p}(S^c) = \frac{1}{n}\sum_{i=1}^n f(x_i) \, dx 
		&\leq  \int f(y) p_0(y) dy dx  +  \sqrt{\frac{1}{2n} \log \frac{3}{\delta}} \\
		&= \int_{S^c} \int K_h(y,x) p_0(y) dy dx +  \sqrt{\frac{1}{2n} \log \frac{3}{\delta}}\\  
		&\leq h + \frac{V_S}{\gamma}\phi(\sqrt{h}) + \sqrt{\frac{1}{2n} \log \frac{3}{\delta}},
	\end{align*}
	where,  the second inequality follows by invoking \Cref{lem:smooth-support-concentration} with the choice $q=p_0$ whenever $h \leq h_0$.
	
	Finally to show Item 4, we choose the weights $w_i = \frac{q(x_i)}{Zp_0(x_i)}$ for $i = 1, \ldots, n$, where $Z = \sum_{i=1}^n \frac{q(x_i)}{p_0(x_i)}$. By Hoeffding's inequality, we have that with probability at least $1-\delta/6$, 
	\[|Z - n| = |Z - \E[Z]| \leq \frac{1}{\gamma^2} \sqrt{\frac{n}{2} \log \frac{12}{\delta}} . \]
	In the event that this holds, if $n \geq \frac{2}{\gamma^4} \log \frac{12}{\delta}$, then $Z \in [n/2, 3n/2]$ and we have that $w \in \Delta_n^{\gamma^2/4}$. Moreover, it implies that for any $x \in \R^d$,
	\begin{align*}
		\left|\q{w}(x) - \frac{1}{n} \sum_{i=1}^n \frac{q(x_i)}{p_0(x_i)} K_h(x_i, x)\right| &=
		\left| \sum_{i=1}^n w_i K_h(x_i, x) - \frac{1}{n} \sum_{i=1}^n \frac{q(x_i)}{p_0(x_i)} K_h(x_i, x) \right| \\
		&\leq \left| \frac{n - Z}{Z} \right| \max_{i \in [n]} \frac{q(x_i)}{p_0(x_i)} K_h(x_i, x) \\
		&\leq \frac{2}{n}|n - Z| \frac{c_0}{\gamma^2 h^d} \leq \frac{2c_0}{\gamma^4 h^d} \sqrt{\frac{1}{n} \log \frac{12}{\delta}}.
	\end{align*}
	Next, observe that \Cref{lem:bbl-rademacher-lipschitz} and \Cref{lem:kernel-rademacher-bounds} imply that the Rademacher complexity $\Rad$ of the function class $\{ \cdot \mapsto \frac{q(\cdot)}{p_0(\cdot)} K_h(\cdot, x) \, : \, x \in \R^d \}$ on domain $S$ is bounded above by $\frac{2c_0}{\gamma^2 h^d \sqrt{n}}$. Therefore by \Cref{lem:bbl-rademacher-bounds}, with probability $1-\delta/6$,  we have
	\begin{align*}
\sup_{x \in \R^d} \left|  \frac{1}{n} \sum_{i=1}^n \frac{q(x_i)}{p_0(x_i)} K_h(x_i, x) - \int_S \frac{q(y)}{p_0(y)} K_h(y,x) \, p_0(y) dy \right| \leq \frac{c_0}{\gamma^2 h^d \sqrt{n}} \left(4 +  \sqrt{2\log \frac{12}{\delta}}\right). 
	\end{align*}
	Since $p_0, q$ are both supported on $S$ with $\inf_{x \in S} p_0(x) > 0$, we have $ (K_h \star q)(x) = \int_S q(y) K_h(x,y) \, dy = \int_S \frac{q(y)}{p_0(y)} K_h(x,y) p_0(y) dy$. Hence combining the previous two display equations and using the union bound, the uniform bound in Item 4 between $\q{w}$ and $\smooth{q}$ is seen to hold with probability $1-\delta/3$. Finally to bound the differences between the KL terms, note that
	\begin{align*}
		\abs{\KL_S(\q{w}|p_\theta) - \KL_S(\smooth{q}|p_\theta)} &\leq \int_S |\q{w}(x) \log \q{w}(x) - \smooth{q}(x) \log \smooth{q}(x)| dx \\
		&\quad + \int_S |\log p_\theta(x)| |\q{w}(x) - \smooth{q}(x)| dx \\
		&\leq \|\q{w} - \smooth{q}\|_{\infty} V_S (1 + \log \frac{u}{\gamma l}),
	\end{align*}
	where the last inequality follows by using the Lipschitz continuity of the map $\Phi(x) = x \log x$ on the interval $[l, u]$, with Lipschitz constants $1 + \log(u/l)$, and using the upper bound $|\log p_\theta(\cdot)| \leq \log(1/\gamma)$.
	
	A final application of the union bound shows that Items 1 through 4 can be simultaneously satisfied with probability $1-\delta$.
\end{proof}

\subsubsection{KL and TV approximation by kernel smoothed densities}
\label{sec:okl-smoothed-approx}

 In this section we show that the terms $\KL(q|p_\theta)$ and $\tv(q,p_0)$ do not increase by much when a density $q \in \Den$ is replaced by its kernel-smoothed version $\smooth{q} = K_h \star q$ for a suitably small bandwidth parameter $h > 0$.

\begin{lemma}
	\label{lem:smoothed-densities-kl-tv-approx}
	Suppose \Cref{assump:bounded-support,assump:bounded-densities,assump:smooth-densities,assump:kernel-properties} hold. There exists a constant $\tilde{c} >0$ and $\tilde{h} > 0$ depending only on the constants in the above assumption, such that the following statement holds. Suppose that $q \in \Den$ is supported on $S$ and bounded above by $1/\gamma$ and $h \leq \tilde{h}$ then
	\begin{itemize}
		\item $\tv(\smooth{q}, p_0) \leq \tv(q, p_0) + \tilde{c} \left( h^{\alpha/2} + h + \phi(\sqrt{h})  \right)$
		\item $\KL_S(\smooth{q} | p_\theta) \leq \KL(q | p_\theta) +  \tilde{c} \left( h^{\alpha/2} + h \log \frac{1}{h} + \phi(\sqrt{h}) \log \frac{1}{\phi(\sqrt{h})} \right)$.
	\end{itemize}
\end{lemma}

We will prove the TV and KL statements separately. We start with the TV statement.
\begin{lemma}
	\label{lem:bounded-tv}
	Suppose that $q \in \Den$ is supported on $S$ and suppose that \Cref{assump:bounded-support,assump:bounded-densities,assump:smooth-densities,assump:kernel-properties} hold. 
	Then 
	\[\tv(\smooth{q}, p_0) \leq \tv(q, p_0) 
	+ V_S \left( C_\alpha h^{\alpha/2} + h + \frac{\phi(\sqrt{h})}{\gamma}  \right)\]
	whenever $h \in (0,h_0)$.
\end{lemma}
\begin{proof}
	We first observe that for any densities $p, q \in \Den$, we have
	\begin{align*}
		\tv(\smooth{p},\smooth{q}) 
		&= \frac{1}{2}\int_{\R^d} |\Ez [q(x-hZ) - p(x-hZ)]| \, dx\\
		&\leq \frac{1}{2} \Ez \int_{\R^d}  |q(x-hZ) - p(x-hZ)| \, dx \\
		&= \frac{1}{2} \int_{\R^d} |q(x)-p(x)| \, dx = \tv(p, q),
	\end{align*}
	where $Z$ is an $\R^d$-valued random vector with probability density $x \mapsto \kappa(\|x\|_2)$, and the inequality follows from Jensen's inequality and Fubini's theorem. By the triangle inequality, we then have
	\begin{align*}
		\tv(\smooth{q}, p_0) \leq \tv(\smooth{q}, K_h \star p_0) + \tv(K_h \star p_0, p_0) \leq \tv(q, p_0) + \tv(K_h \star p_0, p_0).
	\end{align*}
	Next, note that for any $x \in S_{-\sqrt{h}}$ we have
\begin{align*}
    &\hspace{-3em}\left| p_0(x) - K_h \star p_0(x) \right| \\
    &= \left| \int_{\R^d} p_0(x) K_h(x, y)\, dy - \int_{\R^d} p_0(y) K_h(x,y)\, dy \right| \\
    &\leq \int_{\R^d} \left| p_0(x) -  p_0(y) \right| K_h(x,y) \, dy \\
    &= \int_{B(x,\sqrt{h})} \left| p_0(x) -  p_0(y) \right| K_h(x,y) \, dy 
    + \int_{B(x,\sqrt{h})^c} \left| p_0(x) -  p_0(y) \right| K_h(x,y) \, dy \\
    &\leq C_\alpha h^{\alpha/2} + h, 
\end{align*}
	where the first inequality follows from Jensen's inequality, and the second follows from \Cref{assump:smooth-densities} and \Cref{lem:kernel-support-concentration}, using the upper bound on the density $p_0(\cdot) \leq 1/\gamma$.
	
	Thus we have
	\begin{align*}
		&\hspace{-2em} \tv(K_h \star p_0, p_0) \\
		&= \frac{1}{2}\int_{\R^d} \left| p_0(x) - K_h \star p_0(x) \right| \, dx \\
		&= \frac{1}{2}\int_{S_{-\sqrt{h}}} \left| p_0(x) - K_h \star p_0(x) \right| \, dx + \frac{1}{2}\int_{\R^d \setminus S_{-\sqrt{h}}} \left| p_0(x) - K_h \star p_0(x) \right| \, dx \\ 
		&\leq \frac{\lambda(S_{-\sqrt{h}})}{2} \left( C_\alpha h^{\alpha/2} +  h \right) + \frac{1}{2}\int_{S \setminus S_{-\sqrt{h}}} \left| p_0(x) - (K_h \star p_0)(x) \right| \, dx
		+ \frac{1}{2} (K_h \star p_0)(S^c) \\
		&\leq V_S (C_\alpha h^{\alpha/2} + h + \gamma^{-1}\phi(\sqrt{h}) ).
	\end{align*}
	where the first inequality follows from the last display equation and using the fact that $p_0(x) = 0$ whenever $x \in S^c$, and the second inequality follows by using that  $|p_0(\cdot) - (K_h \star p_0)(\cdot)| \leq 1/\gamma$ and $\lambda(S\setminus S_{-r}) = V_S \phi(r)$, along with the bound $(K_h \star p_0)(S^c) \leq h + \frac{V_S}{\gamma} \phi(\sqrt{h})$ from \Cref{lem:smooth-support-concentration}.
\end{proof}

We now turn to proving the KL bound. To do so, we will use the decomposition of the KL into negative entropy and cross-entropy and bound each term separately.

\begin{lemma}
	\label{lem:bounded-neg-entropy}
	Suppose Assumptions~\ref{assump:bounded-support} and~\ref{assump:kernel-properties} hold, $q \in \Den$ has support $S$ and is bounded above by $1/\gamma$. Then 
	\begin{equation*}
		\int_S \smooth{q}(x) \log \smooth{q}(x) dx \leq \int_S q(x) \log q(x) dx 
		+ D(h)\left[ \frac{d}{2} \log(2\pi e M) + \left(\frac{d}{2} + 1 \right) \log \frac{1}{D(h)} \right],
	\end{equation*}
	where $M = 2 R^2 + 2 v_d \left(c_0 t_0^{d+3} + \frac{C_\rho}{\rho} \Gamma \left( \frac{d+3}{\rho} \right)\right)$ and $D(h) = h + \frac{V_S}{\gamma} \phi(\sqrt{h})$, whenever $h \leq \min(h_0, h_1)$ where $h_1 = \inf\{h > 0 : D(h) > 1/e\}.$
\end{lemma}
\begin{proof}
	First let us note that the integrals $\int_{\R^d} q(x) \log q(x) dx$ and $\int_{\R^d} \smooth{q}(x) \log \smooth{q}(x) dx$ are well-defined. Indeed, the first integral is well defined since $q$ is bounded above by $1/\gamma$ on a set $S$ with finite Lebesgue measure. Although, $\smooth{q}$ can also be shown to be bounded above by $1/\gamma$, its support is unbounded if the kernel $K_h$ also has unbounded support. Instead, for the existence of the second integral, it suffices to show that the second moment of $\smooth{q}$ is finite (e.g.~see \cite{ghourchian2017existence}). To show the latter, we will use the property a random variable $Y$ with density $\smooth{q}$, has the same distribution as $X + hZ$, where $X$ is random variable with density $q$, and let $Z$ be a random variable with density $x \mapsto \K(\|x\|_2)$ independent of $X$. Hence
	\begin{equation}
		\E_{Y \sim \smooth{q}}[\| Y \|_2^2] 
		= \E_{Z, X}[\|X+ h Z \|^2] 
		\leq 2 R^2 + 2 h^2 v_d \left(c_0 t_0^{d+3} + \frac{C_\rho}{\rho} \Gamma \left( \frac{d+3}{\rho} \right) \right) = M 
		\label{eq:l2bound}
	\end{equation}  
	where we have used the inequality $\|x + z \|_2^2 \leq 2\|x\|_2^2 + 2\|z \|_2^2$ along with the bound on $\E \|Z\|^2$ from \Cref{lem:bounded-second-moment} and the bound $\E\|X\|^2 \leq R^2$.

	Next, let us use the convexity of the function $\Phi:\Rnn \to \R$ given by $\Phi(x)=x \log x$ to conclude that
	\begin{equation*}
		\begin{aligned}
			\int_{\R^d} \smooth{q}(y) \log \smooth{q}(y) dy 
			&= \int_{\R^d} \Phi(\Ez[q(y-hZ)])  dy
			\overset{(i)}{\leq} \int \Ez [\Phi(q(y-hZ))] dy\\
			&\overset{(ii)}{=} \Ez \int_{\R^d} \Phi(q(y-hZ)) dy \overset{(iii)}{=} \int_S \Phi(q(u)) du = \int_S q(u) \log q(u) du
		\end{aligned}
	\end{equation*}
	where we have used Jensen's inequality in Step (i) and Fubini's theorem in Step (ii). In Step (iii) we used the change of variables $u=y-hZ$ as well as the convention that $\Phi(0) = 0$.
	
	Now observe that we can decompose the negative entropy of $\smooth{q}$ as
	\begin{align*}
		\int_{S} \smooth{q}(y) \log \smooth{q}(y) dy  
		&= \int_{\R^d} \smooth{q}(y) \log \smooth{q}(y) dy + \int_{S^c} \smooth{q}(y) \log \frac{1}{\smooth{q}(y)} dy .
	\end{align*}
	Now define the conditional probability density $\bar{q}(x) = \frac{\smooth{q}(x) \I{x \in S^c}}{\smooth{q}(S^c)}$. Then we have the identity
	\begin{align*}
		\int_{S^c} \smooth{q}(y) \log \frac{1}{\smooth{q}(y)} dy = \smooth{q}(S^c) \log \frac{1}{\smooth{q}(S^c)} + \smooth{q}(S^c) \int_{\R^d} \bar{q}(y) \log \frac{1}{\bar{q}(y)} dy .
	\end{align*}
	We now claim that $\bar{q}$ has bounded second moment. To see this, observe that we can write
	\[  \E_{y \sim \smooth{q}}[\|y \|^2] = \smooth{q}(S) \E_{y \sim \smooth{q}}[\|y \|^2 | y \in S] +   \smooth{q}(S^c) \E_{y \sim \smooth{q}}[\|y \|^2 | y \in S^c].\]
Thus, $\E_{y \sim \bar{q}}[\|y\|^2] = \E_{y \sim \smooth{q}}[\|y \|^2 | y \in S^c] \leq M/\smooth{q}(S^c)$ by \cref{eq:l2bound}. 
	
	Now in order to bound the negative entropy of $\bar{q}$, let $\mu = \E_{y \sim \bar{q}} [y]$ and $\sigma^2 = \E_{y \sim \bar{y}}[\|y-\mu\|^2] \leq M/\smooth{q}(S^c)$ denote the mean and mean squared error of $\bar{q}$, and let $g(x) = \frac{1}{(2\pi\sigma^2)^{d/2}} e^{-\frac{\|x-\mu\|^2}{2\sigma^2}}$ be the probability density function of the the multivariate normal distribution with the same first two moments. Examining the property $\KL(\bar{q}|g) \geq 0$, we obtain
	$$ 
	\int_{\R^d} \bar{q}(y) \log \frac{1}{\bar{q}(y)} dy \leq \frac{1}{2} + \frac{d}{2} \log(2 \pi \sigma^2) \leq \frac{d}{2} \log (2 \pi e M\smooth{q}(S^c)^{-1}). 
	$$
	Finally, by \Cref{lem:smooth-support-concentration}, $\smooth{q}(S^c) \leq h + \frac{V_S}{\gamma} \phi(\sqrt{h})$. Putting it all together and noting that the function $x \mapsto x \log \frac{1}{x}$ is monotonically increasing on $x \in (0,1/e]$, gives the lemma statement.
\end{proof}

\begin{lemma} Suppose that $q \in \Den$ is supported on $S$ and bounded above by $1/\gamma$. Further suppose that \Cref{assump:bounded-support,assump:bounded-densities,assump:smooth-densities,assump:kernel-properties} hold. Then
	\label{lem:bounded-cross-entropy}
\[ \left| \int_S \smooth{q}(x) \log p_\theta(x) dx -  \int_S q(x) \log p_\theta(x) dx \right| \leq  C_\alpha h^{\alpha/2} +  2h  \log \frac{1}{\gamma} 
	+\frac{V_S \phi(\sqrt{h})}{\gamma} \log \frac{1}{\gamma}
	\]
	whenever $h \leq h_0$.
\end{lemma}
\begin{proof}
	Using Fubini's theorem, we obtain 
	$$
	\int_S \smooth{q}(x) \log p_\theta(x) dx = \int_{\R^d} q(y) l(y) dy = \int_S q(y) l(y) dy
	$$
	where $l(y) = \int_S K_h(x,y) \log p_\theta(x) dx$ and last equality follows since $q$ is supported on $S$.
	For any $y \in S_{-\sqrt{h}}$ (recall this means that $B(y,\sqrt{h}) \subseteq S$) and hence we have:
	\begin{align*}
		\abs{\log p_\theta(y) - l(y)} &= \abs*{\int_{\R^d} \log p_\theta(y) K_h(x,y) dx - \int_S \log p_\theta(x)  K_h(x,y) dx } \\
		&\leq \int_{B(y,\sqrt{h})}  K_h(x,y) |\log p_\theta(y)- \log p_\theta(x)| \, dx \\
		&+ \int_{B(y,\sqrt{h})^c} K_h(x,y) \left(|\log p_\theta(y)| + |\log p_\theta(x)|\right) \, dx \\
		&\leq   C_\alpha h^{\alpha/2} +  2h  \log \frac{1}{\gamma},
	\end{align*}
	where we have used \Cref{assump:bounded-densities,assump:smooth-densities} and \Cref{lem:kernel-support-concentration} in the last step.

	Thus we have
	\begin{align*}
		\left| \int_S q(y) l(y) \, dy -  \int_S q(y) \log p_\theta(y) \, dy \right| 
		&\leq \int_S q(y) \left|  l(y)  - \log p_\theta(y) \right| \, dy \\
		&=  \int_{S_{-\sqrt{h}}} q(y) \left|  l(y)  - \log p_\theta(y) \right| \, dy \\
		&+ \int_{S\setminus S_{-\sqrt{h}}} q(y) \left|  l(y)  - \log p_\theta(y) \right| \, dy \\
		&\leq  C_\alpha h^{\alpha/2} +  2h  \log \frac{1}{\gamma} 
		+\frac{V_S \phi(\sqrt{h})}{\gamma} \log \frac{1}{\gamma} . \qedhere
	\end{align*}
\end{proof}

\begin{proof}[Proof of \Cref{lem:smoothed-densities-kl-tv-approx}]
	The TV statement follows immediately from \Cref{lem:bounded-tv}. To see the KL statement, we first observe that
	\[ \KL_S(\smooth{q}, p_\theta) = \int_S \smooth{q}(x) \log \frac{\smooth{q}(x)}{p_\theta(x)} \, dx = \int_S \smooth{q}(x) \log \smooth{q}(x) \, dx - \int_S \smooth{q}(x) \log p_\theta(x)\, dx.  \]
	Plugging in \Cref{lem:bounded-neg-entropy} for the first term and \Cref{lem:bounded-cross-entropy} for the second term completes the proof.
\end{proof}

\subsection{Proof of \Cref{thm:okl-convergence-formal}}
\label{sec:proof-of-theorem}

We now put all of the above together to prove \Cref{thm:okl-convergence-formal}. We suppose \Cref{assump:bounded-densities,assump:bounded-support,assump:kernel-properties,assump:smooth-densities} hold and thus obtain the following corollaries of   \Cref{lem:uniform-convergence,lem:existence-of-good-estimator,lem:smoothed-densities-kl-tv-approx}  when we take 
$\delta_{n} = e^{- (\log n)^{2\beta}}$ 
and $\eta_{n,h} = \frac{(\log n)^\beta}{\sqrt{n} h^d}$ for a fixed constant $\beta > 0$ (e.g.~take $\beta=1$). In the following, $n_0, \bar{h}, \bar{\eta} > 0$ are constant quantities that can depend on any of the terms used in the assumptions and on $\beta > 0$, but are independent of $n$ and $h$. Given two expressions $f$ and $g$, the notation $f = O(g)$ will be used to denote that  $|f| \leq c |g|$ holds for a similar constant $c$. We will delay instantiating a concrete value of $\epsilon'$ in \Cref{cor:existence-of-good-weights} until we start the final proof, but note for now that the constants $c, n_0, \bar{h}, \bar{\eta}$ do not depend on the choice of $\epsilon'$ in any way.

\begin{corollary}
\label{cor:existence-of-good-weights}
Suppose a fixed value $\epsilon' > 0$ is given. Then whenever $n \geq n_0, h \leq \bar{h}$ and $\eta_{n,h} \leq \bar{\eta}$, the following holds with probability at least $1-\delta_{n}$:
\begin{enumerate}
    \item $\tv(\hat{p}, p_0) =  O\left(\eta_{n,h} + h^{\min(\frac{\alpha}{2},1)} + \phi(\sqrt{h})\right)$.
    \item $\Delta_n^{\gamma^2/4} \subseteq \Wcal_{\frac{\gamma^3}{8}, \frac{8}{\gamma^3}}$, where $\Wcal_{\ell, u}$ is as defined in \Cref{lem:uniform-convergence}.
    \item $\sup_{w \in \Delta_n^{\gamma^2/2}} \q{w}(S^c) = O(h + \phi(\sqrt{h}))$.
    \item There is $w^{\epsilon'} \in \Delta_n^{\gamma^2/4}$ such that
    $$
    \KL_S(\q{w^{\epsilon'}}|p_\theta) \leq \okl[\epsilon'] + O \left( \eta_{n,h}  + h^{\frac{\alpha}{2}} + h \log \frac{1}{h} + \phi(\sqrt{h}) \log \frac{1}{\phi(\sqrt{h})}\right)
    $$
    and
    $$
    \tv(\q{w^{\epsilon'}}, p_0) \leq {\epsilon'} + O \left( \eta_{n,h}  +  h^{\min(\frac{\alpha}{2},1)} + \phi(\sqrt{h}) \right).
    $$
\end{enumerate}	
\end{corollary}

\begin{proof}
	Observe by \Cref{lem:info-proj-sandwich} that the information projection $\iproj[{\epsilon'}]$ is sandwiched between $p_0$ and $p_\theta$; in particular, $\iproj[{\epsilon'}]$ has support $S$ and is bounded between values $\gamma$ and $1/\gamma$ on $S$. Thus we will invoke \Cref{lem:existence-of-good-estimator} with $q = \iproj[\epsilon']$ and $\delta = \delta_n$, noting that the condition $n \geq \frac{2}{\gamma^2} \log \frac{12}{\delta}$ is satisfied when $n_0$ is suitably large. Hence we may now prove the respective items.

	We show Item 1 as follows:
	$$
	\tv(\hat{p}, p_0) \leq \tv(\hat{p}, K_h \star p_0) + \tv(K_h \star p_0, p_0) = O(\eta_{n,h} + h) + O(h^{\min(\frac{\alpha}{2},1)} + \phi(\sqrt{h})),
	$$
	where the term $\tv(\hat{p}, K_h \star p_0)$ was bounded by combining Item 1 of \Cref{lem:existence-of-good-estimator} with \Cref{cor:tv-uniform-estimates}, and the term $\tv(K_h \star p_0, p_0)$ was bounded by using  \Cref{lem:smoothed-densities-kl-tv-approx} with $q = p_0$ (assuming suitable choice of constants $\bar{h}$ and $c$).
	
	Next, Item 2 follows from Item 2 of \Cref{lem:existence-of-good-estimator} as long as $\frac{6c_0}{\sqrt{n}h^d} \sqrt{\log \frac{4}{\delta_n}} = O(\eta_{n,h}) \leq \gamma/2$ is satisfied, which holds when $\bar{\eta}$ is a suitably small constant.
	
	Next, Item 3 follows from Item 3 of \Cref{lem:existence-of-good-estimator} noting that $\sqrt{\frac{1}{2n} \log \frac{6}{\delta_n}} = O( \eta_{n,h} h^d ) = O( \bar{\eta} h)$.
	
	Finally to show Item 4, we will invoke Item 4 from \Cref{lem:existence-of-good-estimator} for the choice $q=\iproj[{\epsilon'}]$. Hence we obtain a $w^{\epsilon'} \in \Delta_{n}^{\gamma^2/4}$ such that $\|\q{w^{\epsilon'}}-\iproj[{\epsilon'}]\|_\infty = O(\eta_{n,h})$ and 
	$$
	|\KL_S(\q{w^{\epsilon'}}|p_\theta) - \KL_S(K_h \star \iproj[{\epsilon'}]|p_\theta)|  =  O(\eta_{n,h}).
	$$
	Further, \Cref{cor:tv-uniform-estimates} shows
	$$
	\tv(\q{w^{\epsilon'}}, K_h \star \iproj[{\epsilon'}]) \leq O(\eta_{n,h} + h)
	$$
	since $\mu = \sum_{i=1}^n w^{\epsilon'}_i \delta_{x_i}$ and the measure $\nu$ given by density $\iproj[{\epsilon'}]$, are both supported on $S$. The bounds in Item 4 now follow using \Cref{lem:smoothed-densities-kl-tv-approx} with $q=\iproj[{\epsilon'}]$ and the triangle inequality, noting that $\KL(\iproj[{\epsilon'}]|p_\theta) = \okl[{\epsilon'}]$ and $\tv(\iproj[{\epsilon'}], p_0) \leq \epsilon'$.
\end{proof}

Since our goal is to approximate $\okl$ by $\hatI(\theta)$, let us rewrite \cref{eqn:I-hat-defn} as 
$$
\hatI(\theta) = \inf_{\substack{w \in \Delta_n^{\gamma^2/4} \\ \hatTV(\q{w}, p_0) \leq \epsilon}} \hatKL(\q{w}|p_\theta)
$$
where, given $q \in \Den$,
$$
\hatKL(q|p_\theta) = \frac{1}{n} \sum_{i=1}^n \frac{q(x_i)}{\hat{p}_\gamma(x_i)} \log \frac{q(x_i)}{p_\theta(x_i)}
$$
and
$$
\hatTV(q,p_0) = \frac{1}{2n} \sum_{i=1}^n \abs*{\frac{q(x_i)}{\hat{p}_\gamma(x_i)} - 1}.
$$

Next, the following corollary of \Cref{lem:uniform-convergence}  shows that, with high probability, the estimators $\hatKL(\q{w}|p_\theta)$ and $\hatTV(\q{w}, p_0)$ are close to their population level targets $\KL_S(\q{w}|p_\theta)$ and $\tv(\q{w}, p_0)$, uniformly over all $w \in \Delta_n^{\gamma^2/4}$.   

\begin{corollary}
	\label{cor:all-the-bounds}
	Suppose $n \geq n_0$, $h \leq \bar{h}$ and $\eta_{n,h} \leq \bar{\eta}$, then with probability at least $1-2\delta_n$ the event in \Cref{cor:existence-of-good-weights} holds, and further uniformly over all $w \in \Delta_{n}^{\gamma^2/4}$ it holds that
	$$
	\abs*{\hatKL(\q{w}|p_\theta) - \KL_S(\q{w}|p_\theta)} \leq \psi_{n,h} \quad \text{and } \abs*{\hatTV(\q{w}, p_0) - \tv(\q{w}, p_0)} \leq \psi_{n,h}
	$$
	for a deterministic quantity $\psi_{n,h}$ that satisfies $\psi_{n,h} = O \left(\eta_{n,h} + h^{\min(\frac{\alpha}{2},1)} + \phi(\sqrt{h})\right)$.
\end{corollary}
\begin{proof} Let $E_1$ denote the event in \Cref{cor:existence-of-good-weights}, and let $E_2$ denote the event in \Cref{lem:uniform-convergence} with the choices  $\delta = \delta_n$, $l = \frac{\gamma^3}{8}$, $u = \frac{8}{\gamma^2}$. We will assume that both events $E_1$ and $E_2$ hold simultaneously, which happens with probability at least $1-2\delta_n$.
	
	Recall that on the event $E_1$, we have $\Delta_{n}^{\gamma^2/4} \subseteq \Wcal_{\ell, u}$, $\sup_{w \in \Delta_n^{\gamma^2/4}} \q{w}(S^c) = O(h + \phi(\sqrt{h}))$, are  $\tv(\hat{p}, p_0) = O\left(\eta_{n,h} + h^{\min(\frac{\alpha}{2}, 1)} + \phi(\sqrt{h})\right)$. The proof can be completed by combining the above with the bounds from \Cref{lem:uniform-convergence} under the event $E_2$. 
\end{proof}

For a suitable constant $c$, the term $\xi_{n,h} = c\left(\eta_{n,h} + h^{\frac{\alpha}{2}} + h \log \frac{1}{h} + \phi(\sqrt{h}) \log \frac{1}{\phi(\sqrt{h})}\right)$ dominates all the error terms (i.e.~the terms of the form $O(\cdot)$) in the statement of \Cref{cor:existence-of-good-weights} and \Cref{cor:all-the-bounds}. Using this term, we are ready to finish the proof of \Cref{thm:okl-convergence-formal}.

\begin{proof}[Proof of \Cref{thm:okl-convergence-formal}]
We will invoke \Cref{cor:all-the-bounds} with the choice $\epsilon' = \epsilon - 2 \xi_{n,h}$ (see \Cref{cor:existence-of-good-weights}). We first claim that the following holds:
\begin{align}
\label{eqn:okl-upper-and-lower-bound}
\okl[\epsilon + 2\xi_{n,h}] - O(\xi_{n,h}) \leq \hatI(\theta) \leq \okl[\epsilon - 2\xi_{n,h}] + O(\xi_{n,h})
\end{align}
with probability at least $1-2\delta_n$.
	
For the upper bound, let $w^{\epsilon'} \in \Delta_n^{\gamma^2/4}$ be the vector from \Cref{cor:existence-of-good-weights,cor:all-the-bounds} and note that
$$
\hatTV(\q{w}, p_0) \leq \tv(\q{w}, p_0) + \xi_{n,h} \leq \epsilon ' + 2\xi_{n,h} \leq \epsilon,
$$
and hence 
$$
\hatI(\theta) \leq \hatKL(\q{w}|p_\theta) \leq \KL_S(\q{w}|p_\theta) + \xi_{n,h} \leq \okl[\epsilon'] + 2\xi_{n,h}.
$$

To lower bound $\hatI(\theta)$, let $w \in \Delta_{n}^{\gamma^2/4}$ be such that $\hatTV(\q{w}, p_0) \leq \epsilon$. Then by  \Cref{cor:existence-of-good-weights}, \Cref{cor:all-the-bounds}, and \Cref{lem:restrict-to-S}, there is a $\bar{q}_w$ that is supported on $S$ such that
$$
\tv(\bar{q}_w, p_0) \leq \tv(\q{w}, p_0) + \q{w}(S^c) \leq \hatTV(\q{w}, p_0) + 2\xi_{n,h} \leq \epsilon + 2\xi_{n,h}
$$
and 
$$
\begin{aligned}
    \hatKL(\q{w}|p_\theta) &\geq \KL_S(\q{w}|p_\theta) - \xi_{n,h}\\
     &= (1-\q{w}(S^c))  \KL(\bar{q}_w|p_\theta) + (1-\q{w}(S^c))\log (1-\q{w}(S^c)) - \xi_{n,h} \\
    &\geq (1-\xi_{n,h}) \okl[\epsilon + 2\xi_{n,h}] - 2\xi_{n,h}
\end{aligned}
$$
where we have used that $\q{w}(S^c) \leq \xi_{n,h}$, $h(x) = (1-x) \log (1-x) \geq -x$ by the convexity of $h$, and $\KL(\bar{q}_w|p_\theta) \geq \okl[\epsilon + 2\xi_{n,h}]$ since $\tv(\bar{q}_w, p_0) \leq \epsilon + 2\xi_{n,h}$. Since the right hand side of the previous display doesn't depend on the choice of $w$, by considering the infimum over all $w \in \Delta_n^{\gamma^2/4}$ such that $\hatTV(\q{w}, p_0) \leq \epsilon$ we see
$$
\begin{aligned}
    \hatI(\theta) &\geq \okl[\epsilon + 2\xi_{n,h}] - \xi_{n,h} (1+\okl[\epsilon + 2\xi_{n,h}]) \geq \okl[\epsilon + 2\xi_{n,h}] - \xi_{n,h} (1+ \KL(p_0|p_\theta)).
\end{aligned}
$$
Note that $\KL(p_0|p_\theta) \leq   2\log(1/\gamma)$ by H\"older's inequality, and we have thus shown the lower bound and established \cref{eqn:okl-upper-and-lower-bound}.

Finally, we appeal to the continuity of the OKL function with respect to the coarsening radius (i.e., \Cref{lem:okl-continuity}) to see that
\[ I_{\epsilon - 2\xi_{n,h}}(\theta) \leq I_{\epsilon}(\theta) + \frac{2\xi_{n,h}}{\epsilon} \KL(p_0 | p_\theta) \leq I_{\epsilon}(\theta) + \frac{4 \xi_{n,h}}{\epsilon} \log(1/\gamma) = I_{\epsilon}(\theta) + O(\xi_{n,h}) \]
and
\[ I_{\epsilon + 2\xi_{n,h}}(\theta) \geq I_{\epsilon}(\theta) - \frac{2\xi_{n,h}}{\epsilon + 2\xi_{n,h}} \KL(p_0 | p_\theta) \leq I_{\epsilon}(\theta) - \frac{4 \xi_{n,h}}{\epsilon} \log(1/\gamma) = I_{\epsilon}(\theta) - O(\xi_{n,h}), \]
where we have used \Cref{assump:bounded-densities} and H\"{o}lder's inequality to establish $\KL(p_0|p_\theta) \leq 2 \log(1/\gamma)$. Combining the above two inequalities with \Cref{eqn:okl-upper-and-lower-bound} completes the proof.
\end{proof}

\section{Consistency of OWL estimators} 
\label{sec:owl-consistency}
Consider the estimator $\hatI(\theta)$ for $\okl$ from \cref{eq:okle} based on data $x_1, \ldots, x_n \sim P_0$. In this section we establish basic theory that, under \Cref{ass:misspecification}, shows that the OWL estimator $\hat{\theta} = \argmin_{\theta \in \Theta} \hatI(\theta)$,  satisfies $\left(\tv(P_{\hat{\theta}}, P_0) - \epsilon \right)^{+} \pconv 0$ as $n \to \infty$. 

For this result, we will consider the setup when the model family $\{P_\theta\}_{\theta \in \Theta}$ has densities $\{p_\theta\}_{\theta \in \Theta}$ that are always positive and continuous on a compact data space $\cX$. This allows us to define the following notion of distance on $\Theta$:
$$
\dsuplog(\theta_1, \theta_2) = \sup_{x \in \cX} \, \abs*{\log\frac{ p_{\theta_1}(x)}{p_{\theta_2}(x)}}.
$$
When the parameters are identifiable (i.e. $P_{\theta_1} = P_{\theta_2} \implies \theta_1 = \theta_2$), $d$ is a metric on $\Theta$ which may  take the value $\infty$. 
\begin{assume}
	$\Theta$ is a compact topological space, and the ball $B_\dsuplog(\theta, r) \doteq \{ \theta' \in \Theta : \dsuplog(\theta, \theta') < r \}$ is open for every $\theta \in \Theta$ and $r > 0$.
	\label{ass:compact-d}
\end{assume}

For instance, \Cref{ass:compact-d} is satisfied whenever $\cX$ and $\Theta$ are both compact metric spaces and the map $(\theta, x) \mapsto \log p_{\theta}(x)$ from $\Theta \times \cX \to \R$ is jointly (and thus uniformly) continuous.

\begin{lemma} The OKL function $I_\epsilon$ and its estimator $\hatI$ are both Lipschitz continuous in $\dsuplog$. In particular, for any $\theta_1, \theta_2 \in \Theta$:
	\begin{align*}
			\abs*{I_\epsilon(\theta_1) - I_\epsilon(\theta_2)} \leq \dsuplog(\theta_1, \theta_2) \text{ and } \abs*{\hatI(\theta_1) - \hatI(\theta_2)} \leq (1+2\epsilon) \dsuplog(\theta_1, \theta_2). 
	\end{align*}
	\label{lem:okl-lipschitz-d}
\end{lemma}
\begin{proof}
	Recall, $I_\epsilon(\theta) = \inf_{\substack{q \in \Den \\ \tv(q, p_0) \leq \epsilon}} F(q, \theta)$, where  $F(q, \theta) = \int q(x) \log \frac{q(x)}{p_\theta(x)} \lambda(x)$. Observe that
	$$
	\abs{F(q, \theta_1) - F(q, \theta_2)} \leq  \abs*{ \int  q(x) \log \frac{p_{\theta_1}(x)}{p_{\theta_2}(x)} d\lambda(x)}  \leq \dsuplog(\theta_1, \theta_2)
	$$
	for any $q \in \Den$. Thus the Lipschitz property of $I_\epsilon$ now follows by using 
	$$
	 I_\epsilon(\theta_1) = \inf_{\substack{q \in \Den \\ \tv(q, p_0) \leq \epsilon}} F(q, \theta_1) \leq \inf_{\substack{q \in \Den \\ \tv(q, p_0) \leq \epsilon}} F(q, \theta_2) + d(\theta_1, \theta_2) = I_\epsilon(\theta_2) + \dsuplog(\theta_1,\theta_2),
	$$
	and its symmetric bound obtained by exchanging $\theta_1$ and $\theta_2$. 
	
	The Lipschitz property of $\hatI(\theta) = \inf_{\substack{w \in \hat{\Delta}_n \\ \|w-o\|_1 \leq 2\epsilon}}  \sum_{i=1}^n w_i \log \frac{w_i n \hat{p}(x_i)}{p_\theta(x_i)}$ can similarly be shown by noting that for $F(w, \theta) = \sum_{i=1}^n w_i \log \frac{n w_i \hat{p}(x_i)}{p_\theta(x_i)}$, 
	$$
	\abs*{F(w,\theta_1) - F(w,\theta_2)} =  \abs*{\sum_{i=1}^n w_i \log \frac{p_{\theta_2}(x_i)}{p_{\theta_1}(x_i)}} \leq \dsuplog(\theta_1, \theta_2) \|w\|_1 \leq (1+2\epsilon) \dsuplog(\theta_1, \theta_2)
	$$ 
	whenever $\|w-o\|_1 \leq 2\epsilon$.
\end{proof}

This Lipschitz property can be used to show convergence of the estimators $\hatI(\theta)$ to $\okl$ uniformly over all $\theta \in \Theta$.

\begin{lemma} Suppose that \Cref{ass:compact-d} is satisfied, and the pointwise convergence $\abs{\hatI(\theta) - \okl} \pconv 0$ as $n \to \infty$ holds at each fixed $\theta \in \Theta$. Then we have the uniform convergence
	\begin{equation*}
			\sup_{\theta \in \Theta} \abs*{\hatI(\theta) - \okl} \pconv 0
	\end{equation*}
	as $n \to \infty$.
	\label{lem:uniform-conv-okle}
\end{lemma} 
\begin{proof}
	For any $r > 0$, let $S_r \subseteq \Theta$ denote a subset of minimal size, such that balls $B_\dsuplog(\theta, r) \doteq \{ \theta' | \dsuplog(\theta, \theta') < r \}$ with centers $\theta$ in $S_r$ cover all of $\Theta$. Since $\{B_\dsuplog(\theta, r)\}_{\theta \in \Theta}$ is an open cover of $\Theta$ for every $r > 0$, we have $|S_r| < \infty$ by the compactness of $\Theta$. Hence, \Cref{lem:okl-lipschitz-d} shows that
	$$
		\sup_{\theta \in \Theta} \abs*{\hatI(\theta) - \okl} \leq \max_{\theta \in S_r} \abs*{\hatI(\theta) - \okl} + 2(1+\epsilon) r
	$$
	for every $r > 0$.
	
	Thus, given any $\delta > 0$, we can choose an $r > 0$ such that $2(1+\epsilon) r \leq \delta/2$. This would show that
	\begin{align*}
			\limsup_{n \to \infty} \prob\left(\sup_{\theta \in \Theta} \abs*{\hatI(\theta) - \okl}  > \delta\right) &\leq  \limsup_{n \to \infty} \prob\left(\max_{\theta \in S_r} \abs*{\hatI(\theta) - \okl}  > \delta/2\right) \\
			&\leq \sum_{\theta \in S_r } 
			\lim_{n \to \infty} \prob\left( \abs*{\hatI(\theta) - \okl}  > \delta/2\right) = 0.
	\end{align*}
	The equality in the right hand side limit to zero follows from the assumed convergence $\abs*{\hatI(\theta) - \okl} \pconv 0$ for each $\theta \in \Theta$, and that $S_r$ is a fixed finite subset of $\Theta$. Since $\delta > 0$ is arbitrary, this completes the proof.
\end{proof}

The following result, obtained by applying Pinsker's inequality to our setup, bounds the total variation distance $\tv(P_\theta, P_0)$ in terms of the OKL value $\okl$ for any $\epsilon \geq 0$. This shows that when $\okl$ is small, $\tv(P_\theta, P_0)$ cannot be substantially larger than $\epsilon$. 

\begin{lemma} For any $\theta \in \Theta$ and $\epsilon \geq 0$
$$
\tv(P_\theta, P_0) \leq \sqrt{\frac{\okl}{2}} + \epsilon.
$$
\label{lem:okl-pinsker}
\end{lemma}
\begin{proof} It suffices to consider the case $\okl < \infty$. Then, as mentioned in \Cref{def:okl}, there is a unique distribution $Q^\theta \in \cP(\cX)$ called as the I-projection, such that $\tv(Q^\theta, P_0) \leq \epsilon$ and $\okl = \KL(Q^\theta|P_\theta)$ (we take $Q^\theta = P_0$ when $\epsilon = 0$). 
	
Pinsker's inequality shows
$$
\tv(Q^\theta, P_\theta) \leq \sqrt{\frac{1}{2} \KL(Q^\theta|P_\theta)} = \sqrt{\frac{\okl}{2}},
$$
and thus the proof follows by noting that $\tv(P_\theta, P_0) \leq \tv(Q^\theta, P_0) + \tv(Q^\theta, P_\theta) \leq \epsilon + \sqrt{\frac{\okl}{2}}$.
\end{proof}

We can now state and prove our main result in this section:

\begin{proposition} Suppose that  \Cref{ass:misspecification} is satisfied, and let $x_1, \ldots, x_n \iid P_0$. 
\begin{enumerate}
	\item Further suppose that the uniform convergence in the conclusion of  \Cref{lem:uniform-conv-okle} holds, and
	\item Consider the OWL estimator $\hat{\theta} \in \Theta$ defined to be any global minimizer of $\hatI$, i.e. $|\hatI(\hat{\theta}) - \min_{\theta \in \Theta} \hatI(\theta) | \pconv 0$ as $n \to \infty$.
\end{enumerate}
Then, for large values of $n$, the OWL estimated model $P_{\hat{\theta}}$ will be within $\epsilon$ TV-distance of $P_0$. More precisely
$$
\max\left(\tv(P_{\hat{\theta}}, P_0) - \epsilon, 0\right) \pconv 0 \qquad \text{as } n \to \infty.
$$
\end{proposition}
\begin{proof} The argument uses the following inequalities:
	\begin{equation*}
			\tv(P_{\hat{\theta}}, P_0) - \epsilon \leq \sqrt{\frac{1}{2} \hatI(\hat{\theta})}
			 \leq  \sqrt{\frac{1}{2} \inf_{\theta \in \Theta} \hatI(\theta) + o_P(1)} 
			 \leq \sqrt{\frac{1}{2} \inf_{\theta \in \Theta} \okl+ o_P(1)} = o_P(1),
	\end{equation*}
	Here, the first inequality follows from \Cref{lem:okl-pinsker}, the second from the fact that $\hat{\theta}$ is a global minimizer of $\hatI$ (item 2), the third from the uniform convergence $\sup_{\theta \in \Theta} \abs*{\hatI(\theta) - \okl} \pconv 0$ (item 1), and the last equality from the fact that $\min_{\theta \in \Theta} \okl = 0$ under \Cref{ass:misspecification}.
\end{proof}

\section{Performing OWL computations}
\label{sec:comp}

In this section, we expand on the computational details of the OWL methodology. Specifically, in \Cref{sec:i-projection} we discuss our approach to solve the constrained convex optimization problems in \cref{eq:finiteokle,eq:okle}, and in \Cref{sec:weighted-likelihood} we discuss the details of optimizing a weighted likelihood for exponential families and mixture models.

\subsection{Computing the \texorpdfstring{$w$}{w}-step (I-projection)}
\label{sec:i-projection}

\subsubsection{Unkernelized I-projection}
Suppose that $\Xcal$ is finite. Then the finite approximation to the OKL of \cref{eq:finiteokle} is given by the solution to the convex optimization problem
\[ \min_{\substack{w \in {\Delta}_n \\ \frac{1}{2} \|w-o\|_1 \leq \epsilon}} \sum_{i=1}^n w_i \log \frac{w_i n_i}{p_\theta(x_i)}, \]
where $n_i = |\{j : x_j = x_i \}|$ is the number of times that $x_i$ occurs in our sample.

This is a convex optimization problem for which there are a number of candidate solutions. Our approach is based on (consensus) Alternating Direction Method of Multipliers (ADMM)~\citep{boyd2011distributed, parikh2014proximal}. To frame the OKL optimization problem in the language of ADMM, we rewrite it as
\begin{equation}
\label{eqn:admm-form}
\min_{w \in \R^n} \overbrace{\sum_{i=1}^n w_i \log \frac{w_i n_i}{ p_\theta({x}_i)}}^{f_1(w)} + \overbrace{\ci{ w \in \Delta_n}}^{f_2(w)} + \overbrace{\ci{\frac{1}{2} \|w-o\|_1 \leq \epsilon }}^{f_3(w)},
\end{equation}
where $\ci{C}$ is 0 whenever the condition $C$ is true and $\infty$ otherwise. Written in this form, ADMM boils down to implementing the individual \emph{proximal operators} for the $f_i$'s, where the proximal operator of a closed proper convex function $f: \R^n \rightarrow \R$ and parameter $\lambda > 0$ is defined as the function
\[ \prox_{\lambda f}(x) = \argmin_{z \in \R^n} \left\{ f(z) + \frac{1}{2\lambda} \| z - x \|_2^2\right\}. \]
In our setting, the proximal operator for the KL term $f_1$ can be computed element-wise using the Lambert-W function~\citep{barratt2021optimal}, or in log-scale via the Wright omega function, for which fast algorithms exist~\citep{lawrence2012algorithm}, leading to $O(n)$ computational complexity. The proximal operator for $f_2$ corresponds to projection onto the $n$-dimensional scaled probability simplex, which can be computed in $O(n)$ time~\citep{condat2016fast}. The proximal operator for $f_3$ corresponds to projection onto an $\ell_1$-ball, and in fact it can be reduced to projection onto the simplex~\citep{condat2016fast}.

\begin{algorithm}[t]
\caption{Consensus ADMM to compute the I-projection}
\label{alg:consensus-admm}
\begin{algorithmic}
\STATE \textbf{Input:} Proximal operators $\prox_{\lambda f_i}$ for functions in \cref{eqn:admm-form}, dimension $n$, number of iterations $T$, and proximal penalty parameters $\lambda_i > 0$.
\STATE Initialize $w_i^0, y_i^0, z^0 = \vec{0} \in \R^n$ for $i=1, 2, 3$.
\FOR{$t=0,\ldots,T-1$}
	\STATE Set $w^{(i,t+1)} = \prox_{\lambda_i f_i}( z^{(t)} + \lambda_i y^{(i,t)})$ for $i=1, 2, 3$.
	\STATE Set $z^{(t+1)} = \frac{1}{\sum_{i=1}^k 1/\lambda_i}  \sum_{i=1}^k \left( w^{(i,t+1)}/\lambda_i - y^{(i,t)} \right)$.
	\STATE Set $y^{(i,t+1)} = y^{(i,t)} + (z^{(t+1)} - w^{(i,t+1)})/\lambda_i$ for $i=1, 2, 3$.
	\STATE \emph{Optionally, adjust $\lambda_i$ for $i=1, 2, 3$.}
\ENDFOR
\STATE \textbf{Output:} $w_1^T$. 
\end{algorithmic}
\end{algorithm}

To simplify our discussion on the implementation, we rewrite \cref{eqn:admm-form} for an arbitrary number of convex functions $f_1, \ldots, f_k: \R^n \rightarrow \R$ as
\[ \min_{\substack{w^{(1)}, \ldots, w^{(k)}, z \in \R^n \\ w^{(i)} = z \text{ for all }i}} \sum_{i=1}^k f_i\left(w^{(i)} \right) . \]
For penalty parameters $\lambda_1, \ldots, \lambda_k > 0$, the augmented Lagrangian associated with this problem is given by
\[ L\left(w^{(1:k)}, y^{(1:k)}, z \right) = \sum_{i=1}^k f_i\left(w^{(i)}\right) + \langle y^{(i)} , w^{(i)} - z \rangle + \frac{1}{2\lambda_i} \| w^{(i)} - z \|^2, \]
where $w^{(1)}, \ldots, w^{(k)} \in \R^n$ are the primal variables, $y^{(1)}, \ldots, y^{(k)} \in \R^n$ are the dual variables, and $z \in \R^n$ is the consensus variable.
Then the consensus ADMM algorithm is derived by iteratively optimizing the augmented Lagrangian coordinate-wise. That is, starting from some initialization $w^{(i,0)}, y^{(i,0)}, z^{(0)} \in \R^n$, we perform the following updates
\begin{align*}
w^{(i,t+1)} &= \argmin_{w^{(i)} \in \R^n} L\left(w^{(-i,t+1)}, w^{(i)}, y^{(1:k,t+1)}, z^{(t)} \right) = \prox_{\lambda_i f_i}\left(z^{(t)}  + \lambda_i y^{(i,t)} \right)\\
z^{(t+1)} &= \argmin_{z  \in \R^n} L\left(w^{(1:k,t+1)}, y^{(1:k,t+1)}, z \right) = \frac{1}{\sum_{i=1}^k 1/\lambda_i}  \sum_{i=1}^k \left( \frac{1}{\lambda_i}w^{(i,t+1)} - y^{(i,t)} \right) \\
y^{(i,t+1)} &= y^{(i,t)} + \frac{1}{\lambda_i} (z^{(t+1)} - w^{(i,t+1)}).
\end{align*}

\Cref{alg:consensus-admm} displays the full consensus ADMM algorithm. In our implementation, we use the self-adaptive rule suggested by \cite{he2000alternating} to independently update each of the penalty parameters.

\subsubsection{Kernelized I-projection}

In the Euclidean case with $\Xcal = \R^d$, beyond the sample $x_{1}, \ldots, x_n$ and parameter $\theta \in \Theta$, we additionally have a kernel matrix $K \in \R^{n \times n}$ with induced row sums $s_i = \sum_j K_{ij}$ and row-normalized matrix $A_{ij} = K_{ij}/s_i$. The optimization problem of \cref{eq:okle} translates to
\begin{equation}
\label{eqn:general-admm-form}
\min_{v \in \R^n} \sum_{i=1}^n (Av)_i \log \frac{(Av)_i s_i}{ p_\theta({x}_i)} + \ci{v \in \Delta^n } + \ci{  \frac{1}{2}\| Av-o\|_1 \leq \epsilon}.
\end{equation}
Although all the terms in \cref{eqn:general-admm-form} remain convex in $v$, the proximal operators for the KL objective and the $\ell_1$ constraint are no longer available in closed form. To circumvent this issue, we can rewrite problems of this form as
\[ \min_{\substack{w^{(1)}, \ldots, w^{(k)}, z \in \R^n \\ w^{(i)} = M_i z \text{ for all }i}} \sum_{i=1}^k f_i\left(w^{(i)} \right), \]
where each $M_i \in \R^{n \times n}$. The augmented Lagrangian in this case becomes
\[ L\left(w^{(1:k)}, y^{(1:k)}, z \right) = \sum_{i=1}^k f_i\left(w^{(i)}\right) + \langle y^{(i)} , w^{(i)} - M_i z \rangle + \frac{1}{2\lambda_i} \| w^{(i)} - M_i z \|^2, \]
which leads to the ADMM updates
\begin{align*}
w^{(i,t+1)} &= \prox_{\lambda_i f_i}\left(M_i z^{(t)}  + \lambda_i y^{(i,t)} \right)\\
z^{(t+1)} &= \left( \sum_{i=1}^k M_i^T M_i/\lambda_i \right)^{-1}  \sum_{i=1}^k M_i^T \left( w^{(i,t+1)}/\lambda_i - y^{(i,t)} \right) \\
y^{(i,t+1)} &= y^{(i,t)} + \frac{1}{\lambda_i} ( M_i z^{(t+1)} - w^{(i,t+1)}).
\end{align*}
In our setting, each $M_i$ is either equal to $A$ or to the identity matrix $I$. Thus, the matrix inverse in the update of $z$ can be computed efficiently using a single singular value decomposition of $A$, even if the penalty parameters change between iterations. To see this, suppose that $A = U \diag(\sigma_1, \ldots, \sigma_n) V^T$ is the SVD of $A$ and let $J = \{i : M_i = A \}$. Then the $z$ update can be written as
\[ z^{(t+1)} =  V \diag\left( (\beta_{J^c} + \sigma_1^2 \beta_{J})^{-1}, \ldots, (\beta_{J^c} + \sigma_n^2 \beta_{J})^{-1} \right) V^T \sum_{i=1}^k  M_i^T \left( w^{(i,t+1)}/\lambda_i - y^{(i,t)} \right), \]
where $\beta_{J} = \sum_{i \in J} \lambda_i^{-1}$ and  $\beta_{J^c} = \sum_{i \not \in J} \lambda_i^{-1}$.

\subsubsection{Computational complexity and generalization to other distances}
\label{sec:admm-complexity}

As mentioned above, each of the three proximal operators in \cref{eqn:admm-form} and \cref{eqn:general-admm-form} can be implemented in $O(n)$ time. For the non-kernelized version of the ADMM algorithm in \Cref{alg:consensus-admm}, the remaining linear algebraic steps can also be implemented in $O(n)$ time, bringing the total computational complexity of the procedure to $O(Tn)$, where $T$ is the number of ADMM steps taken. For the kernelized version of the ADMM algorithm, we require an SVD computation as a preprocessing step, which takes $O(n^3)$ time, and each step additionally involves matrix-vector products, taking time $O(n^2)$. Thus, the total time complexity of the kernelized ADMM procedure amounts to $O(n^3 + Tn^2)$, which is polynomial in $n$ but not feasible for large values of $n$.

We remark here that from an optimization perspective, for both \cref{eqn:admm-form,eqn:general-admm-form}, the $\ell_1$-distance is not special. We may substitute in any distance that is convex in its arguments and for which there are computationally efficient methods for projection onto the corresponding balls with a specified center and radius. This includes $\ell_2$-distance, whose projection can be computed by a simple shift and rescaling, and maximum mean discrepancy \citep{gretton2006kernel}
whose projection operator is onto an ellipsoid and can be reduced to a one-dimensional root finding problem~\citep{kiseliov1994algorithms}.

\subsection{Computing the \texorpdfstring{$\theta$}{theta}-step: maximizing weighted likelihoods}
\label{sec:weighted-likelihood}

We now shift focus to the second step in the OWL procedure, maximizing a weighted likelihood:
\[  \max_{\theta \in \Theta} \sum_{i=1}^n w_i \log p_\theta({x}_i).  \]

\subsubsection{Weighted maximum likelihood for exponential families}
\label{sec:wmle-for-exp-families}

Consider the setting where $p_\theta$ is an exponential family and $\Theta \subset \R^d$ is the natural parameter space so that
\[ p_\theta(x) = h(x) \exp \left( \theta^T T(x) - A(\theta) \right), \]
where $T: \Xcal \rightarrow \R^d$ is a sufficient statistic, $h$ is a base measure, and $A:\theta \rightarrow \R$ is the log-normalizing factor. Then for $w \in \Delta_n$, the weighted maximization step solves \begin{align}
\label{eqn:weighted-ml-expfam}
\max_{\theta \in \Theta} \sum_{i=1}^n w_i \log p_\theta({x}_i) = \min_{\theta \in \Theta} \left\{ A(\theta)  - \theta^T \sum_{i=1}^n w_i T(x_i) \right\}.
\end{align}
This solution satisfies the gradient condition $\nabla A(\theta) = \sum_{i=1}^n w_i T(x_i).$

For exponential families, whenever  $\cov(T(X))$ is positive definite for $X \sim p_\theta$, the function $\nabla A(\theta)$ is invertible. Thus, this step can be solved quickly whenever we can compute the inverse of $\nabla A(\cdot)$. When $\cov(T(X))$ is not strictly positive definite, or more generally when the inverse of $\nabla A(\cdot)$ is not available in closed form, the objective in \cref{eqn:weighted-ml-expfam} is still convex in $\theta$ and can be solved using tools from convex optimization.

\subsubsection{Weighted maximum likelihood for mixture models}
\label{sec:mixture-models}

In the mixture model setting, the parameters $\theta$ encode mixing weights $\pi \in \Delta_K$ and component parameters $\phi_1, \ldots, \phi_K \in \Phi$ such that $p_{\phi}$ denotes a probability density over $\Xcal$. Then the likelihood under $\theta$ can be written as
\[ p_\theta(x_{1:n}) =  \prod_{i=1}^n \left( \sum_{k=1}^k \pi_k p_{\phi_k}(x_i) \right). \]
To compute the maximum likelihood estimate, it is standard to introduce latent categorical random variables $z_i \stackrel{iid}{\sim} \pi$, $z_i \in \{1,\ldots,K\}, i=1,\ldots,n,$ and rewrite the likelihood as
\begin{align*}
p_\theta(x_{1:n}) =  \E_{z \sim \pi} \left[ \prod_{i=1}^n p_{\phi_{z_i}}(x_i) \right].
\end{align*}
The above can then be maximized using the EM algorithm~\citep{dempster1977maximum}. Unfortunately, the introduction of weights into the likelihood no longer allows for an easy decomposition via these latent variables. 

However, consider the likelihood with respect to the latent variables $\tilde{\theta} = (\theta, z_{1:n}) \in \tilde{\Theta}$:
\[ p_{\tilde{\theta}}(x_{1:n}) = \prod_{i=1}^n p_{\phi_{z_i}}(x_i) . \]
Then, following the setup from \Cref{sec:niid}, the weighted log-likelihood can be written as
\[ \sum_{i=1}^n w_i \log p_{\phi_{z_i}}(x_i) = \sum_{k=1}^K \sum_{i : z_i = k} w_i \log p_{\phi_{k}}(x_i).  \]
Thus, to maximize the above for fixed latent variables $z_1, \ldots, z_n$, we can maximize the weighted log-likelihood of the individual component parameters over its assigned (weighted) data. For many component distributions (e.g., Gaussians, Poissons), this can be computed in closed-form. On the other hand, for fixed component variables $\phi_1,\ldots, \phi_k$, we can maximize this weighted log-likelihood over the $z_1, \ldots, z_n$ by assigning each data point to the component that maximizes its individual likelihood. This suggests the following scheme, reminiscent of the `hard EM' algorithm~\citep{samdani2012unified}:
\begin{align*}
\phi_k^{t+1} &= \argmax_{\phi \in \Phi} \sum_{i: z^t_i = k} w_i \log p_{\phi}(x_i) \text{ for each } k=1,\ldots,K \\
z_i^{t+1} &= \argmax_{k \in \{1, \ldots, K \}} \log p_{\phi_k}(x_i) \text{ for each } i=1,\ldots,n .
\end{align*}
In general, optimizing over the augmented parameter space $\tilde{\Theta}$ implicitly assumes a generative model different from the one assumed by optimizing over $\Theta$. In particular, when one optimizes over $\Theta$, the assumed generative model is that the data come from some distribution $P_0$ that is $\epsilon$-close to $P_{\theta}$ for some $\theta \in \Theta$. The implied generative model when optimizing over $\tilde{\Theta}$ is that the data come from some mixture distribution $\sum_k \pi_k p_k$ such that each component $p_k$ is $\epsilon$-close to some $\phi_k \in \Phi$. For most settings, the first generative model is a strict generalization of the second. However, in the $\epsilon$-contamination model, these are equivalent, as implied by the following result.
\begin{proposition}
\label{prop:eps-contam-mixture-model}
Let $P_0, P_1, \ldots, P_K$ denote probability measures, $\pi \in \Delta_K$, and $\epsilon > 0$. Then the following are equivalent.
\begin{itemize}
	\item There exists a probability measure $Q$ such that $P_0 = (1- \epsilon) \sum_{k=1}^K \pi_k P_k + \epsilon Q$.
	\item There exist probability measures $Q_1, \ldots, Q_k$ such that $P_0 = \sum_{k=1}^K \pi_k ( (1- \epsilon) P_k + \epsilon Q_k )$.
\end{itemize}
\end{proposition}
The proof of \Cref{prop:eps-contam-mixture-model} follows from simple algebraic substitution.

\section{Asymptotics of the coarsened likelihood}
\label{sec:coarsened-likelihood-asymptotics}

In this section, we will show the asymptotic convergence of the coarsened likelihood to the OKL function.

\subsection{Asymptotics for finite spaces}
\label{sec:coarsened-likelihood-asymptotic-finite}
Suppose $\cX$ is a finite space. We will use the notation from \Cref{sec:asymp-finite}; in particular, recall the probability simplex $\Delta_{\cX} = \{q \in [0,1]^{\cX} | \sum_{x \in X} q(x) = 1\}$ and the OKL function $I_\epsilon(\theta)$ in terms of a general distance $\D$ on $\Delta_{\cX}$ as in \cref{eqn:okl-general-distance}. We begin by showing an elementary rounding lemma, which will be useful when applying Sanov's theorem in \Cref{lem:finite-sanov}.

\begin{lemma}
\label{lem:rounded-probability-vector}
Let $p, p_\theta \in \Delta_\Xcal$ such that $\KL(p|p_\theta) < \infty$. For any integer $n > |\Xcal|$, there exists a $q_n \in \Delta_\Xcal$ such that $n q_n(x)$ is integral for all $x \in \Xcal$, $\|q_n - p \|_1 \leq \frac{2|\Xcal|}{n}$, and 
\[ \KL(q_n | p_\theta) \leq \left(1+ \frac{|\Xcal|^2}{n}\right)\KL(p|p_\theta) + \frac{|\Xcal|}{n} \left( 2 \log \frac{n}{2} + \log |\Xcal| + \frac{|\Xcal|}{e} \right).\]
\end{lemma}
\begin{proof}
Choose an arbitrary $x^\star = \argmax_{x \in \cX} p(x)$, choosing arbitrarily if there are ties. Then we define $q_n \in \Delta_\Xcal$ as follows. For all $x \neq x^\star$, let $q_n(x) = \frac{1}{n} \lfloor n q_n(x) \rfloor$, and let $q_n(x^\star) = 1 - \sum_{x\neq x^\star}q_n(x)$. As $p \in \Delta_\Xcal$, we have $q_n \in \Delta_\Xcal$. Moreover, by construction we also have that $n q_n(x)$ is integral for all $x \in \Xcal$. 

Observe that for all $x \neq x^\star$, we have $0 \leq p(x) - q_n(x) \leq 1/n$. This implies that \[ 0 \leq q_n(x^\star) - p(x^\star) = \left(1 - \sum_{x \neq x^\star}q_n(x)\right) - \left(1-\sum_{x \neq x^\star}p(x)\right) \leq \frac{|\Xcal|}{n}.\]
Together, these statements give us that $\| p - q_n \|_1 \leq \frac{2|\Xcal|}{n}$. Further since $\KL(p|p_\theta) < \infty$, we must also have $\KL(q_n|p_\theta) < \infty$ since $\supp(q_n) \subseteq \supp(p) \subseteq \supp(p_\theta)$.  

To prove the more detailed KL bound, we first observe that by known bounds on the entropy function \citep[c.f. Theorem 17.3.3]{cover2006elements}, we have
\begin{align*}
    \sum_{x} q_n(x) \log q_n(x) 
    &\leq \sum_{x} p(x) \log p(x) + \| p - q_n \|_1 \log \frac{|\Xcal|}{\| p - q_n \|_1} \\
    &\leq \sum_{x} p(x) \log p(x) + \frac{2|\Xcal|}{n} \log \frac{n}{2}.
\end{align*} 
Next observe that we can bound the cross-entropy between $q_n$ and $p_\theta$ as follows.
\begin{align*}
\sum_x q_n(x) \log \frac{1}{p_\theta(x)} 
&= \sum_{x \neq x^\star} q_n(x) \log \frac{1}{p_\theta(x)} + q_n(x^\star) \log \frac{1}{p_\theta(x^\star)} \\
&\leq \sum_{x \neq x^\star} p(x) \log \frac{1}{p_\theta(x)} + (p(x^\star) + q_n(x^\star) - p(x^\star)) \log \frac{1}{p_\theta(x^\star)}\\
&\leq \sum_x p(x) \log \frac{1}{p_\theta(x)} + \frac{|\Xcal|}{n} \log \frac{1}{p_\theta(x^\star)}.
\end{align*}
Now observe that $p(x^\star)\geq 1/|\Xcal|$. This implies that $\log \frac{1}{p_\theta(x^\star)} \leq |\Xcal| (\KL(p | p_\theta) + 1/e) + \log |\Xcal|$. To see this, first note that if $p_\theta(x^\star) \geq 1/|\Xcal|$, then the claim is trivial. Thus, we may assume $p_\theta(x^\star) < 1/|\Xcal| \leq p(x^\star)$. The log-sum inequality implies that
\begin{align*}
\KL(p | p_\theta) 
&\geq p(x^\star) \log \frac{p(x^\star)}{p_\theta(x^\star)} + (1-p(x^\star)) \log \frac{1 - p(x^\star)}{1 - p_\theta(x^\star)} \\
&\geq p(x^\star) \log \frac{p(x^\star)}{p_\theta(x^\star)} - (1-p(x^\star)) \log \frac{1}{1 - p(x^\star)} \\
&\geq p(x^\star) \log \frac{p(x^\star)}{p_\theta(x^\star)} - \frac{1}{e} \geq \frac{1}{|\Xcal|} \log \frac{1}{|\Xcal| p_\theta(x^\star)} - \frac{1}{e}.
\end{align*}
Here, the last inequality follows from our bound on $p(x^\star)$ combined with the fact that $a\log \frac{a}{b}$ is an increasing function in $a$ for $a \geq b$. Rearranging the above gives us the claim.
\end{proof}

Our analysis of the coarsened likelihood will rely heavily on the following result, which is essentially a consequence of Sanov's theorem.

\begin{lemma}
\label{lem:finite-sanov} 
Suppose \Cref{assum:finite-continuous-distance,assum:finite-okl} hold, and let $r > \epsilon_0$ and $n \geq \frac{4C |\Xcal|}{r - \epsilon_0}$. Then 
\[  \left| \frac{1}{n} \log \prob_\theta\left(\D(\EmpDist{Z_{1:n}}, p_0) \leq r\right) + I_r(\theta) \right| \leq 
\frac{|\Xcal|}{n} \left( 3 \log(n+1) + \log|\Xcal| + |\Xcal|\left(V + \frac{1}{e} \right) + \frac{2C}{r - \epsilon_0} \right) \]
 where the probability operation $\prob_\theta$ is taken over random points $Z_1, \ldots, Z_n \in \cX$ drawn from the distribution $p_\theta \in \Delta_{\cX}$, and $\EmpDist{Z_{1:n}} \in \Delta_{\cX}$ is the empirical distribution of the data points $Z_1, \ldots, Z_n$.
\end{lemma}
\begin{proof}
Observe that $I_r(\theta) \leq V < \infty$ for all $r > \epsilon_0$. Sanov's theorem~\cite[Theorem~11.4.1]{cover2006elements} implies that
\[ \frac{1}{n} \log \prob_\theta\left(\D(\EmpDist{Z_{1:n}}, p_0) \leq r\right) \leq - I_r(\theta) + \frac{|\Xcal|}{n}\log(n+1). \]
Thus, we only need to show the lower bound for $\frac{1}{n} \log \prob_\theta\left(\D(\EmpDist{Z_{1:n}}, p_0) \leq r\right) + I_r(\theta)$. 
Pick $\delta > 0$ and for any $t > 0$, let $q_t \in \Delta_\Xcal$ satisfy $\D(q_t, p_0) \leq t$ and
\[ \KL(q_t | p_\theta) \leq I_t(\theta) + \delta. \]
Now let $\alpha_n = \frac{2C |\Xcal|}{n}$ and observe that for $n > \frac{2 C |\Xcal|}{r-\epsilon_0}$, such a $q_{r - \alpha_n} \in \Delta_{\cX}$ exists. Letting $q^{(n)}_{r - \alpha_n} \in \Delta_\Xcal$ be the discretization promised by \Cref{lem:rounded-probability-vector} and utilizing \Cref{assum:finite-continuous-distance}, we have
\begin{align*}
\D(q^{(n)}_{r - \alpha_n}, p_0) 
\leq \D(q^{(n)}_{r - \alpha_n}, q_{r - \alpha_n}) + \D(q_{r - \alpha_n}, p_0) 
\leq C \|q^{(n)}_{r - \alpha_n} - q_{r - \alpha_n}\|_1 + r - \alpha_n \leq r.
\end{align*}
Thus we have
\begin{align*}
&\hspace{-3em}\frac{1}{n} \log \prob_\theta\left(\D(\EmpDist{Z_{1:n}}, p_0) \leq r \right) \geq \frac{1}{n} \log \prob_\theta\left( \EmpDist{Z_{1:n}} = q_{r-\alpha_n}^{(n)}  \right)\\
&\geq -\frac{|\Xcal|}{n} \log (n+1) - \KL( q^{(n)}_{r - \alpha_n} | p_\theta) \\
&\geq - \left(1+ \frac{|\Xcal|^2}{n}\right)(I_{r - \alpha_n}(\theta)+\delta) - \frac{|\Xcal|}{n} \left( 3 \log(n+1) + \log |\Xcal| + \frac{|\Xcal|}{e} \right) \\
&\geq - I_{r - \alpha_n}(\theta) - \delta - \frac{|\Xcal|}{n}\left( 3 \log(n+1) + \log |\Xcal| + |\Xcal| \left( V + \frac{1}{e} + \delta \right) \right) \\
&\geq -I_r(\theta) - \delta - \frac{\alpha_n V}{r - \epsilon_0} - \frac{|\Xcal|}{n}\left( 3 \log(n+1) + \log |\Xcal| + |\Xcal| \left( V + \frac{1}{e} + \delta \right) \right),
\end{align*}
where the second line follows from~\cite[Theorem~11.1.4]{cover2006elements}, the third line follows from \Cref{lem:rounded-probability-vector}, and the last line follows from \Cref{lem:okl-continuity}. Rearranging the above and utilizing the fact that $\delta > 0$ was arbitrary gives us the lemma statement.
\end{proof}

With \Cref{lem:finite-sanov} in hand, we can prove the following convergence result for the coarsened likelihood.
\begin{theorem}
    \label{lem:finite-coarsened-likelihood-convergence}
    Suppose \Cref{assum:finite-continuous-distance,assum:finite-okl} hold, and let $\epsilon > \epsilon_0$. If $x_1,\ldots,x_n \sim p_0$, then with probability at least $1-\delta$, 
    \begin{align*}
    \left| \frac{1}{n} \log L_\epsilon(\theta|x_{1:n} ) + I_{\epsilon}(\theta)  \right| &\leq  \frac{C V |\Xcal|}{\epsilon - \epsilon_0} \sqrt{\frac{2}{n} \log \frac{2|\Xcal|}{\delta}} \\
    &\hspace{3em} + \frac{3|\Xcal|}{n} \left( 3 \log(n+1) + \log|\Xcal| + |\Xcal|\left(V + \frac{1}{e} \right) + \frac{4C}{\epsilon - \epsilon_0} \right).   
    \end{align*}
    whenever $n > \max \left\{2 \left( \frac{C |\Xcal|}{\epsilon - \epsilon_0} \right)^2 \log \frac{2|\Xcal|}{\delta}, \frac{8C |\Xcal|}{\epsilon - \epsilon_0} \right\}$.
\end{theorem}
\begin{proof}
For $r > 0$, define $M_{n,r}(\theta)= \prob_\theta\left(\D(\EmpDist{Z_{1:n}}, p^0) \leq r\right)$. By \Cref{lem:finite-sanov}, 
\begin{equation}
\label{eqn:intermediate-coarsened-okl-approx}
    \left| \frac{1}{n}\log M_{n,r}(\theta) + I_r(\theta) \right| \leq  \frac{|\Xcal|}{n} \left( 3 \log(n+1) + \log|\Xcal| + |\Xcal|\left(V + \frac{1}{e} \right) + \frac{2C}{r - \epsilon_0} \right)
\end{equation}
for all $r > \epsilon_0$ satisfying that $n \geq \frac{4C |\Xcal|}{r - \epsilon_0}$.

Now suppose that $x_1, \ldots, x_n \sim p_0$. Then Hoeffding's inequality combined with \Cref{assum:finite-continuous-distance} implies that with probability at least $1-\delta$,
\[ \D(\EmpDist{x_{1:n}}, p_0) \leq C \|\EmpDist{x_{1:n}} - p_0 \|_1 \leq  C|\Xcal| \sqrt{\frac{1}{2n} \log \frac{2|\Xcal|}{\delta}} =: \alpha_n.  \]
Let us condition on this event occurring. For $n > 2\left( \frac{C |\Xcal|}{\epsilon - \epsilon_0}\right)^2 \log \frac{2|\Xcal|}{\delta}$, we have $\alpha_n < (\epsilon - \epsilon_0)/2$. Thus, we may write
\begin{align*}
\left| \log M_{n,\epsilon}(\theta) -  \log L_\epsilon(\theta | x_{1:n} )\right| 
&= \left| \log \prob_\theta\left( \D(\EmpDist{Z_{1:n}}, p_0) \leq \epsilon \right) - \log \prob_\theta\left( \D(\EmpDist{Z_{1:n}}, \EmpDist{x_{1:n}}) \leq \epsilon \right) \right| \\
&= \log \max \left\{ \frac{\prob_\theta\left(\D(\EmpDist{Z_{1:n}}, p_0) \leq \epsilon\right)}{\prob_\theta\left(\D(\EmpDist{Z_{1:n}}, \EmpDist{x_{1:n}}) \leq \epsilon\right) },  \frac{\prob_\theta\left(\D(\EmpDist{Z_{1:n}}, \EmpDist{x_{1:n}}) \leq \epsilon\right)}{\prob_\theta\left(\D(\EmpDist{Z_{1:n}}, p_0) \leq \epsilon\right)} \right\} \\
&\leq \log \max \left\{  \frac{\prob_\theta\left(\D(\EmpDist{Z_{1:n}}, p_0) \leq \epsilon\right)}{\prob_\theta\left(\D(\EmpDist{Z_{1:n}}, p_0) \leq \epsilon - \alpha_n\right)},  \frac{\prob_\theta\left(\D(\EmpDist{Z_{1:n}}, p_0) \leq \epsilon + \alpha_n\right) }{\prob_\theta\left(\D(\EmpDist{Z_{1:n}}, p_0) \leq \epsilon\right)} \right\} \\
&= \log \max \left\{ \frac{M_{n,\epsilon}(\theta)}{M_{n,\epsilon-\alpha_n}(\theta)},  \frac{M_{n,\epsilon + \alpha_n}(\theta)}{M_{n,\epsilon}(\theta)} \right\}
\end{align*}
By \cref{eqn:intermediate-coarsened-okl-approx} and \Cref{lem:okl-continuity}, we have
\begin{align*}
&\frac{1}{n}\log\frac{ M_{n,\epsilon}(\theta)}{M_{n,\epsilon - \alpha_n}(\theta)} \\
&\leq I_{\epsilon}(\theta) - I_{\epsilon - \alpha_n}(\theta) + \frac{2|\Xcal|}{n} \left( 3 \log(n+1) + \log|\Xcal| + |\Xcal|\left(V + \frac{1}{e} \right) + \frac{2C}{\epsilon - \alpha_n - \epsilon_0} \right) \\
&\leq \frac{2\alpha_n}{\epsilon - \epsilon_0}V + \frac{2|\Xcal|}{n} \left( 3 \log(n+1) + \log|\Xcal| + |\Xcal|\left(V + \frac{1}{e} \right) + \frac{4C}{\epsilon - \epsilon_0} \right).
\end{align*}
Similarly, we also have
\begin{align*}
&\frac{1}{n}\log \frac{M_{n,\epsilon+\alpha_n}(\theta)}{M_{n,\epsilon}(\theta)} \\
&\leq I_{\epsilon + \alpha_n}(\theta) - I_{\epsilon}(\theta) + \frac{2|\Xcal|}{n} \left( 3 \log(n+1) + \log|\Xcal| + |\Xcal|\left(V + \frac{1}{e} \right) + \frac{2C}{\epsilon + \alpha_n - \epsilon_0} \right) \\
&\leq \frac{\alpha_n}{\epsilon - \epsilon_0 + \alpha_n}V + \frac{2|\Xcal|}{n} \left( 3 \log(n+1) + \log|\Xcal| + |\Xcal|\left(V + \frac{1}{e} \right) + \frac{4C}{\epsilon - \epsilon_0} \right).
\end{align*}
Putting it all together gives us the theorem statement.
\end{proof}
 
\subsection{Asymptotics for  continuous spaces}
\label{sec:coarsened-likelihood-asymptotic-cont}

In this section, we will show that the coarsened likelihood continues to converge in probability to the OKL when $\cX$ is a metric space. More precisely, let $\cX$ be a Polish space (i.e.~a complete and separable metric space) equipped with its Borel sigma algebra $\cB(\cX)$. Then $\cP(\cX)$, the set of (Borel) probability measures on $\cX$, can be equipped with the topology of weak-convergence \citep{billingsley2013convergence}. In more detail, we say that a sequence of measures $\{\cP_n\}_{n \in \nat}  \subseteq \cP(\cX)$ weakly converges to a measure $P$, denoted as $P_n \dconv P$, if $\lim_{n \to \infty} \int f dP_n = \int f dP$ for each continuous and bounded function $f: \cX \to \R$. The space $\cP(\cX)$ is also a Polish space under this topology \citep{billingsley2013convergence}. 

 Recall the definition of the OKL for a general distance $\D$ over $\cP(\cX)$ from \cref{eqn:okl-general-distance}:
\[ \okl[\epsilon] = \inf_{\substack{Q \in  \cP(\cX) \\ \D(Q, P_0) \leq \epsilon}} \KL(Q | P_\theta). \]
To establish an asymptotic connection between the coarsened likelihood and the OKL, we will make the following assumption on $\D$.
\begin{assume}
\label{assump:general-distance}
    For any $P,Q, R \in \cP(\cX)$ and $\lambda \in (0,1)$, the following holds:
    \begin{enumerate}[(a)]
		\item $\D(P, P) = 0$.
		\item $\D(P, Q) = \D(Q, P)$.
		\item $\D(P, Q) \leq \D(P, R) + \D(R, Q)$.
		\item $\D$ is convex in its arguments: $\D(P, (1-\lambda)Q + \lambda R) \leq \lambda \D(P, Q) + (1-\lambda) \D(P, R)$.
		\item For any sequence of probability measures $P_n \dconv P$, $\D(P_n, Q) \to \D(P, Q)$ as $n \to \infty$.
	\end{enumerate}
\end{assume}
Conditions~(a)-(c) simply require that $\D$ is a pseudometric, i.e. it satisfies all the requirements of a metric except for the requirement that $\D(P,Q) > 0$ whenever $P \neq Q$. Condition~(d) (combined with condition~(a)) allows us to apply \Cref{lem:okl-continuity}, and condition~(e) ensures that $\D$ is continuous with respect to the topology of weak convergence.

Given this assumption on $\D$, we can demonstrate the following asymptotic convergence.

\begin{theorem}
\label{thm:clikelihood-asymptotics}
Suppose that \Cref{assump:general-distance} holds and that $\okl[\epsilon_0] < \infty$ for some $\epsilon_0 \geq 0$. For any $\epsilon > \epsilon_0$, if $x_1, \ldots, x_n \iid P_0$, then 
\begin{equation*}
	\frac{1}{n} \log L_\epsilon(\theta|x_{1:n}) \pconv -\okl[\epsilon]
\end{equation*}
as $n \to \infty$.
\end{theorem}

\subsubsection{Proof of \texorpdfstring{\Cref{thm:clikelihood-asymptotics}}{thmclike}}

Similar to the setting in \Cref{sec:coarsened-likelihood-asymptotic-finite}, we will study asymptotics of the coarsened likelihood $\owl$ by first studying the asymptotics of its population level analog  
\begin{equation}
	\label{eq:pop-clike}
	\owlM[\epsilon] = \prob_\theta(\D(\EmpDist{Z_{1:n}}, P_0)\leq \epsilon)
\end{equation}
obtained by replacing the  empirical distribution of the data  $\EmpDist{x_{1:n}}$ by the  population level quantity $P_0$.

Next, to study the asymptotics of $\owlM$, we will invoke Sanov's theorem from Large Deviation theory which says that the law of the empirical distribution $\EmpDist{Z_{1:n}}$  satisfies a Large Deviation principle in the space $\cP(\cX)$ with rate function $\mu \mapsto \KL(\mu|P_\theta)$, when $Z_1, \ldots, Z_n \iid P_\theta$. More precisely, we will show that the error term 
\begin{equation}
	\label{eq:sanov-error}
	\lderr = \left|\frac{1}{n}\log \owlM + \okl\right|
\end{equation}
converges to zero as $n \to \infty$.

\begin{lemma}
If \Cref{assump:general-distance} holds and $\okl[\epsilon_0] < \infty$ for some $\epsilon_0 \geq 0$, then for any $\epsilon > \epsilon_0$, $\lim\limits_{n \to \infty} \lderr = 0$.
\label{lem:sanov}
\end{lemma}
\begin{proof} 
	
	Assume $Z_1, \ldots, Z_n \iid P_\theta$. Then Sanov's  theorem \cite[Theorem 6.2.10]{demboLargeDeviationsTechniques2010} shows that for any Borel measurable subset $\Gamma \subseteq \cP(\cX)$,
	\begin{equation}
		\label{eq:sanov}
		\begin{aligned}
			-\inf_{Q \in \Gamma^\circ} \KL(Q|P_\theta) &\leq \liminf_{n \to \infty} \frac{1}{n} \log \prob_\theta(\EmpDist{Z_{1:n}} \in \Gamma)\\ 
			&\leq \limsup_{n \to \infty} \frac{1}{n} \log \prob_\theta(\EmpDist{Z_{1:n}} \in \Gamma) \leq -\inf_{Q \in \bar{\Gamma}} \KL(Q|P_\theta),
		\end{aligned}
	\end{equation}
	where $\Gamma^\circ$ and $\bar{\Gamma}$ refer to the interior and closure of $\Gamma$ under the weak-topology on $\cP(\cX)$.

	Now consider the set $\Gamma_t = \{Q \in \cP(\cX) | \D(Q, P_0) \leq t \}$ indexed by the parameter $t \geq 0$. By the assumed continuity of $\D$ under the topology of weak-convergence (\Cref{assump:general-distance}(e)), the set  $\Gamma_\epsilon$ is closed while the set $\Gamma_{\epsilon^-} = \{Q \in \cP(\cX) \mid \D(Q,P_0) < \epsilon\}$ is an open set. Thus using $\bar{\Gamma}_{\epsilon} = \Gamma_\epsilon$ and $\cup_{r < \epsilon} \Gamma_r = \Gamma_{\epsilon^-}\subseteq \Gamma_\epsilon^\circ$, we see 
	\begin{equation*}
		\begin{aligned}
			\okl[\epsilon] = 
			\inf_{Q \in \bar{\Gamma}_\epsilon} \KL(Q|P_\theta) \leq \inf_{Q \in \Gamma^\circ_\epsilon}  \KL(Q|P_\theta)
			\leq \inf_{Q \in \Gamma_{\epsilon^-}} \KL(Q|P_\theta)  \leq \okl[r]
		\end{aligned}
	\end{equation*}
	for each $0 < r < \epsilon$. Next by the condition $\okl[\epsilon_0] < \infty$, we may apply \Cref{lem:okl-continuity} to conclude that the map $t \mapsto \okl[t]$ is continuous at $t = \epsilon$, i.e. $\lim_{t \to \epsilon} \okl[t] = \okl[\epsilon]$. Hence letting $r$ increase to $\epsilon$ in the above display, we find that
	$$
	\inf_{Q \in \bar{\Gamma}_\epsilon} \KL(Q|P_\theta) = \inf_{Q \in \Gamma^\circ_\epsilon} \KL(Q|P_\theta) = \okl[\epsilon].
	$$
	Thus taking $\Gamma = \Gamma_\epsilon$ in \cref{eq:sanov} shows that the limit
	$$
	\frac{1}{n} \log \prob_\theta(\D(\EmpDist{Z_{1:n}},P_0) \leq \epsilon) = \frac{1}{n} \log \prob_\theta(\EmpDist{Z_{1:n}} \in \Gamma_\epsilon) \to - \okl
	$$
	holds as $n \to \infty$. Recalling  the definition of $\owlM$ in \cref{eq:pop-clike}, we see that the asymptotic error term $\lderr$ in \cref{eq:sanov-error} is converging to zero as $n \to \infty$.

\end{proof}

Now we can prove \Cref{thm:clikelihood-asymptotics} by carefully accounting for the error between the coarsened likelihood  $\owl$ and its population-level analog $\owlM$ from  \Cref{eq:pop-clike}.

\begin{proof}[Proof of \Cref{thm:clikelihood-asymptotics}] Let $x_1, \ldots, x_n \iid P_0$ and pick $0 < t < \epsilon - \epsilon_0$. Define $E_{n,t}$ as the event that $\D(\EmpDist{x_{1:n}}, P_0) \leq t$. Since $\D$ is continuous with respect to weak convergence and $\EmpDist{x_{1:n}} \dconv P_0$ as $n \to \infty$ by the weak law of large numbers, we have $\lim_{n \to \infty} \prob(E_{n,t}) = 1$ for any $t > 0$. By \Cref{assump:general-distance}(c) we have $|\D(\EmpDist{Z_{1:n}}, P_0) -  \D(\EmpDist{Z_{1:n}}, \EmpDist{x_{1:n}})| \leq \D(P_0, \EmpDist{x_{1:n}})$, implying that on the event $E_{n,t}$ we have
\begin{equation}
	\label{eq:owlboundsonEnt}
	\owlM[\epsilon-t] \leq \owl \leq \owlM[\epsilon+t].
\end{equation}
Thus on the event $E_{n,t}$, we can bound
\begin{align*}
	\abs*{\frac{1}{n} \log \owl + \okl} 
	&\leq \abs*{\frac{1}{n} \log \owlM[\epsilon+t] + \okl}+ \abs*{\frac{1}{n} \log \owlM[\epsilon-t] + \okl}\\
	&\leq \lderr[\epsilon+t] + \lderr[\epsilon-t] + \abs{\okl[\epsilon+t] - \okl}+\abs{\okl[\epsilon-t] - \okl}\\
	&\leq \lderr[\epsilon+t] +  \lderr[\epsilon-t] + \frac{2t}{\epsilon - \epsilon_0} \okl[\epsilon_0]
\end{align*}
where the first inequality uses  \cref{eq:owlboundsonEnt}, the second uses \cref{eq:sanov-error}, and the last uses \Cref{lem:okl-continuity}. Hence given any $\delta > 0$ and $t \in (0,\epsilon - \epsilon_0)$ such that $t \okl[\epsilon_0] <  (\epsilon - \epsilon_0) \delta/4$, we have
\begin{align*}
	\prob\brR*{\abs*{\frac{1}{n} \log \owl + \okl} > \delta} \leq \prob(E_{n,t}^c) + \I{\lderr[\epsilon-t] + \lderr[\epsilon+t] > \delta/2}.
\end{align*}
Hence using \Cref{lem:sanov} and the fact that $\prob(E_{n,t}) \to 1$ as $n\to \infty$, we now see that 
$$
\lim_{n \to \infty} \prob\brR*{\abs*{\frac{1}{n} \log \owl - \okl} > \delta} = 0.
$$ 
Since $\delta > 0$ was arbitrary, this completes the proof.
\end{proof}

\subsection{The smoothed total variation distance}
\label{sec:smoothed-tvd}

In this section, we continue the discussion about smoothed total-variation distance from \Cref{sec:asymp-continuous} when $\cX = \R^d$ and $\Den$ denotes the set of densities on $\cX$ with respect to the Lebesgue measure.
To formally define the smoothed TV distance, let $\phi \in \Den$ be a continuous and bounded probability density function (e.g.~standard Gaussian density), let $h > 0$ be a bandwidth parameter. Then the kernel $K_h: \cX \times \cX \to [0,\infty)$ is defined as  $K_h(x,y) = \frac{1}{h^d} \phi((x-y)/h)$, and for any measure $\mu \in \cP(\cX)$, the convolved density $K_h \star \mu \in \Den$ is defined as $(K_h \star \mu)(x) = \int K_h(x, y) \mu(dy)$.  

\begin{definition} Given two measures $\mu, \nu \in \cP(\cX)$ and bandwidth $h > 0$, the \emph{smoothed total variation (TV) distance} is defined as:
	$$
	\tvk[h](\mu,\nu) = \frac{1}{2} \int |(K_h \star \mu)(x) - (K_h \star \nu)(x)| dx.
	$$
 
	We extend the notion of smoothed TV distance $\tvk[h](p,q) = \tvk[h](\mu, \nu)$ to densities $p, q \in \Den$\ based on their  induced measures $\mu, \nu \in \cP(\cX)$.
 \label{def:smoothed-tvd}
\end{definition}

We show in \Cref{sec:smoothed-tvd-is-continous-ips} that $\D=\tvk[h]$  satisfies conditions of \Cref{thm:clikelihood-asymptotics-density}.
Further, when $\phi$ has fast tail-decay and densities $p, q \in \Den$ satisfy appropriate regularity conditions, standard results on kernel density estimation (e.g.~\cite{rinaldo2010generalized, jiang2017uniform}) show the pointwise convergence of densities  $K_h \star q \to q$ as $h \to 0$. This, when combined with Scheffe's lemma and the triangle inequality, shows that $\lim_{h \to 0} \tvk[h](p, q) = \tv(p,q)$. In other words, for suitably small bandwidth parameter  $h > 0$, the neighborhoods based on the smoothed total variation distance    approximate those based on the total variation distance.

Thus, by invoking \Cref{thm:clikelihood-asymptotics-density} with the choice $\D=\tvk[h]$, one expects $- \frac{1}{n} \log L_\epsilon(\theta|x_{1:n}) \approx I_\epsilon(\theta)$ when $n$ is large and $h$ is small. As in the finite setting, we again see that maximizing the coarsened likelihood is closely related to minimizing the OKL function in the large sample regime. Hence the OWL methodology can be used to approximately maximize the coarsened likelihood when $\D=\tvk[h]$  for large sample size $n$ and a suitably small bandwidth $h$. In fact for many other metrics $\D$ satisfying the conditions of \Cref{thm:clikelihood-asymptotics-density}, one can adapt the OWL methodology to maximize the function $\theta \mapsto L_\epsilon(\theta|x_{1:n})$ as $n \to \infty$.

\subsubsection{Smoothed total variation distance satisfies \texorpdfstring{\Cref{assump:general-distance}}{assgenraldist}}
\label{sec:smoothed-tvd-is-continous-ips}

Here we introduce the smoothed-TV distance for a general metric space $\cX$, and show that it is a convex-pseudo metric that is continuous with respect to the weak convergence topology. This shows that one may take $\D$ in \Cref{thm:clikelihood-asymptotics-density} to be the  smoothed TV distance.

Suppose $\K: \cX \times \cX \to [0,\infty)$ is a probability kernel with respect to measure $\lambda$. That is, assume $\int \K(x,y) d\lambda(x) = 1$ for each $y \in \cX$.  Given a measure $\mu \in \cP(\cX)$, this allows us to define a smoothed probability measure $\K \star \mu \in \cP(\cX)$ that has density  

\begin{equation*}
	\frac{d(\K \star \mu)}{d\lambda} (x) =  f_{\K, \mu}(x) = \int \K(x, y) d\mu(y)
\end{equation*}
with respect to $\lambda$.

Recall the definition of the total variation distance on $\cP(\cX)$, 
\begin{equation}
	\label{eq:tvd}
	\tv(\mu, \nu) = \sup_{B \in \cB(\cX)}|\mu(B)-\nu(B)| = \sup_{g: \cX \to [-1,1]} \abs*{\int g d\mu - \int g d\nu}.
\end{equation}
When $\mu$ and $\nu$ have densities $f_\mu$ and $f_\nu$ with respect to a common measure $\lambda$, one can additionally show $\tv(\mu, \nu) = \frac{1}{2}\int  \abs{f_\mu-f_\nu} d\lambda$. 

Although $\tv$ is not continuous with respect to weak convergence, we can show that the following kernel-smoothed version of TV distance is.
\begin{definition} Given a probability density kernel $\K$, the smoothed total-variation distance is given by  
	\begin{equation*}
		\tvk(\mu, \nu) = \tv(\K \star \mu, \K \star \nu) = \frac{1}{2} \int \abs{f_{\K,\mu}(x) - f_{\K, \nu}(x)} d\lambda(x)
	\end{equation*}
\end{definition}
Now we will show that $\tvk$ satisfies \Cref{assump:general-distance} when $\K$ is a bounded and continuous kernel. 
\begin{proposition}
If $\K$ is a bounded and continuous kernel, then $\tvk$ satisfies \Cref{assump:general-distance}.
\end{proposition}
\begin{proof}
First observe that the smoothed TV distance is just an ordinary TV distance between smoothed densities.
Since the ordinary TV distance is a metric, this immediately implies that the smoothed TV  satisfies the identity, symmetry, and triangle inequality properties. 

To establish convexity of $\tvk$, note that for measures $\mu, \nu \in \cP(\cX)$ and $v \in [0,1]$, we have
\[f_{\K, (1-v)\mu + v\nu}(x) = (1-v) \int \K(x, y) d\mu(y) + v  \int \K(x, y) d\nu(y) = (1-v)f_{\K, \mu}(x) + v f_{\K, \nu}(x) .  \]
Thus, for $\mu, \nu, \pi \in \cP(\cX)$, we have
\begin{align*}
\tvk(\pi, (1-v)\mu + v\nu) 
&= \frac{1}{2}\int \abs{f_{\K,\pi}(x) - f_{\K, (1-v)\mu + v\nu}(x)} d\lambda(x) \\
&= \frac{1}{2} \int \abs{f_{\K,\pi}(x) - ((1-v)f_{\K, \mu}(x) + v f_{\K, \nu}(x))} d\lambda(x) \\
&\leq (1-v) \frac{1}{2} \int \abs{f_{\K,\pi}(x) - f_{\K, \mu}(x)} d\lambda(x) + v \frac{1}{2}\int \abs{f_{\K,\pi}(x) - f_{\K, \nu}(x)} d\lambda(x) \\
&= (1-v) \tvk(\pi, \mu )  + v \tvk(\pi,\nu),
\end{align*}
where the inequality follows from the convexity of the absolute value function. Thus, $\tvk$ is convex in its arguments.

To establish the continuity of $\tvk$ under the topology of weak-convergence, suppose $\mu_n \dconv \mu$ in $\cP(\cX)$. Then since $y \mapsto \K(x,y)$ is a continuous and bounded function, the convergence 
$$
f_{\K, \mu_n}(x) \to f_{\K, \mu}(x)
$$
follows for each $x \in \cX$ as $n \to \infty$. Since $f_{\K, \nu}$ for any $\nu \in \cP(\cX)$ is a density with respect to $\lambda$, Scheff\'e's lemma shows that
$$
\tvk(\mu_n, \mu) = \frac{1}{2} \int \abs{f_{\K, \mu_n}(x) - f_{\K, \mu}(x)}d \lambda(x) \to 0  \ \ \text{ as } n \to \infty.
$$
Finally, by the triangle inequality, $\abs{\tvk(\mu_n, \nu) - \tvk(\mu, \nu)} \leq \tvk(\mu_n, \mu)$.
\end{proof}

\section{OWL and MM/DC algorithms}  
\label{sec:owl-and-MM-DC}
In general, finding the global minimizer of the function $\theta \mapsto \hatI(\theta)$ is difficult, since this function can be non-convex even in the simplest case when the model densities $\{p_\theta\}_{\theta \in \Theta}$ follow a natural exponential family (\Cref{sec:owl-optimization}). 
However, the OWL algorithm (\Cref{alg:owl}) exploits a special form of the objective $\hatI$ to reduce this optimization to a sequence of simpler optimization problems. We now discuss some theoretical properties of the iterates $\{(\theta_t, w_t)\}_{t \geq 1}$ from the OWL algorithm.

First, it is straightforward to see that the OWL iterations must decrease the objective function $\theta \mapsto \hatI(\theta)$ at each step, as we have
\begin{align*}
    \hatI(\theta_{t+1}) 
    &= \sum_{i=1}^n w_{t+1,i} \log \frac{n w_{t+1,i} \hat{p}(x_i)}{ p_{\theta_{t+1}}({x}_i)} 
	\leq \sum_{i=1}^n w_{t,i} \log \frac{n w_{t,i} \hat{p}(x_i)}{ p_{\theta_{t+1}}({x}_i)} \\
	&
	\leq \sum_{i=1}^n w_{t,i} \log \frac{n w_{t,i} \hat{p}(x_i)}{ p_{\theta_{t}}({x}_i)} = \hatI(\theta_t).
\end{align*}
Since \cref{eq:okle} is a convex optimization problem with a strictly convex objective, the first inequality above is strict unless $w_{t+1} = w_{t}$. This is reminiscent of the Majorization Minimization (MM) class of algorithms \citep{hunter2004tutorial} (of which EM is a special case). Indeed, let $w(\theta) \in \hat{\Delta}_n$ denote the solution for the optimization in \cref{eq:okle}. Then the iterates $\{\theta_t\}_{t \geq 1}$ can be seen as arising from the instance of the MM algorithm based on the majorant $Q(\theta|\theta') =  \sum_{i=1}^n w_i(\theta')\log \frac{n w_{i}(\theta') \hat{p}(x_i)}{ p_{\theta}({x}_i)} \geq \hatI(\theta)$ for every $\theta, \theta' \in \Theta$, with an equality when $\theta = \theta'$. In \Cref{sec:owl-optimization}, we further show that the iterates $\{w_t\}_{t \geq 1}$  arise as an instance of the well-studied class of DC algorithms \citep{le2018dc} that use convex optimization techniques to minimize non-convex objectives that are expressed as a difference of two convex functions. Hence, theoretical results surrounding the MM \citep{lange2021nonconvex, kang2015global} and  DC algorithms \citep{le2018dc, de2020abc} could be used to study properties of the OWL iterates $\{(\theta_t, w_t)\}_{t \geq 1}$ in any given instance, and we leave this theoretical analysis as a direction for future. 
\subsection{DC programming}
\label{sec:owl-optimization}
DC programming \citep{le2018dc,de2020abc} is a well-studied class of non-convex optimization techniques that use tools from convex optimization to minimize objective functions of the form $f(x) = g(x)-h(x)$ over $x \in X$, where $f$ and $g$ are two lower-semicontinuous and proper convex functions mapping $X \to (-\infty, \infty]$.  While the  class of functions $f$ that can be represented in this way is large, containing in particular all sufficiently smooth functions and being closed under linear combinations and finite maximums, DC programming algorithms  crucially rely on knowing the constituent functions $g$ and $h$. 

In this section, we show how the OWL procedure (\Cref{alg:owl}) and the function $\theta \mapsto \hatI(\theta)$ fall into the framework of DC programming. We have two results in this regard. First when the family of models $\{p_\theta\}_{\theta \in \Theta}$ follows an exponential family, we show that $\hatI(\theta)$ is a difference of two convex functions of $\theta$, and that the iterates $\{\theta_t\}_{t \geq 1}$ from \Cref{alg:owl} coincide with the DC Algorithm \citep{le2018dc} for this case. Secondly, even if the model density does not follow an exponential family, we show that the iterates $\{w_t\}_{t \geq 1}$ from \Cref{alg:owl} can always be regarded as an instance of the DC Algorithm \citep{le2018dc}. 

\begin{lemma}
	\label{lem:linear-dc}
	Suppose $\Theta \subseteq \R^p$ is a convex subset, and the model $\{p_\theta\}_{\theta \in \Theta}$ follows a natural exponential family with sufficient statistic $T: \cX \to \R^p$. In particular, suppose that $p_\theta(x) = \exp \left( \langle \theta, T(x) \rangle - A(\theta) \right)$. Then the unkernalized OKL estimator $\hatI$ (i.e.~\cref{eq:okle} with $\K(x,y) = \I{x = y}$) can be written as 
	$$
	\hatI(\theta)  = A(\theta) - B(\theta) 
	$$
	where both $A(\theta)$ and 
	\begin{equation}
    \label{eq:Bfunc}
	B(\theta) \doteq \max_{\substack{w \in \Delta_n \\ \|w-o\|_1 \leq 2\epsilon}} \left\{ - \sum_{i=1}^n w_i \log w_i + \langle \theta, \sum_{i=1}^n w_i T(x_i )\rangle\right\} 
	\end{equation}
	are convex functions in the parameter $\theta$. Moreover, iterates $\{\theta_t\}_{t \geq 1}$ of the DC Algorithm \citep{le2018dc} based on the representation $\hatI(\theta)  = A(\theta) - B(\theta)$, given by
	$$
	\theta_{t+1} = \argmin_{\theta \in \Theta} \left\{ A(\theta) - \langle \nabla B(\theta_t), \theta \rangle \right\}
	$$ 
	coincide with the iterates $\{\theta_t\}_{t \geq 1}$ obtained from \Cref{alg:owl}.
\end{lemma} 
\begin{proof}
	First let us show that $\hatI(\theta)  = A(\theta) - B(\theta)$. Indeed from \cref{eq:okle}
    \begin{align}
        \hatI(\theta) &= \min_{\substack{w \in \Delta_n \\ \frac{1}{2}\|w-o\|_1 \leq \epsilon}} \left\{ \sum_{i=1}^n w_i \log w_i  - \sum_{i=1}^n w_i (\langle \theta, T(x) \rangle - A(\theta)) \right\} \label{eq:okl-rewrite-exp-family}\\ 
                    &= A(\theta) + \min_{\substack{w \in \Delta_n \\ \frac{1}{2}\|w-o\|_1 \leq \epsilon}} \left\{ \sum_{i=1}^n w_i \log w_i  - \sum_{i=1}^n w_i \langle \theta, T(x) \rangle  \right\}  = A(\theta) - B(\theta) \nonumber
    \end{align}
    since we have used that $\sum_{i=1}^n w_i = 1$ for any $w \in \Delta_n$. Note that $A$ is convex in $\theta$ since it is the log-partition function of a natural exponential family, while $B$ is convex in $\theta$ since it is the pointwise maximum of linear functions $f_w(\theta) = \sum_{i=1}^n w_i \log w_i  - \sum_{i=1}^n w_i \langle \theta, T(x) \rangle$ in $\theta$. Further, since the optimization in \cref{eq:Bfunc} involves maximization of a bounded strictly concave objective function on a convex set, this optimization problem has a unique maximizer $w(\theta) \in \Delta_n$. Note that $w(\theta)$ is also the minimizer in \cref{eq:okl-rewrite-exp-family}. Hence Danskin's theorem \citep[e.g.][]{nonlinearprog} shows that $B(\theta)$ is differentiable at any $\theta \in \Theta$ with
	$$
	\nabla B(\theta) = \sum_{i=1}^n w_i(\theta) T(x_i).
	$$
	Since $w(\theta) \in \Delta_n$, the iterates of the DC Algorithm can be re-written as
    \begin{align*}
	   \theta_{t+1} &= \argmin_{\theta \in \Theta} \left\{ A(\theta) - \langle \sum_{i=1}^n w_i(\theta_t) T(x_i), \theta \rangle \right\} = \argmax_{\theta \in \Theta} \left\{ \sum_{i=1}^n  w_i(\theta_t) \log p_\theta(x_i)\right\},
    \end{align*}
    which is seen to coincide with \Cref{alg:owl} by noting that $w_t = w(\theta_t)$.
\end{proof}

\begin{remark} The connection to DC Algorithm provides several insights about the iterates $\{\theta_t\}_{t > 1}$ in \Cref{lem:linear-dc}. 
	\begin{enumerate}
        \item Under suitable conditions \citep{abbaszadehpeivasti2023rate}, the sequence $\{\theta_t\}_{t \geq 1}$ will converge  at a linear rate to a stationary point $\theta^* \in \Theta$ satisfying $\nabla \hat{I}_\epsilon(\theta^*) = 0$.
        \item For the exponential family in \Cref{lem:linear-dc},  the condition $\nabla \hat{I}_\epsilon(\theta^*) = 0$ is equivalent to 
        $$
        \sum_{i=1}^n w_i(\theta^*) T(x_i) = \E_{\theta^*} T(X),
        $$ 
        where $\E_{\theta}$ is the expectation under $X \sim p_\theta$ and $w(\theta)$ is the minimizer in \cref{eq:okl-rewrite-exp-family}. Thus the stationary points of $\hat{I}_\epsilon(\theta)$ are points $\theta$ at which the empirical re-weighted moment of the sufficient statistic $\sum_{i=1}^n w_i(\theta) T(x_i)$ matches the population moment $\E_{\theta} T(X)$.
        
        \item Finally, $\theta^*$ also approximates the \emph{global minimizer} of $\theta \mapsto \hat{I}_\epsilon(\theta)$ in the sense that $\theta^*$ is an optimal solution to the surrogate problem \citep{le2018dc}:
        $$
        \inf_{\theta} \left\{ A(\theta) - \tilde{B}(\theta)\right\}
        $$
        where $\tilde{B}(\theta) = \sup_{t \in \nat} \left\{ - \sum_{i=1}^n w_{i,t} \log w_{i,t} + \langle \theta, \sum_{i=1}^n w_{i,t} T(x_i )\rangle\right\}$ is a convex under-approximation of $B$ based only on iterates $\{w_t\}_{t \geq 1}$ of \Cref{alg:owl}. Thus for global minimization of $\hat{I}_\epsilon(\theta)$, one can initialize \Cref{alg:owl} by using random re-starts so that $\tilde{B}(\theta)$ approximates $B(\theta)$ well.
	\end{enumerate}
    \label{rem:owl-iterates-exponential-family}
\end{remark}

Similar to \Cref{lem:linear-dc}, the next result exploits another DC representation for the OWL minimization problem $\min_{\theta \in \Theta} \hatI(\theta)$ to show that the iterates $\{w_t\}_{t \geq 1}$ in \Cref{alg:owl} also follow the DC Algorithm. Importantly, this equivalence does not assume that the model follows an exponential family, and holds for a wide class of models $\{p_\theta\}_{\theta \in \Theta}$.

\begin{lemma}  The iterates $\{w_t\}_{t \geq 0}$ in \Cref{alg:owl} coincide with that of the DC Algorithm:
    \begin{equation}
    \label{eq:owl-dc2}
	w_{t+1} = \argmin_{\substack{w \in \hat{\Delta}_n \\ \|w-o\|_1 \leq 2\epsilon}} \left\{ \sum_{i=1}^n w_i \log (n w_i \hat{p}(x_i)) - \langle w, v_{t+1} \rangle \right\}
    \end{equation}	
	where $v_{t+1} \in \partial h(w_t)$ is a suitable sub-gradient at $w=w_t$ of the convex function 
    \begin{equation}
    \label{eq:okl-dc-rep2}
	h(w) = \max_{\theta \in \Theta} \left\{ \sum_{i=1}^n w_i \log p_\theta(x_i)\right\}.
    \end{equation}
\end{lemma}
\begin{proof}
    By exchanging the order of minimization, the OWL optimization problem can be expressed as the following DC program:
	\begin{equation*}
	    	\min_{\theta \in \Theta} \hatI(\theta) = \min_{\substack{w \in \hat{\Delta}_n \\ \|w-o\|_1 \leq 2\epsilon}} \min_{\theta \in \Theta} \left\{ \sum_{i=1}^n w_i \log \frac{n w_i \hat{p}(x_i)}{p_\theta(x_i)}\right\}  = \min_{\substack{w \in \hat{\Delta}_n \\ \|w-o\|_1 \leq 2\epsilon}} \left\{ \sum_{i=1}^n w_i \log (n w_i \hat{p}(x_i)) - h(w)\right\}
	\end{equation*}
	where $h$ as defined in \eqref{eq:okl-dc-rep2}, being a pointwise maximum of linear functions of $w$, is a convex function of $w$. This results in the DC Algorithm from \cref{eq:owl-dc2}.
 
    Observe next that we can consider iterates $\{w_t\}_{t \geq 1}$ given by \cref{eq:owl-dc2} with $v_{t+1} = (\log p_{\theta_{t+1}}(x_1), \ldots, \log p_{\theta_{t+1}}(x_n))$, where
    $$
    \theta_{t+1} \in \argmax_{\theta \in \Theta} \sum_{i=1}^n w_{t,i} \log p_\theta(x_i).
    $$
    Indeed by Danskin's theorem \citep{nonlinearprog}, it follows that $v_{t+1} \in \partial h(w_t)$. However this shows that
    $$
    w_{t+1} = \argmin_{\substack{w \in \hat{\Delta}_n \\ \|w-o\|_1 \leq 2\epsilon}} \left\{ \sum_{i=1}^n w_i \log (n w_i \hat{p}(x_i)) - \sum_{i=1}^n w_i \log p_{\theta_{t+1}}(x_i) \right\},
    $$
    and hence $\{(\theta_t,w_t)\}_{t \geq 1}$ can be seen to coincide with the iterates in \Cref{alg:owl}. 
\end{proof}

\section{More simulation study details and results}
\label{app:more-simulations}

\new{For problems requiring a kernel density estimate, we used the Gaussian/RBF kernel: 
\[ K_h(x,y) = (2\pi h )^{-p/2} \exp(- \frac{1}{2h}\|x - y \|^2).\] 
We adaptively set the bandwidth by searching over a fixed grid and using the method's criterion to select the optimal bandwidth parameter. In the case of kernelized OWL, the criterion to be minimized was the final parameter's OKL estimator in \cref{eq:okle}. In the case of Pearson's residuals, we used the final model's empirical Hellinger distance to the density estimate:
\[ D(h, \theta) = \frac{1}{2n} \sum_{i=1}^n \left( 1 - \sqrt{ p_\theta(x_i)/\hat{p}_{n,h}(x_i)} \right)^2,\]
where $\hat{p}_{n,h}$ is the kernel density estimate 
\[ \hat{p}_{n,h}(x) = \frac{1}{n} \sum_{i=1}^n K_h(x_i, x).\]
The grid of candidate bandwidth parameters was created from the average $k$-nearest neighbor distances in the dataset for $k=5, 10, 25, 50$.

In the Gaussian setting, we were primarily concerned with the role of kernelization in the OWL procedure. As such, we set the parameter $\epsilon$ equal to the true level of corruption in the dataset.}

\Cref{fig:rand-sim} presents the simulation results for the random corruption setting. For each of the problem settings, all of the details remain exactly the same as in the max-likelihood corruption simulations, with the only change coming from the fact that the corruptions were chosen completely at random.

\begin{figure}
    \centering
    \includegraphics[width=\textwidth]{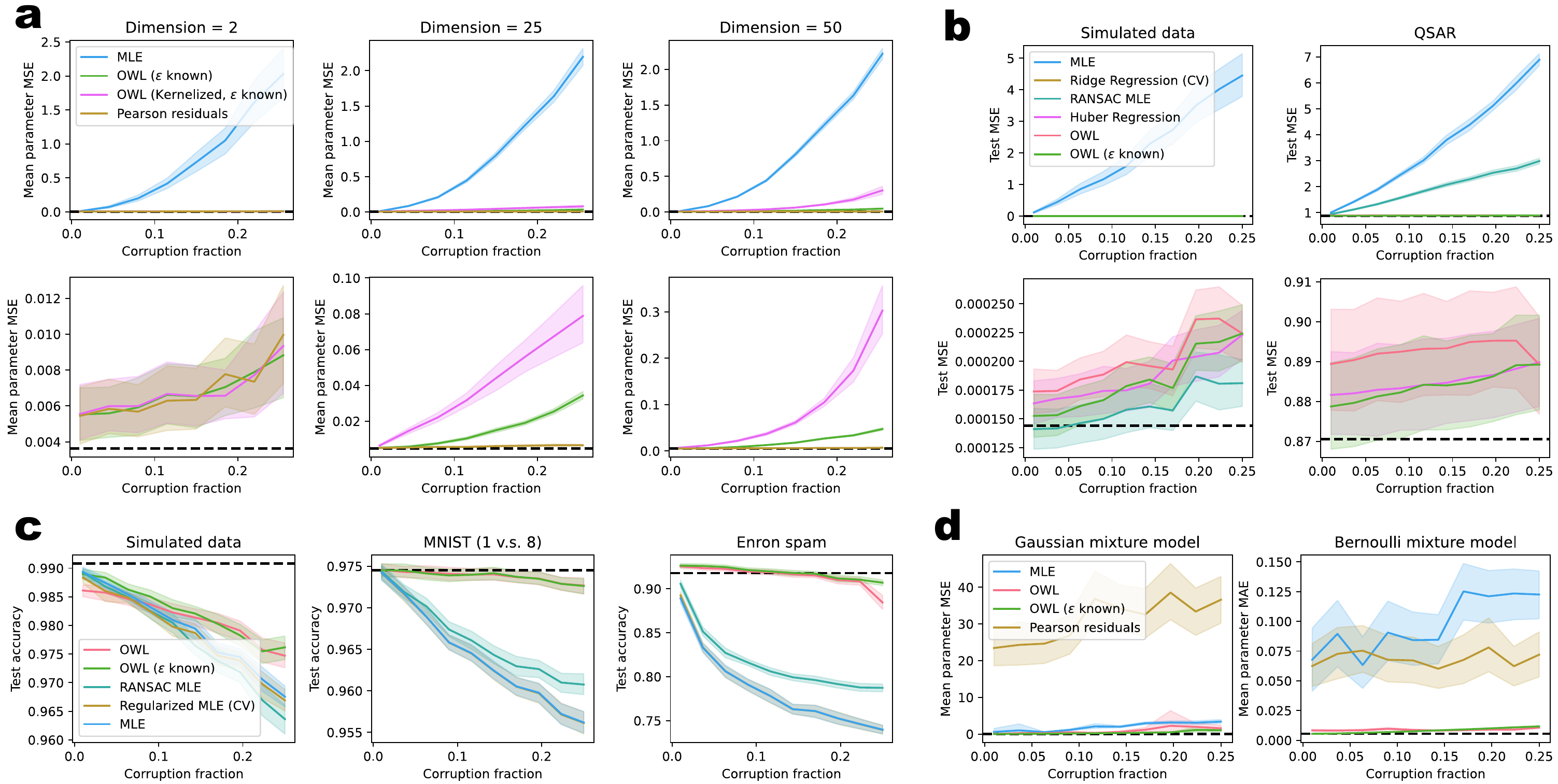}
    \caption{Simulation results for random corruptions. (a) Mean parameter reconstruction for a multivariate normal model. (b) Test MSE for linear regression. (c) Test accuracy for logistic regression. (d) Mean parameter reconstruction for mixture models. In all figures, the dashed black line denotes median performance of MLE on full uncorrupted training set, and the shaded regions denote bootstrapped 95\% confidence intervals over 50 random seeds.}
    \label{fig:rand-sim}
\end{figure}

\Cref{fig:rand-sim}a presents the random corruption results for estimating a multivariate normal distribution. In contrast with the max-likelihood corruption setting, we see that OWL without kernelization outperforms OWL with kernelization in higher dimensions. It is possible that this is due to the difficulty of density estimation in higher dimensions.

\Cref{fig:rand-sim}b presents the random corruption results for linear regression. The results here are qualitatively similar to those for the max-likelihood corruption setting, with all three robust methods performing well in the simulated data setting but with RANSAC performing notably worse with QSAR data.

\Cref{fig:rand-sim}c presents the random corruption results for logistic regression. As in the max-likelihood corruption case, we see that OWL outperforms the other methods across all three datasets. Moreover, we see that on the Enron spam dataset, OWL even outperforms the uncorrupted MLE baseline, which is entirely possible if the logistic regression model is mis-specified.

\Cref{fig:rand-sim}d presents the random corruption results for the mixture model settings. The results here are qualitatively similar to those for the max-likelihood corruption setting.

\new{\Cref{fig:kernel-choice-sim} presents the simulation results in the Gaussian setting for different choices of the kernel bandwidth parameter for kernelized OWL. We observe that while the choice of bandwidth parameter plays some role in performance, all versions of the kernelized OWL estimator significantly outperform MLE.}

\begin{figure}
    \centering
    \includegraphics[width=\textwidth]{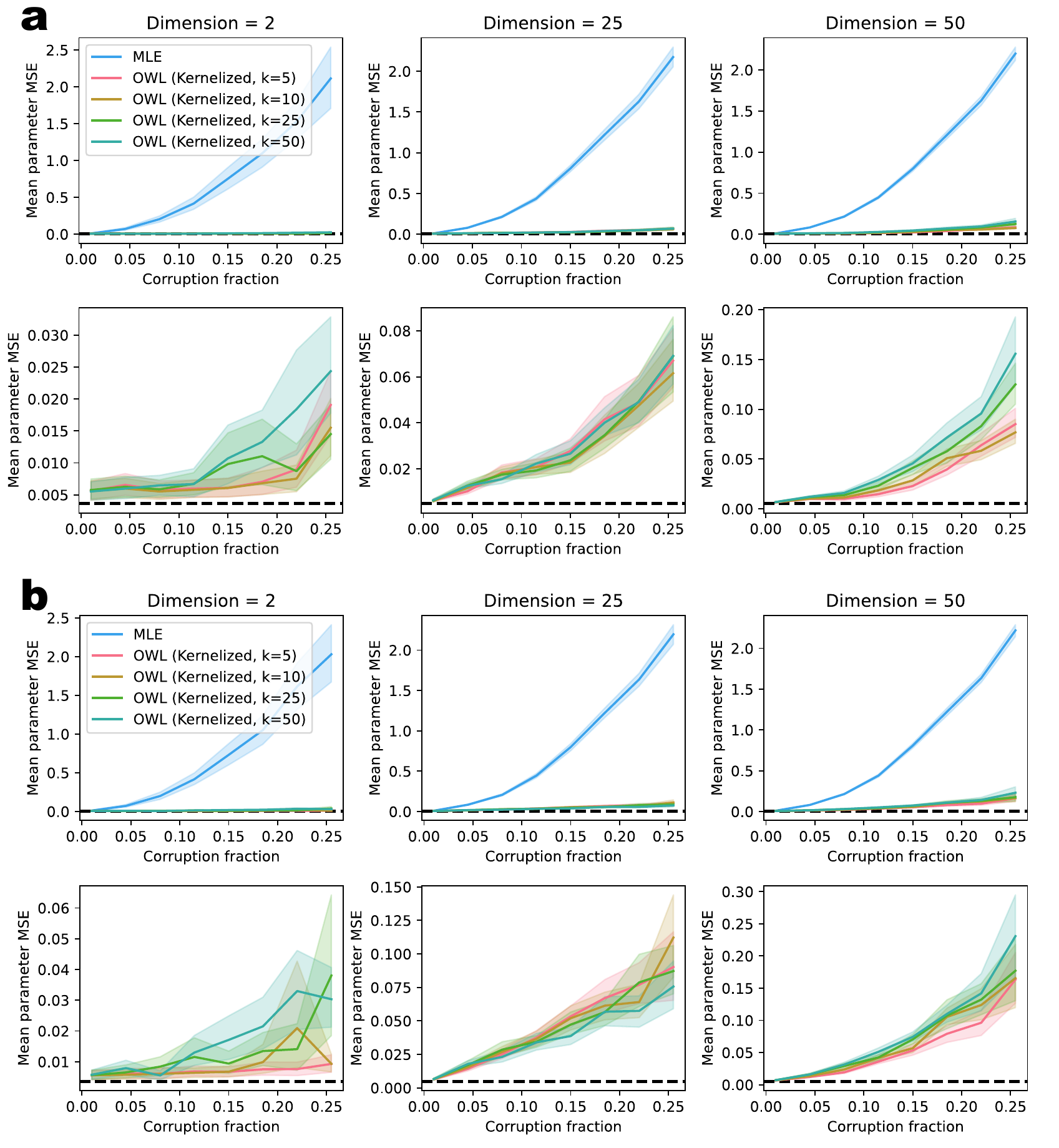}
    \caption{\new{Simulation results comparing different kernel bandwidth choices on the Gaussian simulation. The notation $k=K$ refers to choosing the average $K$-nearest neighbor distance as the kernel bandwidth parameter. (a) Max-likelihood corruptions. (b) Random corruptions. In all figures, the dashed black line denotes median performance of MLE on full uncorrupted training set, and the shaded regions denote bootstrapped 95\% confidence intervals over 50 random seeds.}}
    \label{fig:kernel-choice-sim}
\end{figure} 

\section{Application to scRNA-seq Clustering}
\label{sec:application}

In this section, we apply our OWL methodology to a single-cell RNA sequencing (scRNA-seq) clustering problem. The cell line data from \cite{li2017reference}, accessible at NCBI GEO database \citep{barrett2012ncbi} with accession GSE81861, contains single-cell RNA expression data for 630 cells from 7 cell lines across 57,241 genes. We followed the preprocessing steps of \cite{chandra2023escaping}: we dropped cells with low reads, normalized according to \cite{lun2016pooling}, and dropped uninformative genes with M3Drop~\citep{andrews2019m3drop}. After preprocessing, the dataset contains 531 cells and 7666 genes.  \Cref{tab:cline-breakdown} shows the breakdown of the remaining cells across cell lines. Finally, we used PCA to project down to 10 dimensions. We implemented OWL using a mixture of general Gaussians, $ \sum_{k=1}^K \pi_k \Ncal(\mu_k, \Sigma_k)$, using the same optimization procedure as in the clustering simulations of \Cref{sec:simul}.
\begin{table}
\centering
 \begin{tabular}{||c | c c c c c c c||} 
 \hline
 Cell line & A549 & GM12878 & H1 & H1437 & HCT116 & IMR90 & K562 \\
 \hline
 Counts & 74 & 126 & 164 & 47 & 51 & 23 & 46 \\ 
 \hline
 \end{tabular}
 \caption{Breakdown of samples in GSE81861 dataset by cell line.}
 \label{tab:cline-breakdown}
\end{table}

\subsection{Cluster recovery with OWL}
We measured the ability of OWL to recover the ground-truth clustering of samples. For baseline methods, we compared against maximum likelihood estimation with the same model class and K-means. As a metric of cluster recovery, we measured the adjusted Rand index (ARI)~\citep{rand1971objective,hubert1985comparing}. In all our comparisons, we fixed the number of clusters for all methods to be 7, the number of ground truth cell lines.

The left panel of \Cref{fig:ari-comparison} shows the ARI for OWL over a range of values for the $\ell_1$ radius parameter $\epsilon$, where we also display the performance of MLE and K-means for comparison. We see that OWL performs best when $\epsilon$ takes on values between $0.25$ and $0.45$, but generally has reasonable performance when $\epsilon$ is not too large. Moreover, we see that performance of OWL varies smoothly as a function of $\epsilon$, which may reflect the continuity of the OKL function with respect to $\epsilon$ predicted by our theory.

\Cref{fig:umap-baseline-comparison} shows Uniform Manifold Approximation and Projection (UMAP) visualizations of the dataset clustered under the various methods (for one arbitrary run). We see that of all the methods, K-means performs worst by a significant margin. The improved performance of OWL (with $\epsilon=0.25$) over MLE can be mostly attributed to the better resolution the boundary between the K562 and GM12878. However, all methods struggle to identify the IMR90 cell lines as a cluster distinct from K562.

\begin{figure}
    \centering
    \includegraphics[width=0.95\textwidth]{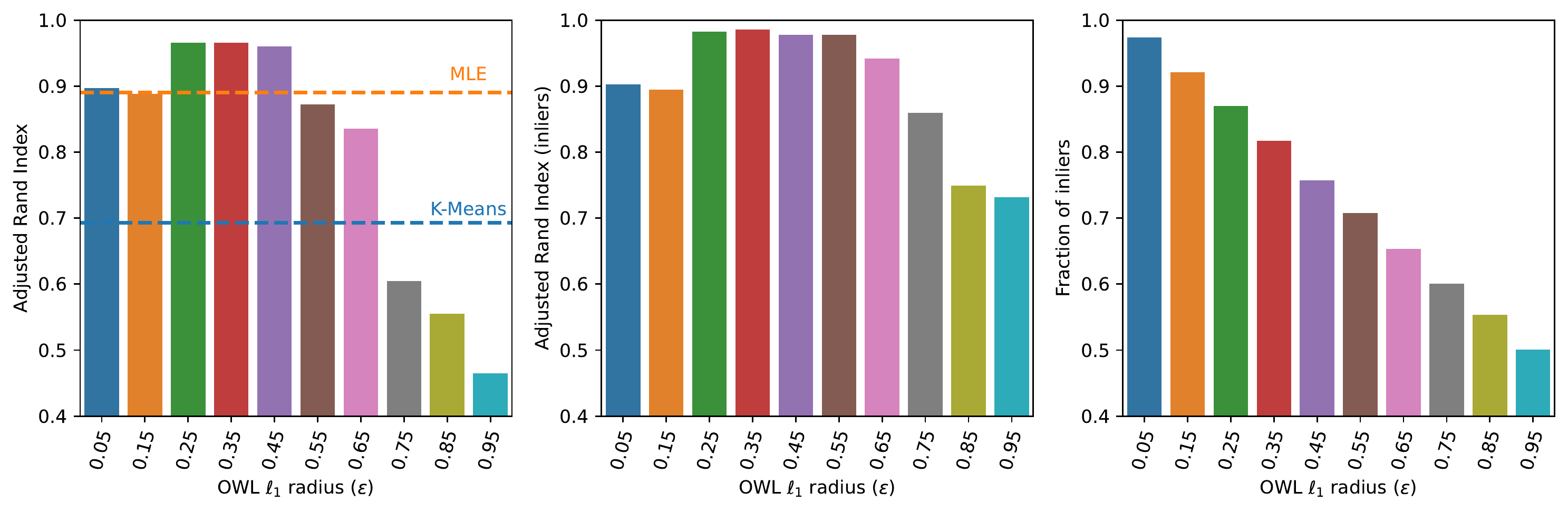}
    \caption{Comparison of clustering methods. \emph{Left}: Adjusted Rand index (ARI) over the entire dataset for each of the methods. \emph{Middle}: ARI of inliers for the OWL methods. \emph{Right}: Fraction of data points classified as inliers for the OWL methods.}
    \label{fig:ari-comparison}
\end{figure}

\begin{figure}
    \centering
    \includegraphics[width=0.8\textwidth]{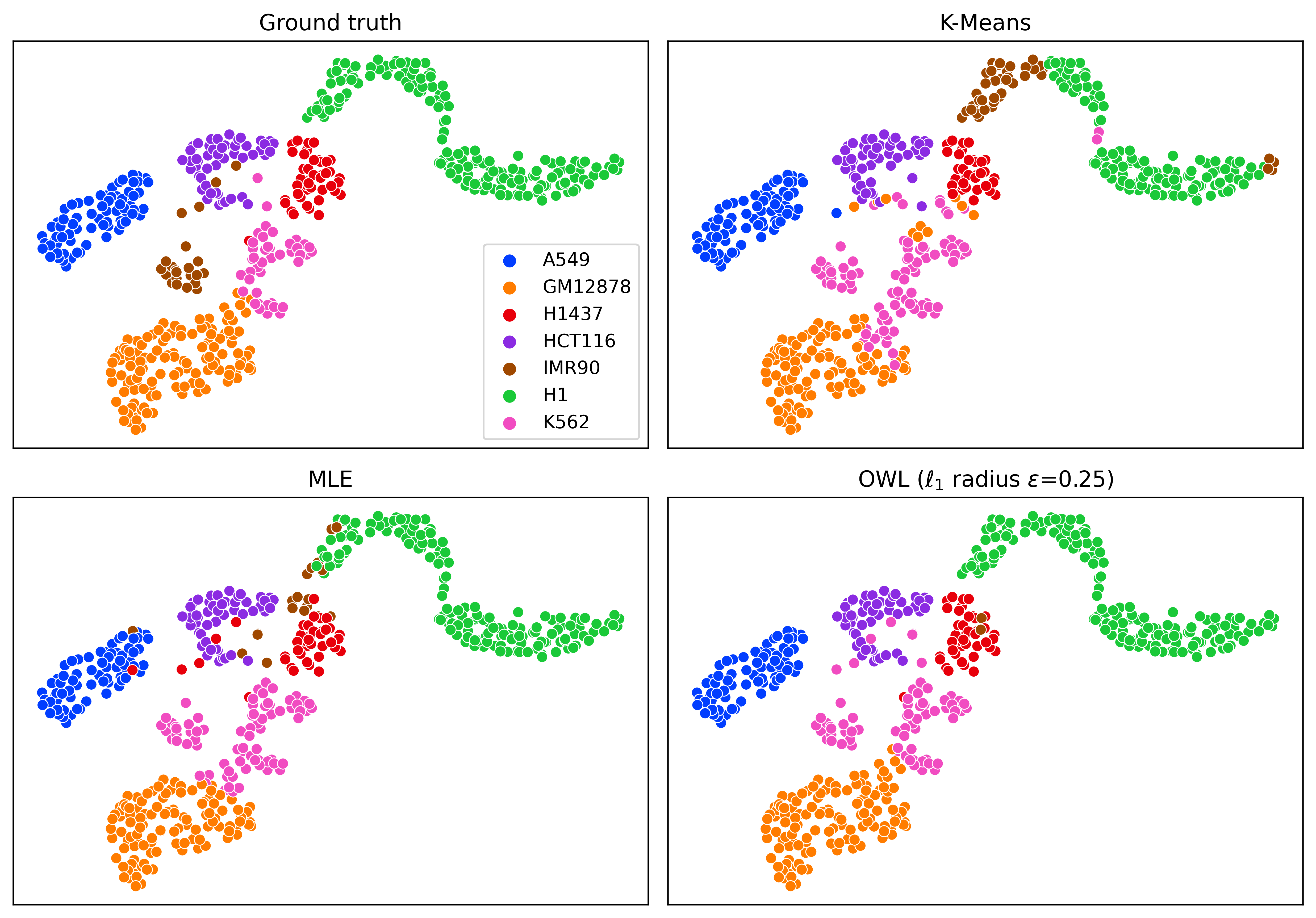}
    \caption{UMAP plots of the GSE81861 dataset under the considered clustering algorithms. Top left displays the ground truth cell lines. For the other panels, colors were selected by maximizing agreement with the ground truth clustering.}
    \label{fig:umap-baseline-comparison}
\end{figure}

\subsection{Exploratory analysis with OWL}

In some settings, it is desirable to segment a dataset into those data points that are well-described by a model in the class (so-called \emph{inliers}) and those that do not conform well to the model class (\emph{outliers}). One interpretation of the weights that are learned by the OWL procedure is that, subject to the constraint that they are close in TV distance to the empirical distribution, they represent the most optimistic reweighting of the data relative to the model class. Thus, one might suspect that data points with higher weights are inliers and those with lower weights are outliers. Here, we explore inlier/outlier detection with OWL weights by classifying all data points with weights less than $1/n$ (the average value) as outliers, and the remainder as inliers.

The middle panel of \Cref{fig:ari-comparison} shows the ARI of the OWL procedure when we restrict to the detected inliers. We observe that for all values of $\epsilon$, the ARI is no lower on the selected inliers than on the whole dataset, and in some cases is significantly higher. This suggests that the OWL procedure identifies a `core' set of points that are both well-described by a mixture of Gaussians as well as aligned with the ground truth clustering. The right panel of \Cref{fig:ari-comparison} shows the fraction of data points that are classified as inliers. Although it is theoretically possible for the OWL weights to classify anywhere from 1 to $n-1$ points as outliers for any value of $\epsilon$, we see that the fraction of outliers is relatively small for low values of $\epsilon$ and only increases gradually as $\epsilon$ increases.

\begin{figure}
    \centering
    \includegraphics[width=0.9\textwidth]{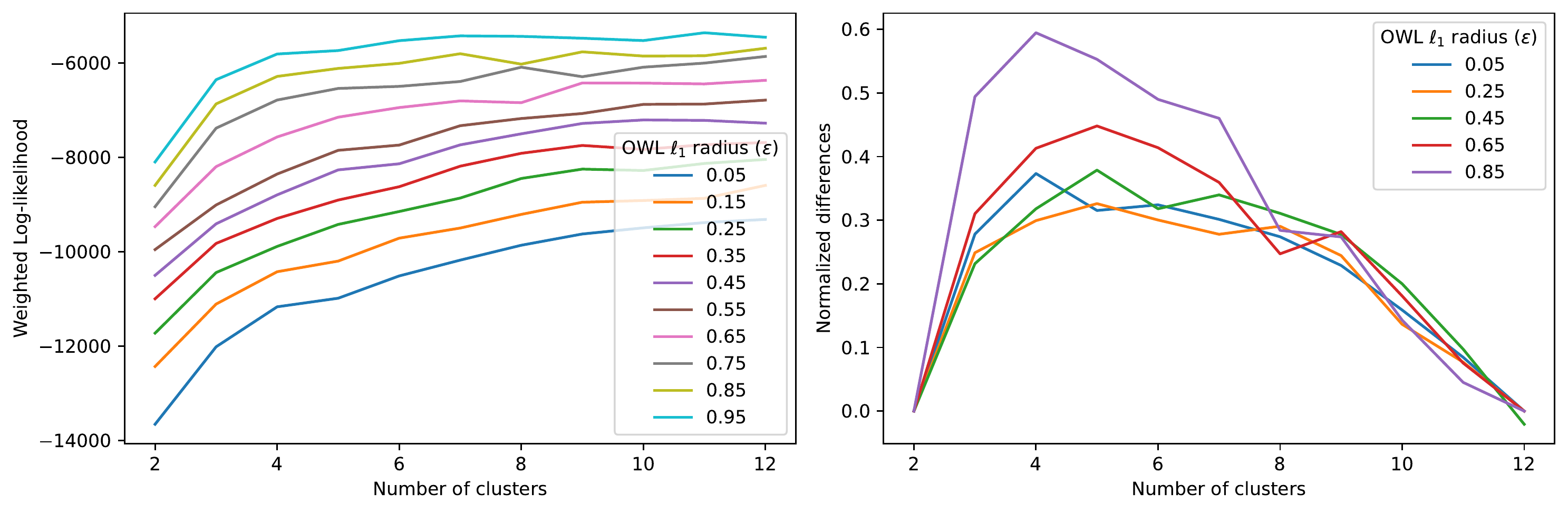}
    \caption{\emph{Left}: Weighted log-likehood of the data for various settings of $\epsilon$ and the number of clusters. \emph{Right}: Normalized difference graph of the weighted log-likehood function for select values of $\epsilon$. The kneedle algorithm chooses the value with the largest corresponding normalized difference.}
    \label{fig:k-selection}
\end{figure}

\begin{table}

\centering
 \begin{tabular}{||c | c c c c c c c c c c ||} 
 \hline
  OWL $\ell_1$ radius ($\epsilon$) & 0.05 & 0.15 &  0.25 & 0.35 & 0.45 & 0.55 & 0.65 & 0.75 & 0.85 & 0.95 \\
 \hline
 Selected $K$ & 6 & 6 & 8 & 7 & 7 & 7 & 5 & 4 & 4 & 4 \\
 \hline
 \end{tabular}
  \caption{Number of clusters chosen by the kneedle method as a function of the $\ell_1$ radius $\epsilon$.}
  \label{tab:k-selection}
\end{table}

In many settings, the number of ground truth clusters are not known a priori. A common way to deal with this problem is to plot a metric such as sum-of-squares errors or log-likelihood and look for `elbows' or `knees' in the graph where there are diminishing returns for increasing model capacity. Here, we apply the `kneedle' algorithm~\citep{satopaa2011finding} to the weighted log-likehood produced by the OWL procedure. The kneedle algorithm computes the normalized differences of a given function and selects the value that maximizes the corresponding normalized differences. \Cref{fig:k-selection} shows both the weighted log-likelihoods as well as a subset of the normalized difference graphs. \Cref{tab:k-selection} shows the selected numbers of clusters for various values of $\epsilon$. We see that for relatively small values of $\epsilon$, this results in number of clusters that is close to the ground truth. While for larger values of $\epsilon$, this procedure underestimates the number of clusters in the data. This agrees with the observation in the right panel of  \Cref{fig:ari-comparison} that larger values of $\epsilon$ result in fewer points being identified as inliers, and thus fewer clusters are needed to describe those points.
 
\section{More details of the micro-credit study}
\label{app:micro-credit}
This section contains additional details of our analysis in \Cref{sec:micro-credit}. 
\subsection{Reproducing the brittleness of MLE}
\label{sec:brittleness}

To reproduce the brittleness in estimating the AIT on household profits demonstrated in \cite{broderick2020automatic}, we first obtained the profit data from \cite{angelucci2015microcredit} as imputed and scaled in \cite{data:meager2019}. The MLE estimate of $\beta_1 = -4.55$ USD PPP per fortnight (standard error [s.e.] of 5.88), changes to $\beta_1 = 0.4$ USD PPP per fortnight (s.e. 3.19) if we remove a single household identified by the \texttt{zaminfluence} R package \citep{broderick2020automatic}. Moreover, by removing 14 further observations which were identified by the \texttt{zaminfluence} package, we observe that the non-significant value of the MLE estimate can be changed to a significant value of $\beta_1 = -6.01$ USD PPP (s.e. 2.57). As seen in a scatter-plot summarizing the data (\Cref{fig:micro-scatter}), this brittleness of the MLE is likely due to a small fraction of households with outlying profit values.

\begin{figure}[h]
	\centering
	\includegraphics{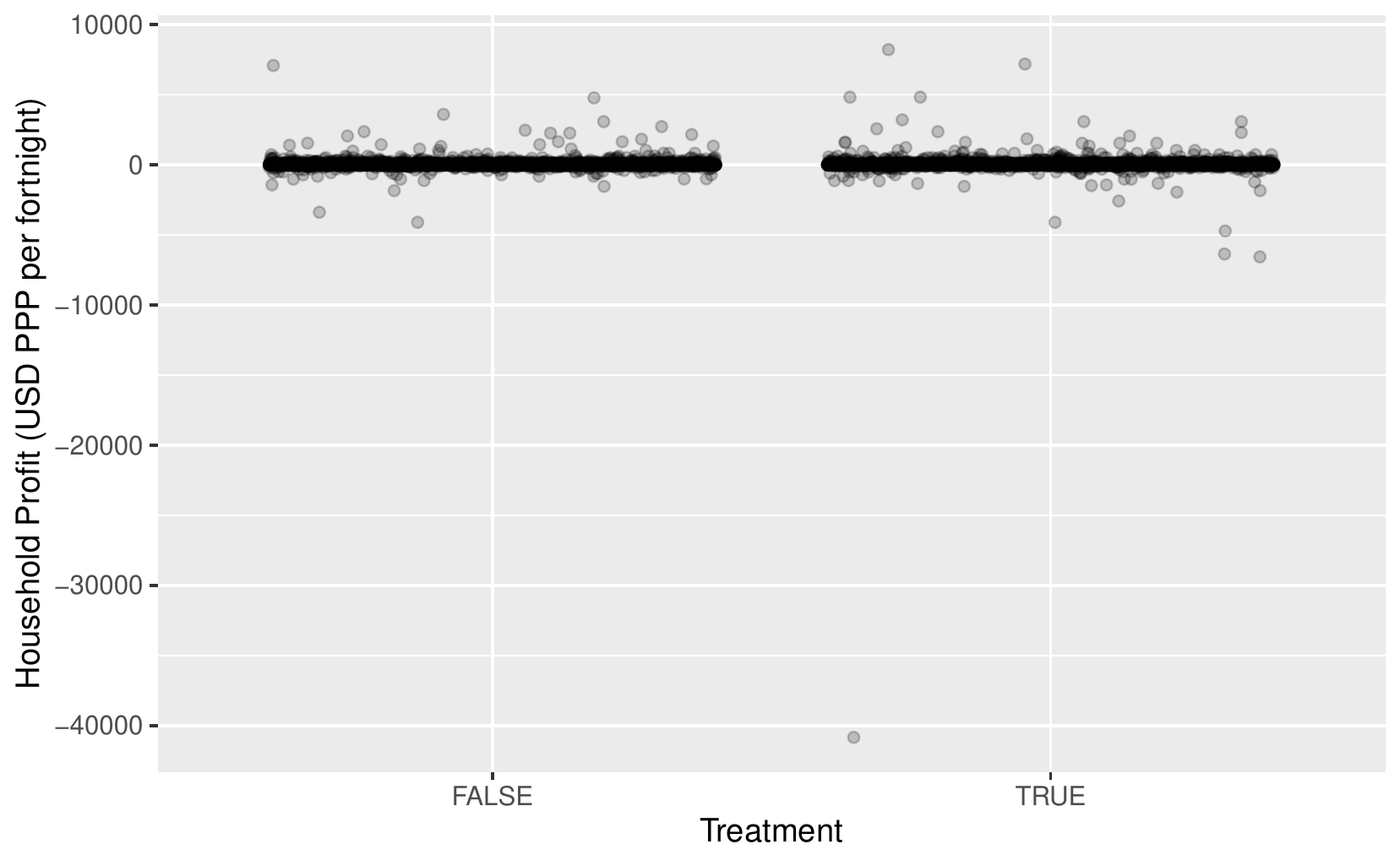}
	\caption{Scatter-plot of the household profit values across treated and non-treated households. Even after removing the  household with the extreme profit of -40000 USD PPP, there are still households with extreme profit values that cause brittleness in estimating the average treatment effect.}
	\label{fig:micro-scatter}
\end{figure}

\subsection{OWL uncertainty using outlier-stratified bootstrap}
\label{sec:os-bootstrap}

To quantify uncertainty in the AIT estimates obtained by OWL, we reran our analysis in \Cref{sec:micro-credit} on $m=50$ independently bootstrapped data sets of size $n$ each. Since we wanted to retain a small fraction of outlying observations in each data set, we used an \emph{outlier-stratified} (OS) sampling strategy. Namely, in each iteration, the  new data set was obtained by combining a bootstrap sample of the (roughly $1\%$) households that were down-weighted by the OWL procedure at $\epsilon_0$ and a bootstrap sample from the remaining households that were not down-weighted. 

The resulting 90\% OS-bootstrap confidence bands for estimates of AIT and minimum-OKL as a function of $\epsilon$ can be found in \Cref{fig:micro-log-scale-plots} ($x$-axis scaled to emphasize the uncertainty for small values of $\epsilon$) and \Cref{fig:micro-okl-plot}, respectively.

\subsection{Additional plots}

\Cref{fig:micro-scatter} shows a scatter-plot for the data and the presence of outliers.\\
\Cref{fig:micro-outlier-hist-plot} shows the distribution of profit values for points that were declared as outliers by the OWL procedure at the  parameter value $\epsilon_0$.

\begin{figure}[h]
    \centering
    \includegraphics{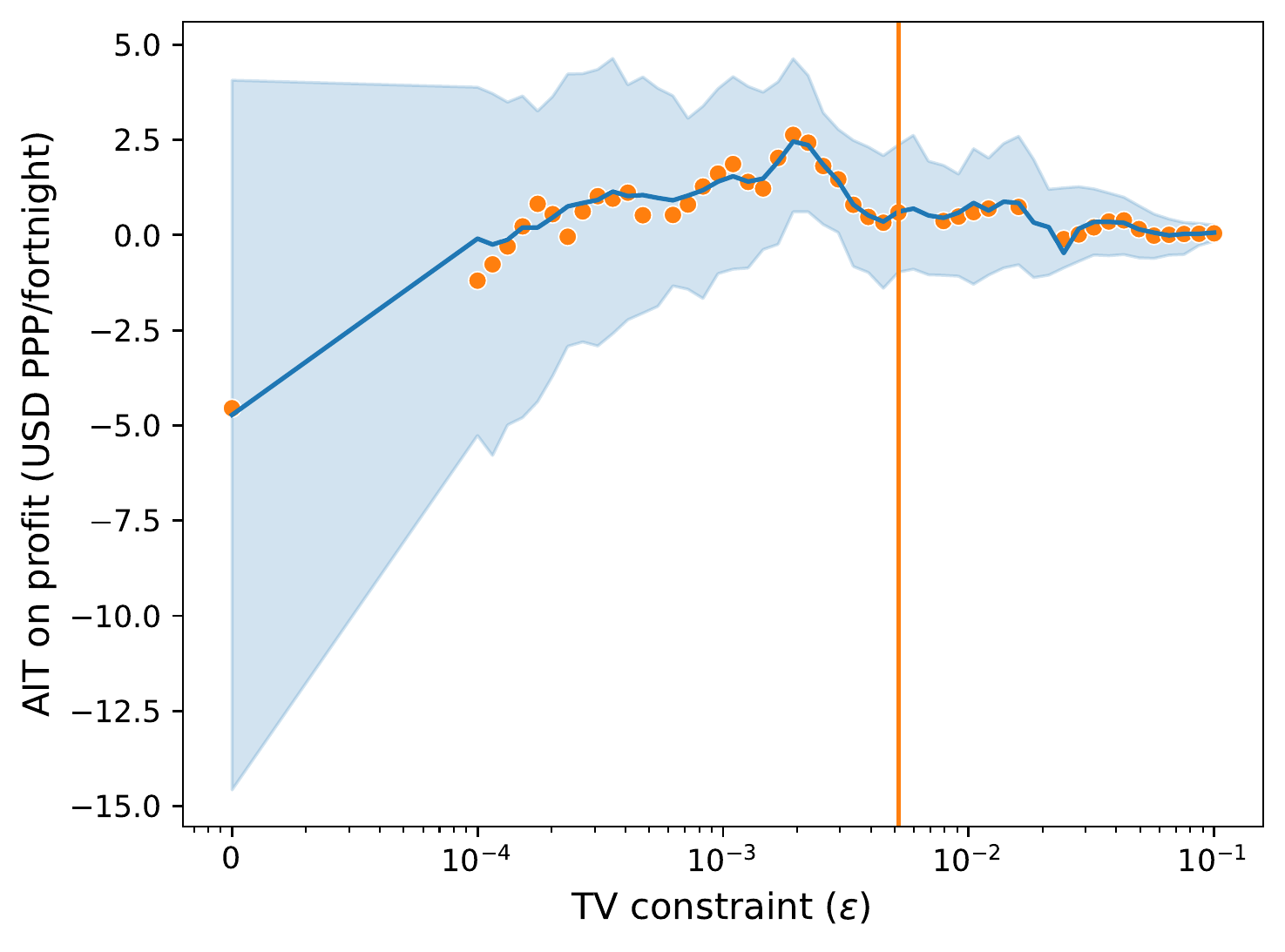}
    \caption{The plot from the left panel of \Cref{fig:micro_fig}, plotted on a re-scaled $x$-axis to emphasize the uncertainty in the AIT estimates for small values of $\epsilon$. The confidence bands become narrow roughly at the tuned value of $\epsilon_0 = 0.005$ (\Cref{sec:tune-epsilon}), suggesting that the outliers that cause brittleness may  have been down-weighted by OWL for $\epsilon = \epsilon_0$.}
    \label{fig:micro-log-scale-plots}
\end{figure}

\begin{figure}[h]
    \centering
    \includegraphics{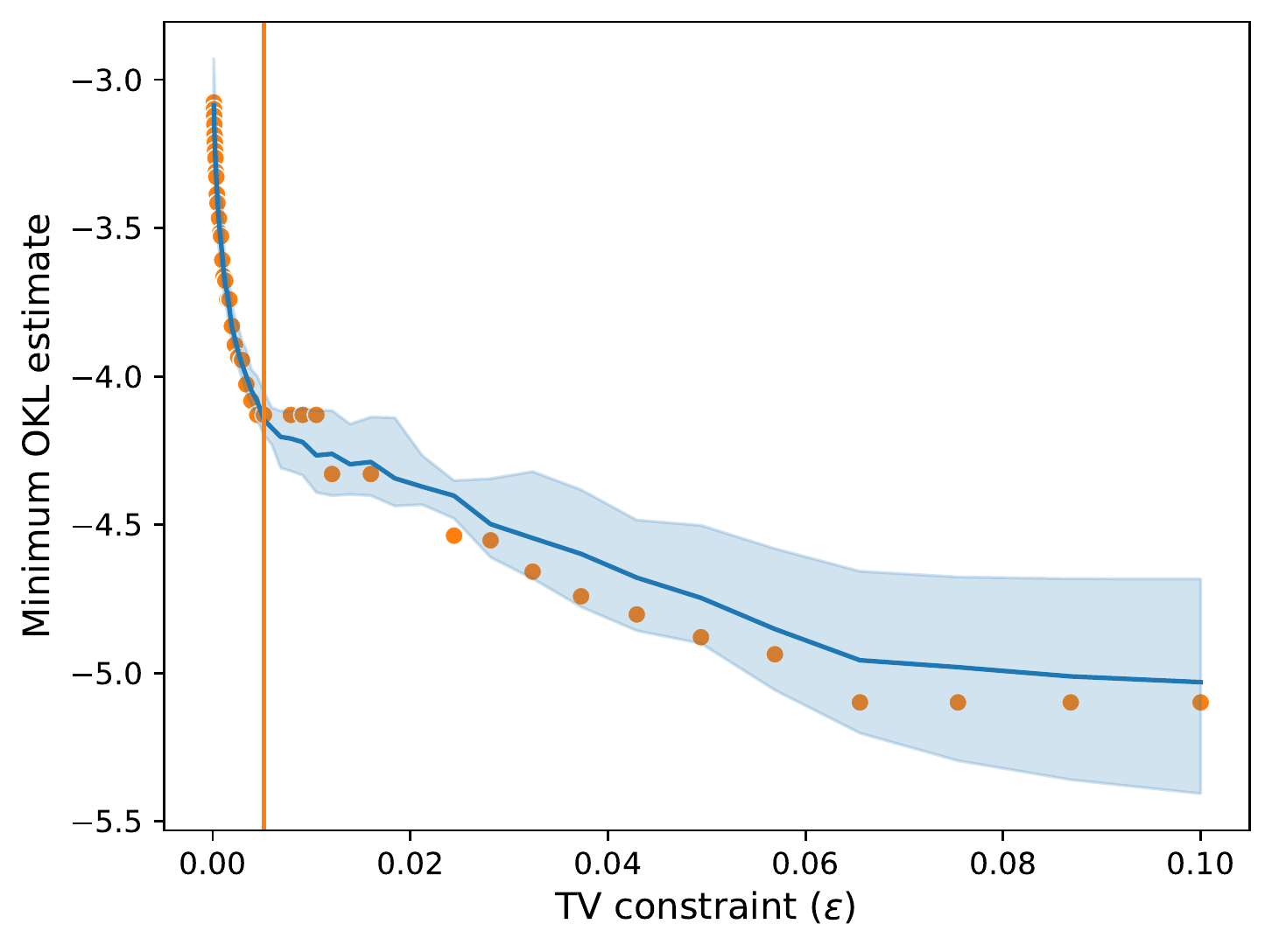}
    \caption{The minimum OKL estimate (i.e. $\hat{R}(\epsilon) = \min_{\theta \in \Theta} \hat{I}_\epsilon(\theta)$) versus $\epsilon$ plot for the micro-credit example. Using the notion of curvature in \Cref{sec:tune-epsilon}, the value $\epsilon_0 = 0.005$ was identified as the point at which this graph has its prominent kink. The $90\%$ confidence bands under $m=50$ OS-bootstrap iterations are also shown.} 
    \label{fig:micro-okl-plot}
\end{figure}

\begin{figure}[h]
    \centering
    \includegraphics{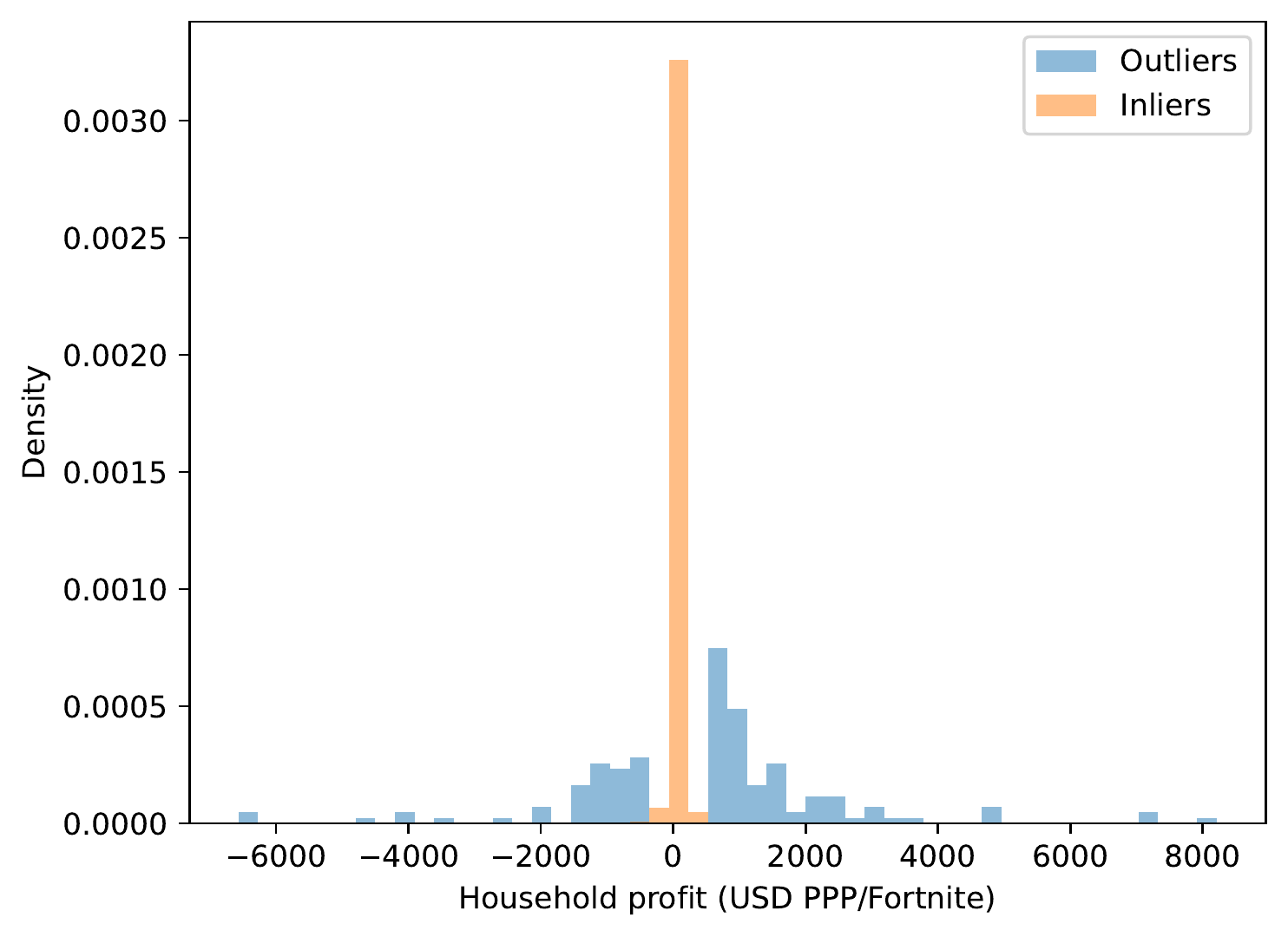}
    \caption{\new{The profit distribution for households that were declared to be inliers ($\{ i : w_i < 0.05 \}$) versus outliers ($\{ i: w_i \leq 0.05 \}$) by the OWL procedure at parameter $\epsilon_0$. For clarity, we omitted an outlying household with a profit value of less that $-40K$ USD PPP.}} 
    \label{fig:micro-outlier-hist-plot}
\end{figure}

    \clearpage
    \putbib[refs]
\end{bibunit}

\end{document}